\documentclass[pra,twocolumn,a4paper,showpacs,aps,10pt,nofootinbib,allowtoday]{revtex4-1}
\usepackage{graphicx}
\usepackage{soul}
\usepackage{hyperref}
\usepackage[section]{placeins}
\usepackage{amsmath}
\usepackage{amsthm}
\usepackage{amssymb}
\usepackage{dsfont}
\usepackage{bbold}
\usepackage{float}
\usepackage{color}
\newtheorem{theorem}{Theorem}
\newtheorem{lemma}{Lemma}

\newcommand{\sgn}{\mathrm{sgn}}
\newcommand{\Imag}{\mathrm{Im}}
\newcommand{\Real}{\mathrm{Re}}

\newcommand{\ack}{\subsection*{\normalsize \textbf{Acknowledgment}}}

\DeclareFontFamily{U}{mathx}{\hyphenchar\font45}
\DeclareFontShape{U}{mathx}{m}{n}{
      <5> <6> <7> <8> <9> <10>
      <10.95> <12> <14.4> <17.28> <20.74> <24.88>
      mathx10
      }{}
\DeclareSymbolFont{mathx}{U}{mathx}{m}{n}
\DeclareFontSubstitution{U}{mathx}{m}{n}
\DeclareMathSymbol{\bigplus}{1}{mathx}{"90}
\DeclareMathSymbol{\bigtimes}{1}{mathx}{"91}

\newcommand{\bit}{\begin{itemize}}
	\newcommand{\eit}{\end{itemize}\par\noindent}
\newcommand{\ben}{\begin{enumerate}}
	\newcommand{\een}{\end{enumerate}\par\noindent}
\newcommand{\beq}{\begin{equation}}
\newcommand{\eeq}{\end{equation}\par\noindent}
\newcommand{\beqa}{\begin{eqnarray*}}
	\newcommand{\eeqa}{\end{eqnarray*}\par\noindent}
\newcommand{\beqn}{\begin{eqnarray}}
\newcommand{\eeqn}{\end{eqnarray}\par\noindent}

\usepackage{hyperref}
\usepackage{amsmath}
\usepackage{amsthm}
\usepackage{graphicx}
\usepackage{braket}
\usepackage[makeroom]{cancel}

\newtheorem{definition}{Definition}
\newtheorem{proposition}{Proposition}
\newtheorem{corollary}{Corollary}
\newtheorem{example}{Example}
\newtheorem{remark}{Remark}
\newtheorem{conjecture}{Conjecture}
\begin{document}
	
\title{Joint measurability structures realizable with qubit measurements:\\ incompatibility via marginal surgery}
	
\author{Nikola Andrejic}
	\email{nikola.andrejic@pmf.edu.rs}
	\affiliation{University of Ni\v s, Faculty of Sciences and Mathematics,  Vi\v segradska 33, 18000 Ni\v s, Serbia}

\author{Ravi Kunjwal}
	\email{rkunjwal@ulb.ac.be}
	\affiliation{Perimeter Institute for Theoretical Physics, 31 Caroline Street North, Waterloo, Ontario Canada N2L 2Y5}
	\affiliation{Centre for Quantum Information and Communication, Ecole polytechnique de Bruxelles,
		CP 165, Universit\'e libre de Bruxelles, 1050 Brussels, Belgium}
	
	\date{\today}
	
	\begin{abstract}
Measurements in quantum theory exhibit incompatibility, i.e., they can fail to be jointly measurable. An intuitive way to represent the (in)compatibility relations among a set of measurements is via a hypergraph representing their joint measurability structure: its vertices represent measurements and its hyperedges represent (all and only) subsets of compatible measurements. Projective measurements in quantum theory realize (all and only) joint measurability structures that are graphs. On the other hand, general measurements represented by positive operator-valued measures (POVMs) can realize arbitrary joint measurability structures. Here we explore the scope of joint measurability structures realizable with qubit POVMs. We develop a technique that we term {\em marginal surgery} to obtain nontrivial joint measurability structures starting from a set of compatible measurements. We show explicit examples of marginal surgery on a special set of qubit POVMs to construct joint measurability structures such as  $N$-cycle and $N$-Specker scenarios for any integer $N\geq 3$. We also show the realizability of various joint measurability structures with $N\in\{4,5,6\}$ vertices. In particular, we show that {\em all} possible joint measurability structures with $N=4$ vertices are realizable. We conjecture that {\em all} joint measurability structures are realizable with qubit POVMs. This contrasts with the unbounded dimension required in \href{https://doi.org/10.1103/PhysRevA.89.052126}{R.~Kunjwal {\em et al.}, Phys.~Rev.~A 89, 052126 (2014)}. Our results also render this previous construction maximally efficient in terms of the required Hilbert space dimension. We also obtain a sufficient condition for the joint measurability of any set of binary qubit POVMs which powers many of our results and should be of independent interest.
\end{abstract}
	
\maketitle
\tableofcontents

\section{Introduction}
A fundamental sense in which quantum theory departs from classical physics is the existence of incompatible measurements in the former. That is, it is impossible to simulate certain sets of quantum measurements by coarse-graining a single quantum measurement \cite{HRS08}. This incompatibility is crucial to the demonstration of nonclassicality in quantum theory, e.g., both Bell inequality violations \cite{Bell64, Bell66, BCP14} and Kochen-Specker (KS) contextuality \cite{KS67} are impossible in the absence of incompatible measurements. The joint measurability (or compatibility) of a set of projective measurements is a binary property, characterized entirely by their pairwise commutativity. On the other hand, the joint measurability of general quantum measurements represented by positive operator-valued measures (POVMs) is, in general, not characterized by pairwise commutativity. It is possible to have nonprojective measurements that are noncommuting and that nonetheless admit a joint POVM that can be
coarse-grained to obtain their statistics. Recent years have seen a  steady increase in research activity geared towards a better understanding of the joint measurability of POVMs, its connection to steering, state discrimination, Bell nonlocality, as well as the general area of determining conditions for joint measurability of sets of POVMs \cite{WPF09, QVB14, UMG14, KHF14, HMZ16, ULMH16, CS16, CHT19, SSC19, UKS19, OB19}.

The joint measurability relations among a set of POVMs can be represented by a hypergraph that we will refer to as the {\em joint measurability structure} of the set, following Ref.~\cite{KHF14}.\footnote{We will recall the definition of joint measurability structure more formally later in this paper but for now it suffices to note that it
specifies whether any subset of POVMs in a given set is compatible or incompatible. Note also that we use the terms `joint measurability' and `compatibility' interchangeably in this paper and `incompatibility' denotes the lack of joint measurability.} This joint measurability structure is the collection of (all and only) jointly measurable subsets of the given set of POVMs. Since any set of projective measurements is jointly measurable if and only if they commute pairwise \cite{HRS08}, its joint measurability structure corresponds to a graph,
i.e., a collection of jointly measurable subsets, each of size two. Ref.~\cite{HFR14} showed that all joint measurability structures corresponding to graphs can be realized by projective measurements. Later, Ref.~\cite{KHF14} realized arbitrary joint measurability structures (going beyond the ones that are graphs) using POVMs. 

In the sphere of contextuality research, the use of nonprojective POVMs in proofs of KS-contextuality has been controversial \cite{Spekkens14}. However, within the general framework of contextuality \`a la Spekkens \cite{Spekkens05}, arbitrary POVMs can be accommodated. Indeed, in an important application of this framework to POVMs \cite{MPK16}, the fact that the same POVM can be achieved via different measurement procedures -- namely, a uniform mixture of three binary outcome measurements and a fair coin flip, each realizing the POVM $\{I/2,I/2\}$ -- was used to obtain a noise-robust noncontextuality inequality. When this framework is applied to the case of KS-type experiments, the different measurement procedures used to implement a POVM correspond to different joint measurements of the POVM with other POVMs. In this situation, the nonclassicality of POVMs can be revealed using noise-robust noncontextuality inequalities inspired by proofs of the KS theorem \cite{KS15,KS18,Kunjwal19, Kunjwal20}. The realizability of arbitrary joint measurability structures by POVMs \cite{KHF14} indicates that there is a whole zoo of joint measurability structures that remains to be explored from the point of view of contextuality in KS-type experiments. These structures were never the object of study in treatments of KS-contextuality because they admit no realizations with projective measurements. For such structures, only nonprojective POVMs can provide any evidence of nonclassicality. The simplest example of such a joint measurability structure is Specker's scenario (Fig.~\ref{speckerscenario}), which was shown to be realizable on a qubit \cite{LSW11} even before the quantum realizability of arbitrary joint measurability structures with POVMs was shown \cite{KHF14}. In this scenario, at least under the assumption that the operational theory is quantum theory, one can witness contextuality \cite{KG14, ZCL17}. 

It remains to extend such a demonstration
of contextuality \cite{KG14} to arbitrary joint measurability structures in general operational theories \cite{Kunjwal14, Kunjwal16, Kunjwal17}. Before such a project can be undertaken, however, one needs a better understanding of the scope of joint measurability structures that can be realized by quantum systems of limited Hilbert space dimension, particularly if one is keen to build experimental tests and quantum information protocols based on such joint measurability structures. Such an understanding can also potentially aid the elucidation of facts about joint measurability that are particular to quantum theory in the broader landscape of general probabilistic theories \cite{Barrett07,BGG13, SB14, FHL17, GKS18}. Already, in Ref.~\cite{GKS18}, it was shown that considering the joint measurability structure corresponding to a complete graph on four vertices is enough to separate almost quantum correlations (generalized to include single laboratory situations) \cite{NGH15, AFL15} from those correlations that can be realized with projective measurements in quantum theory \cite{AFL15}. This result was taken as a failure of Specker's principle \cite{Cabello12, GKS18} -- that any set of pairwise compatible measurements is globally compatible -- for any measurements underlying almost quantum correlations. 

However, Specker's principle also fails to hold within quantum theory if the measurements are not projective, so one can always ask: what if quantum measurements are allowed to be arbitary POVMs rather than just projective measurements? Is there still, in any sense, a qualitative difference between measurements in an almost quantum theory and those in quantum theory if we allow arbitrary POVMs in the latter? Presumably, an answer to this question would require a framework for addressing nonclassicality of quantum correlations arising from arbitrary POVMs and, hence, the ability to deal with arbitary joint measurability structures. Such a framework would be a natural extension of the ones proposed in Refs.~\cite{Kunjwal19, Kunjwal20, KS15, KS18}. 

Coming back to the construction of Ref.~\cite{KHF14}, note that it requires a steadily growing Hilbert space dimension to realize joint measurability structures corresponding to ever larger sets of measurements, rendering it quite inefficient in this sense. This raises the natural question of whether it is possible to realize joint measurability structures for arbitrarily large sets of measurements using the smallest possible Hilbert space dimension, i.e., using qubit POVMs. 

All of the considerations above motivate the present study as a first step towards addressing general features of joint measurability in quantum theory for systems of limited Hilbert space dimension. A summary of our results follows. 

We introduce a technique that we term {\em marginal surgery} for realizing new joint measurability structures starting from a joint POVM of a set of compatible POVMs. We use marginal surgery to show that two families of joint measurability structures can be realized with qubit POVMs: $N$-cycle scenarios and $N$-Specker scenarios for any finite number, $N\geq3$, of measurements. We also show the realizability of many joint measurability structures with $N$ vertices, $N\in\{4,5,6\}$. In particular, for $N=4$, we show that {\em all} conceivable joint measurability structures can be realized with qubit POVMs. This is also obviously true for $N=3$ vertices, where the realizability of $3$-Specker scenario is enough to make this claim, the realizability of other scenarios requiring only incompatibility of pairs of POVMs. Motivated by the fact that all conceivable joint measurability structures for $N\in\{3,4\}$ vertices are thus realizable with qubit POVMs, and the fact that structures such as $N$-Specker and $N$-cycle with arbitrary large $N$ are also realizable in this way, we conjecture that arbitrary joint measurability structures can be realized with qubit POVMs. Although we do not have a concrete proposal, we expect that a general recipe for realizing arbitrary joint measurability structures with qubit POVMs could potentially be obtained by some clever application of the method of marginal surgery. On the other hand, we also mention a potential counter-example to the conjecture, namely, a joint measurability structure which can perhaps be shown to be not realizable with qubit POVMs. Along the way, we obtain a significant result which should be of independent interest: a sufficient condition for the joint measurability of any set of binary qubit POVMs which goes some way towards addressing the general problem of obtaining necessary and sufficient conditions for the joint measurability of arbitrary sets of POVMs.

We now outline the structure of the paper:
\begin{itemize}
	\item In Section \ref{sec2}, we start with a review of definitions and some known facts about joint measurability of POVMs that will be used in this paper.  Sections \ref{subsec2_1}, and \ref{subsec2_3} deal with defining the notion of geometric equivalence between sets of qubit POVMs and its application to a special type of these sets, namely, the planar symmetric POVMs. These definitions are then used in deriving later results.
	
	\item   In Section \ref{sec3}, we introduce marginal surgery and apply it, firstly, to the case of pairs of POVMs out of a set of qubit POVMs (Sec.~\ref{subsec3_1}), realizing $N$-cycle scenarios as an example, and then to the case of arbitrary subsets of a set of qubit POVMs (Sec.~\ref{subsec3_2}), realizing $N$-Specker scenarios as an example. 
	
	\item In Section \ref{sec4}, we obtain a sufficient condition for joint measurability of a set of coplanar and unbiased binary qubit POVMs with the same sharpness parameter (Theorem \ref{counbiqu}). We then generalize this obtain a sufficient condition for joint measurability of any set of binary qubit POVMs (Theorem \ref{suffNnsp}).
	
	\item In Section \ref{sec5}, we study the qubit realizability of various joint measurability structures ``in between" $N$-cycle and $N$-Specker for $N=4,5,6$ vertices (cf.~Sec.~\ref{subsec5_1}). In Sections \ref{subsec5_2} and \ref{subsec5_3}, we show the qubit realizability of all joint measurability structures with $N=4$ vertices.
	
	\item We conclude with discussion and some open questions in Section \ref{sec6}.
\end{itemize}

\section{A review of joint measurability of binary qubit measurements}\label{sec2}
Here we review the joint measurability conditions -- necessary and/or sufficient -- for binary qubit POVMs that are known from previous work in this area. We also make some general observations about joint measurability that will be useful in later sections. We begin with basic definitions below.

\subsection{POVMs and their joint measurability}
\begin{definition}[POVMs and binary qubit POVMs]
Given a non-empty set $\mathcal{O}$ and a $\sigma$-algebra\footnote{A $\sigma$-algebra on $\mathcal{O}$ is any collection of subsets of $\mathcal{O}$ containing the empty subset and closed under taking complements, countable unions, and countable intersections.} $\mathcal{F}$ of subsets of $\mathcal{O}$, a positive operator-valued measure (POVM) $E$ on $\mathcal{F}$ is defined as the map $E:\mathcal{F}\rightarrow \mathcal{B}_{+}(\mathcal{H})$, where $E(\cup_{X\in\mathcal{F}} X)=\sum_{X\in\mathcal{F}}E(X)=I$, $\cup_{X\in\mathcal{F}} X$ being a union of pairwise disjoint subsets $X\in\mathcal{F}$ satisfying $\cup_{X\in\mathcal{F}} X=\mathcal{O}$. Here $\mathcal{B}_+(\mathcal{H})$ denotes the set of positive semidefinite operators on the Hilbert space $\mathcal{H}$ and $I$ is the identity operator on $\mathcal{H}$. $E$ becomes a projection-valued measure (PVM) under the additional constraint that $E(X)^2=E(X)$ for all $X\in\mathcal{F}$.
	
When the Hilbert space is two-dimensional, i.e., $\mathcal{H}\cong\mathbb{C}^2$, and we let $\mathcal{O}=\{x^{(1)},x^{(2)}\}$ with $\mathcal{F}=\{\emptyset,\{x^{(1)}\},\{x^{(2)}\},$ $\{x^{(1)},x^{(2)}\} \}$, we have the notion of a binary qubit POVM as any POVM $E:\mathcal{F}\to\mathcal{B}_+(\mathbb{C}^2)$.
\label{def1}
\end{definition}
Definition \ref{def1} implies that for binary qubit POVMs it is enough to specify $E(\{x^{(1)}\})$ since $E(\{x^{(2)}\})=I-E(\{x^{(1)}\}$. $x^{(1)}$ and $x^{(2)}$ are outcomes of the measurement and we will henceforth label them $x^{(1)} = 1$ and $x^{(2)}=-1$ in keeping with the convention for spin-1/2 POVMs in quantum theory and also for later notational convenience. For simplicity, we write $E(x)\equiv E(\{x\})$ where $x=\pm1$. Thus, $E(x)$ is a positive semidefinite operator bounded above by the identity, i.e., an effect. We can write it in the usual Pauli operator basis of $\mathcal{B}(\mathbb{C}^2)$. Then every binary qubit POVM $E$ can be parametrized with four real parameters $\alpha$ and $\vec{a}=(a_1,a_2,a_3)$ (see Sec.~3.1 in Ref.~\cite{HRS08}):
\begin{align}
&E(1)=\frac{1}{2}\left(\alpha I+\vec{a}\cdot\vec{\sigma}\right),\quad E(-1)=\frac{1}{2}\left((2-\alpha) I -\vec{a}\cdot\vec{\sigma}\right),\nonumber\\
&a\leq\alpha\leq 2-a, \textrm{ where }a\equiv||\vec{a}||=\sqrt{a_1^2+a_2^2+a_3^2}.
\label{biqodef}
\end{align}
The POVM is a projection-valued measure (PVM) when $\alpha=a=1$. We will call the vector $\vec{a}$ the {\em Bloch vector} of the POVM $E$. Three or more POVMs will be called {\em coplanar} if their Bloch vectors are coplanar.\\

We now define the joint measurability (or, equivalently, compatibility) of POVMs, following Ref.~\cite{HRS08}: 

\begin{definition}[Joint measurability of POVMs]
A set of POVMs $\{E_k\}_{k=1}^N$ with respective outcome sets and $\sigma$-algebras given by $\{\mathcal{O}_k,\mathcal{F}_k)\}_{k=1}^N$ is said to be jointly measurable (or compatible) if there exists a (joint) POVM $G$ defined on $(\mathcal{O}_1\times \mathcal{O}_2\times\dots\times \mathcal{O}_N,\mathcal{F}_1\otimes\mathcal{F}_2\otimes\dots\otimes\mathcal{F}_N)$ such that 
\begin{align}
&G(X_1\times \mathcal{O}_1\times\dots\times \mathcal{O}_N)=E_1(X_1),\nonumber\\
&G(\mathcal{O}_1\times X_2\times\dots\times \mathcal{O}_N)=E_2(X_2),\nonumber\\
&\vdots\nonumber\\
&G(\mathcal{O}_1\times \mathcal{O}_2\times\dots\times X_N)=E_N(X_N),
\end{align}
for all $X_k\in \mathcal{F}_k$. Here, for the $\sigma$-algebras $\mathcal{F}_1,\mathcal{F}_2,\dots,\mathcal{F}_N$, we denote by $\mathcal{F}_1\otimes\mathcal{F}_2\otimes\dots\otimes\mathcal{F}_N$  the product $\sigma$-algebra generated by sets of the form $X_1\times X_2\times \dots\times X_N$, $X_k\in\mathcal{F}_k$ for $k=\overline{1,N}$. Formally, this product $\sigma$-algebra is the intersection of all $\sigma$-algebras containing $\{X_1\times X_2\times \dots\times X_N|X_k\in\mathcal{F}_k, k=\overline{1,N}\}$.
\label{def2}
\end{definition}

Following Ref.~\cite{KHF14}, we also need the notion of a joint measurability structure:
\begin{definition}[Joint measurability structure]
A joint measurability structure is a hypergraph with a set of vertices, $V$, and a family of (finite) subsets (called hyperedges) of $V$, denoted $E\subseteq \{e|e\subseteq V\}$. Each vertex denotes a measurement in a set of measurements (indexed by $V$) and each hyperedge denotes a subset of compatible (or jointly measurable) measurements. Any subset of vertices that do not share a common hyperedge represents an incompatible subset of the given set of measurements. To model the requirement that every subset of a set of compatible measurements is also compatible, a joint measurability structure must also satisfy
\begin{equation}
e\in E, e'\subseteq e \Rightarrow e'\in E.
\end{equation}
\end{definition}
We will call a joint measurability structure {\em quantum-realizable}, or that it admits a {\em quantum representation}, if its vertices can be represented by quantum measurements, i.e., POVMs, satisfying the (in)compatibility relations dictated by it.

We now mention some important examples of joint measurability structures, some of which we will realize with binary qubit POVMs:
\begin{example}[$N$-cycle scenario]\label{NCycdefn}
Let $s=\{E_1,\ldots,E_N\}$ be a set of $N$ vertices (representing measurements). (In)compatibility relations on $s$ of the form
\begin{align}
\Big\{&\{E_{1},E_{2}\},\{E_{2},E_{3}\},\ldots,\{E_{N-1},E_N\},\{E_N,E_1\}\Big\},
\end{align}
are said to form a joint measurability structure called the $N$-cycle scenario (Fig.~\ref{NCycSl}).
\begin{figure}[H]
\centering
\includegraphics[scale=0.17]{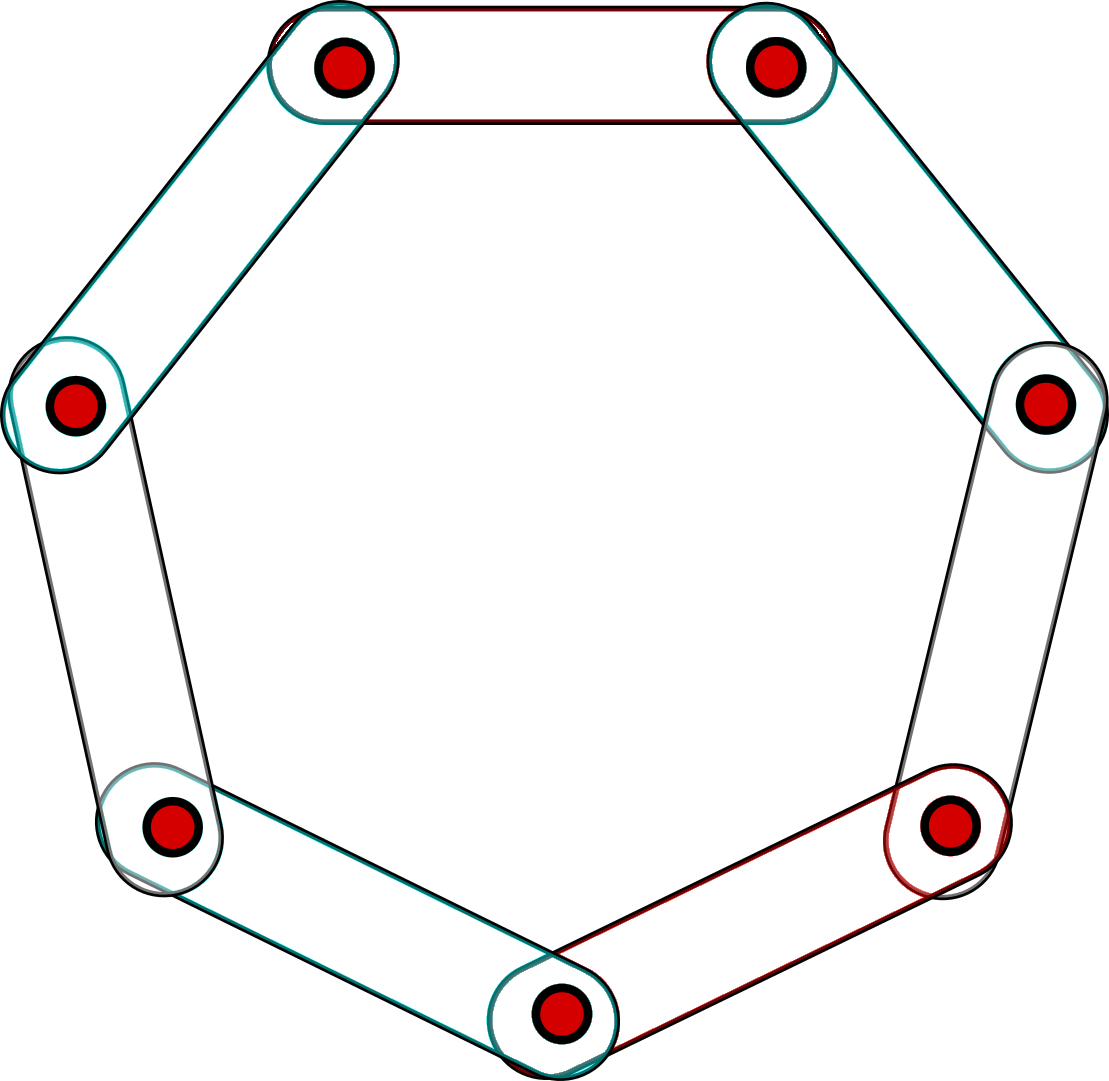}
\caption{$N$-cycle scenario for $N=7$.}
\label{NCycSl}
\end{figure}
\end{example}
\begin{example}[$(N,M)$-compatible set]
Let $s=\{E_1,\ldots,E_N\}$ be a set of $N$ vertices (representing measurements). If each $M$-element subset of $s$ is compatible ($M\leq N$), while every subset of $s$ of higher cardinality than $M$ is incompatible, then the set $s$ is said to be $(N,M)$-compatible and this joint measurability structure is called $(N,M)$-compatibility scenario.
\end{example}
Trivially, an $(N,1)$-compatible set is just a set of $N$ pairwise incompatible observables, while $(N,N)$-compatible set is just a set of $N$ compatible measurements. We now follow with two more special cases of $(N,M)$-compatibility scenarios for $M=2$ and $M=N-1$.
\begin{example}[$N$-complete scenario]
Let $s=\{E_1,\ldots,E_N\}$ be a set of $N$ vertices (representing measurements). A joint measurability structure on $s$ where each pair of measurements is compatible but no three measurements are (i.e., $s$ is an $(N,2)$-compatible set) is called an $N$-complete scenario. The hypergraph representing this structure is a complete graph with $N$ vertices.
\end{example}
\begin{example}[$N$-Specker scenario]\label{NSpeckdefn}
Let $s=\{E_1,\ldots,\\E_N\}$ be the set of $N$ vertices (representing measurements). If each $(N-1)$-element subset of $s$ is compatible while $s$ itself is incompatible (i.e. $s$ is an $(N,N-1)$-compatible set) then the joint measurability structure of $s$ is called an $N$-Specker scenario. This is the generalization of the notion of Specker's scenario which is the simplest non-trivial $N$-Specker scenario for $N=3$.
\begin{figure}[H]
\centering
\includegraphics[scale=0.18]{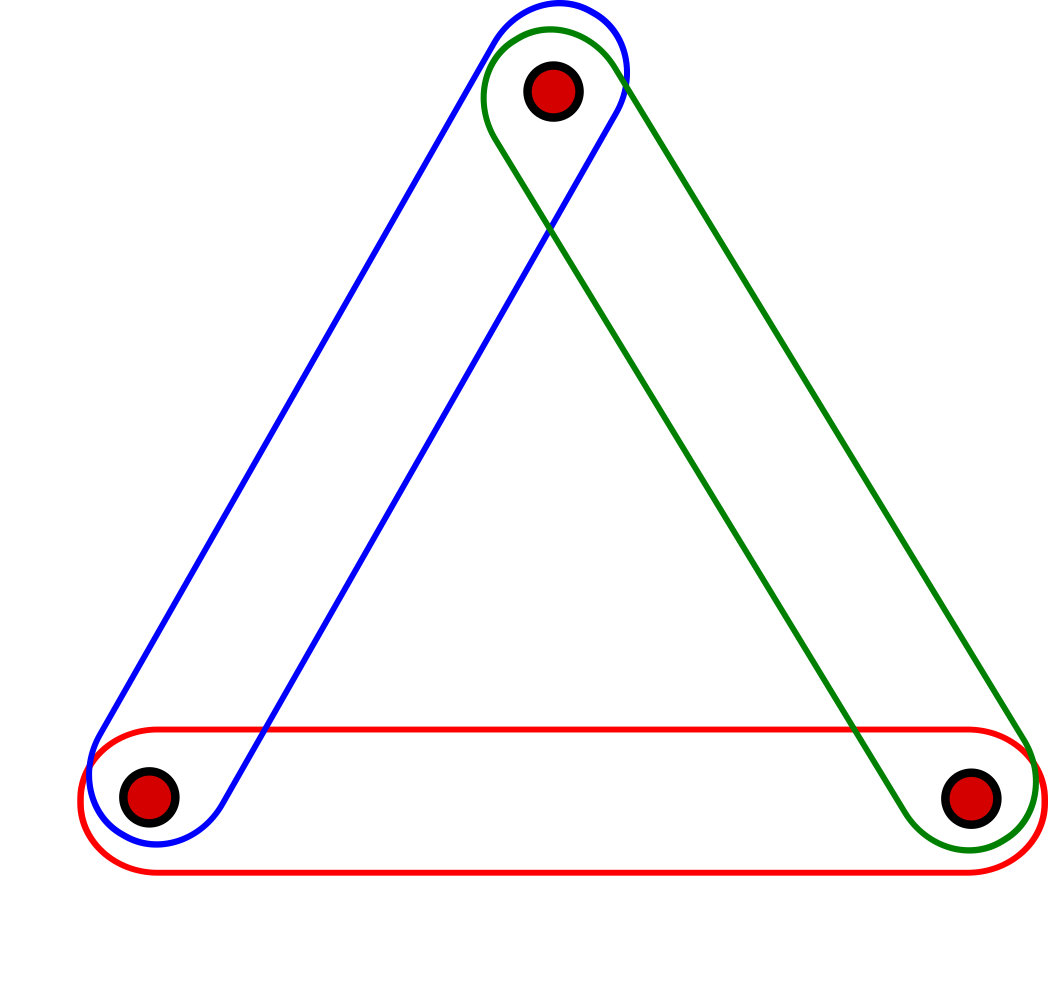}
\caption{Specker's scenario. Notice that this is the special case for $N$-Specker, $N$-Cycle and $N$-Complete scenarios when $N=3$.}
\label{speckerscenario}
\end{figure}
\end{example}
We provide more examples of joint measurability structures in Section~\ref{sec5} along with their quantum realizations with binary qubit POVMs.

\subsection{Conditions for joint measurability of qubit POVMs}
Previous research has uncovered many analytical criteria for the joint measurability of POVMs, particularly for qubits. We now collect some known results on joint measurability of qubit POVMs that will be of interest to us. 

\subsubsection{Two binary qubit POVMs}
Yu {\em et al.}~\cite{YLL10} proved a necessary and sufficient condition for the joint measurability of two binary qubit POVMs. We state this condition below.
\begin{theorem}\label{2povms}
Two binary qubit POVMs $E_1(1)=\frac{1}{2}(\alpha_1I+\vec{a}_1\cdot\vec{\sigma})$ and $E_2(1)=\frac{1}{2}(\alpha_2I+\vec{a}_2\cdot\vec{\sigma})$ are jointly measurable if and only if 
\begin{align}
&(1-F_1^2-F_2^2)\left(1-\frac{(\alpha_1-1)^2}{F_1^2}-\frac{(\alpha_2-1)^2}{F_2^2}\right)\nonumber\\
\leq&\Big(\vec{a}_1\cdot\vec{a}_2-(\alpha_1-1)(\alpha_2-1)\Big)^2,
\label{gencdt2}
\end{align}
where we have
\begin{align}
F_1&\equiv\frac{1}{2}\Bigg(\sqrt{\alpha_1^2-a_1^2}+\sqrt{(2-\alpha_1)^2-a_1^2}\Bigg),\nonumber\\
F_2&\equiv\frac{1}{2}\Bigg(\sqrt{\alpha_2^2-a_2^2}+\sqrt{(2-\alpha_2)^2-a_2^2}\Bigg).
\end{align}
\end{theorem}

Of particular interest to us will be a class of POVMs called {\em unbiased} binary qubit POVMs. We define them below:
\begin{definition}[Unbiased binary qubit POVMs]\label{unbiqodef}
Binary qubit POVMs specified by 
\begin{equation}
E(x)=\frac{1}{2}(I+x\eta\vec{n}\cdot\vec{\sigma})
\end{equation}
where $x=\pm1$, $n=||\vec{n}||=1$ and $0\leq\eta\leq1$ are called unbiased. Otherwise, they are biased. The parameter $\eta$ is usually referred to as the purity or sharpness parameter, it's upper bound $\eta=1$ corresponding to the case of a sharp (projective) measurement.
\end{definition}
In the light of the Definition~\ref{unbiqodef}, we define the bias $b$ associated with an outcome $x=+1$ of a binary qubit POVM $E$ (simply, ``the bias of $E$") by rewriting Eq.~\eqref{biqodef} as
\begin{equation}
E(\pm 1)=\frac{1}{2}\left((1\pm b)I\pm\vec{a}\cdot\vec{\sigma}\right),\quad |b|\leq 1-a.
\end{equation}
An unbiased qubit POVM has zero bias.
\begin{theorem}\label{2unbjmc}
Two unbiased binary qubit POVMs specified by $E_1(x_1)=\frac{1}{2}(I+x_1\eta_1\vec{n}_1\cdot\vec{\sigma})$ and\\ $E_2(x_2)=\frac{1}{2}(I+x_2\eta_2\vec{n}_2\cdot\vec{\sigma})$ are jointly measurable if and only if 
\begin{equation}
||\eta_1\vec{n}_1+\eta_2\vec{n}_2||+||\eta_1\vec{n}_1-\eta_2\vec{n}_2||\leq 2
\end{equation}
\end{theorem}
The proof of this theorem follows directly from Eq.~\eqref{gencdt2}, although it was independently proven earlier \cite{Busch86}. The reader may consult Ref.~\cite{HRS08} for a guide to previous literature on necessary/sufficient conditions for joint measurability.

\begin{corollary}\label{Leta2}
Two unbiased binary qubit POVMs with the same purity $\eta$ -- given by $E_1(x_1)=\frac{1}{2}(I+x_1\eta\vec{n}_1\cdot\vec{\sigma})$ and  $E_2(x_2)=\frac{1}{2}(I+x_2\eta\vec{n}_2\cdot\vec{\sigma})$ -- are jointly measurable if and only if
\begin{equation}
\eta\leq\frac{1}{\left|\sin\frac{\phi}{2}\right|+\left|\cos\frac{\phi}{2}\right|}
\label{eta2},
\end{equation}
where $\phi$ is the angle between their Bloch vectors $\vec{n}_1$ and $\vec{n}_2$, i.e., $\vec{n}_1\cdot\vec{n}_2=\cos\phi$.
\end{corollary}

\subsubsection{Joint measurability and geometrically equivalent sets of POVMs} \label{subsec2_1}
In this subsection we will describe a relation between two sets of qubit POVMs that, when satisfied, leads to the same joint measurability structure for them. We are motivated to define this relation based on two relatively obvious observations that are formalized in the next two propositions.

\begin{proposition} \label{mpovms} Let $s=\{E_1,\ldots,E_N\}$ and $s'=\{F_1,\ldots,F_N\}$ be two sets of POVMs, each $E_k$ and $F_k$ with the same outcome set $\mathcal{O}_k$. If there exist permutations $\mathbb{Perm}=(\mathrm{Perm}_1,\ldots,\mathrm{Perm}_N)$ such that
\begin{align}
\forall k=\overline{1,N},\text{ } \forall x_k\in O_k,\text{ } F_k(x_k)=E_k(\mathrm{Perm}_k(x_k)),
\end{align}
then we formally write $F_k=\mathrm{Perm}_k E_k$ or, succinctly, $s'=\mathbb{Perm} s$, and
\begin{enumerate}
\item the sets $s$ and $s'$ exhibit the same joint measurability structure, 
\item if a subset $r\subseteq s$, denoted $r=\{E_{k_1},\ldots,E_{k_{|r|}}\}$, is jointly measurable, with a joint POVM $G^r$, then its corresponding subset $r'\subseteq s'$, denoted $r'=\{F_{k_1},\ldots,F_{k_{|r|}}\}$, is also jointly measurable, with a joint POVM given by
\begin{align}
G^{r'}(x_{k_1},\ldots,x_{k_{|r|}})&=G^{r}(\mathrm{Perm}_{k_1}(x_{k_1}),\ldots,\mathrm{Perm}_{k_{|r|}}(x_{k_{|r|}}))\nonumber\\
&\equiv\mathbb{Perm}^rG^r(x_{k_1},\ldots,x_{k_{|r|}}),
\label{GPerm}
\end{align}
where $\mathbb{Perm}^r$ contains those permutations that refer to the POVMs from the set $r$.
\end{enumerate}
\begin{proof}
If $s'=\mathbb{Perm}s$ then $s=\mathbb{Perm}^{-1}s'$, so it is enough to show that if $r\subseteq s$ is jointly measurable then its corresponding subset $r'\subseteq s'$ is jointly measurable as well: this would imply that any subset of $s$ is compatible if and only if the corresponding subset of $s'$ is compatible, i.e., $s$ and $s'$ exhibit the same joint measurability structure. This is easy to show by noticing that if $G^r$ is a joint POVM for $r$ then $G^{r'}$ given by Eq.~\eqref{GPerm} is a joint POVM for $r'$ since it is a valid POVM, as it has the same range as the POVM $G^r$, and its marginals recover the different $F_{k}\in r'$, i.e., 
\begin{align}
\sum_{\vec{y}\in\bigtimes_{i=1}^{|r|}\mathcal{O}_i}^{y_k=x_k}G^{r'}(\vec{y})&=\sum_{\vec{y}\in\bigtimes_{i=1}^{|r|}\mathcal{O}_i}^{y_k=\mathrm{Perm}_k(x_k)}G^r(\vec{y})=\nonumber\\
&=E_k(\mathrm{Perm}_k(x_k))=F_k(x_k).
\end{align} 
Here, for convenience, we assumed that the POVMs are labelled such that $r=\{E_k\}_{k=1}^{|r|}$ meaning that also $r'=\{F_k\}_{k=1}^{|r|}$. Also, by putting $y_k=a$ above the summation sign, we mean that the summation is carried over all elements of the string $\vec{y}$ except the $k$th element which is held fixed at $y_k=a$.
\end{proof}
\end{proposition}
\begin{remark}\label{unbiasedremark}
Interpreted for the case of unbiased binary qubit POVMs, Proposition \ref{mpovms} says that their joint measurability is dependent only on the lines on which their Bloch vectors lie and not the particular orientation of a Bloch vector along a line because the choice of which outcome is labelled $+1$ and which is labelled $-1$ does not affect their joint measurability.
\end{remark}
\begin{proposition}\label{geeq1}
Let $s=\{E_1,\ldots,E_N\}$ and $s'=\{F_1,\ldots,F_N\}$ be two sets of qubit POVMs, each $E_k$ and $F_k$ with the same outcome set $\mathcal{O}_k$, where for all $k=\overline{1,N}$ and $x_k\in\mathcal{O}_k$,
\begin{align}
E_k&=\frac{1}{2}\Big(\alpha_k(x_k)I+\vec{e}_k(x_k)\cdot\vec{\sigma}\Big),\nonumber\\
F_k&=\frac{1}{2}\Big(\alpha_k(x_k)I+\vec{f}_k(x_k)\cdot\vec{\sigma}\Big).
\end{align} 
If $\vec{f}_k(\cdot)$ and $\vec{e}_k(\cdot)$ are related by some orthogonal transformation $O\in \mathrm{O}(3)$, where $\mathrm{O}(3)$ is the orthogonal group, i.e., 
\begin{align}
\exists O\in\mathrm{O}(3)\text{ } \forall k=\overline{1,N},\text{ } \forall x_k\in\mathcal{O}_k:\text{ } \text{ }\vec{f}_k(x_k)=O\vec{e}_k(x_k),
\end{align}
then we formally write $F_k=OE_k$, or, succinctly, $s'=Os$, and we have that 
\begin{enumerate}
\item the sets $s$ and $s'$ exhibit the same joint measurability structure,
\item if $r\subseteq s$ is jointly measurable, with a joint POVM $G^r$, then its corresponding $r'=Or\subseteq s'$ is jointly measurable, with a joint POVM given by $G^{r'}=OG^r$. 
\end{enumerate}
\begin{proof}
This follows from the fact that $\vec{e}_k(\cdot),\vec{f_k}(\cdot)\in\mathbb{R}^3$ and $O(3)$ is the group of isometries that fix the origin in $\mathbb{R}^3$. If we passively act on the chosen axes in $\mathbb{R}^3$ with $O^{-1}$ we would get the new orthogonal coordinate axes such that $s'$ looks the same way in the new coordinate system as $s$ looks in the old one and vice-versa. Since joint measurability is a notion independent of choice of the axes in $\mathbb{R}^3$, $s$ and $s'$ must exhibit the same joint measurability structure.

More rigorously, we show that any $r\subseteq s$ is compatible if and only if the corresponding subset $r'=Or\subseteq s'$ is compatible. If $s'=Os$, then $s=O^{-1}s'$, so it is enough to show that the compatibility of $r$ implies the compatibility of $r'$. Consider a jointly measurable subset $r\subseteq s$, assuming a labelling of POVMs such that $r=\{E_k\}_{k=1}^{|r|}$, with a joint POVM
\begin{equation}
G^r(\vec{x})=\frac{1}{2}\Big(\gamma(\vec{x})I+\vec{g}(\vec{x})\cdot\vec{\sigma}\Big),\quad \vec{x}\in\mathcal{O}_1\times\cdots\times\mathcal{O}_{|r|}.
\end{equation}
Requiring that $E_k\in r$ are marginals of $G^r$ we have that 
\begin{align}
\sum_{\vec{y}\in\bigtimes_{i=1}^{|r|}\mathcal{O}_i}^{y_k=x_k}\gamma(\vec{y})=\alpha_k(x_k),\text{ } \sum_{\vec{y}\in\bigtimes_{i=1}^{|r|}\mathcal{O}_i}^{y_k=x_k}\vec{g}(\vec{y})=\vec{e}_k(\vec{x}_k).\label{gemarg}
\end{align}
Consider the set of operators 
\begin{equation}
G^{r'}=OG^{r}=\frac{1}{2}\Big(\gamma(\vec{x})+O\vec{g}(x)\cdot\vec{\sigma}\Big).
\end{equation}
Since $O$ is an orthogonal transformation we have
\begin{equation}
||\vec{g}(\vec{x})||=||O\vec{g}(\vec{x})||,
\end{equation}
so that all of the operators $G^{r'}(\vec{x})$, $\vec{x}\in\mathcal{O}_1\times\cdots\times O_{|r|}$, are positive semidefinite. Using Eq.~\eqref{gemarg} and linearity of $O$ we find that the marginals of $G^{r'}$ are indeed $F_k$s:
\begin{align}
\sum_{\vec{y}\in\vec{y}\in\bigtimes_{i=1}^{|r|}\mathcal{O}_i}^{y_k=x_k}G^{r'}(\vec{x})&=\sum_{\vec{y}\in\bigtimes_{i=1}^{|r|}\mathcal{O}_i}^{y_k=x_k}\frac{1}{2}\Big(\gamma(\vec{y})+O\vec{g}(\vec{y})\cdot\vec{\sigma}\Big)=\nonumber\\
&=\frac{1}{2}\alpha_k(x_k)+\frac{1}{2}O\vec{e}_k(x_k)=\nonumber\\
&=OE_k(x_k)=F_k(x_k).
\end{align}
Hence, $r'$ is compatible, with a joint POVM $G^{r'}$.
\end{proof}
\end{proposition}
\begin{definition}\label{geeq3}
Two sets of qubit POVMs $s$ and $s'$ are said to be geometrically equivalent if they are related by some relabelling of outcomes $\mathbb{Perm}$ and some orthogonal transformation $O$ such that $s'=O\mathbb{Perm}s$.
\end{definition}
\begin{proposition}\label{geeq2}
Two geometrically equivalent sets of POVMs, $s$ and $s'$, with $s'=O\mathbb{Perm}s$, exhibit the same joint measurability structure. If $r\subseteq s$ is compatible, with a joint POVM $G^r$, then its corresponding subset $r'=O\mathbb{Perm}^{r}r\subseteq s'$ is also compatible, with a joint POVM $G^{r'}=O\mathbb{Perm}^{r}G^r$.
\begin{proof}
This follows directly from Propositions~\ref{mpovms} and \ref{geeq1}.
\end{proof} 
\end{proposition}
\begin{corollary}\label{GEbinary}
Two sets of unbiased binary qubit POVMs, $s$ and $s'$, are geometrically equivalent (and therefore exhibit the same joint measurability structure) if and only if the lines defined by their Bloch vectors are related by an orthogonal transformation $O\in\mathrm{O}(3).$ If $r\subseteq s$ is compatible, with a joint POVM $G^r$, then its corresponding subset $r'=O\mathbb{Perm}^{r}r\subseteq s'$ is also compatible, with a joint POVM $G^{r'}=O\mathbb{Perm}^rG^{r}$, where $\mathbb{Perm}^r$ denotes the relabelling of outcomes, if necessary, on the set $r$.
\begin{proof}
This follows from the Definition~\ref{geeq3} (of geometrically equivalent sets of POVMs), Proposition~\ref{geeq2} and Remark~\ref{unbiasedremark} following Proposition~\ref{mpovms}.
\end{proof} 
\end{corollary}

\subsubsection{Multiple binary qubit POVMs}
We first note a necessary condition for the joint measurability of three binary qubit POVMs (that may be biased) obtained by Pal and Ghosh \cite{PG11}:
\begin{theorem}\label{necc3povms}
A necessary condition for the joint measurability of three binary qubit POVMs --- denoted $E_k(x_k)=\frac{1}{2}(\alpha_k I+x_k\eta_k\vec{n}_k\cdot\vec{\sigma})$, where $x_k\in\{\pm1\}, k=\overline{1,3}$ --- is the following: 

\begin{equation}\label{jmfor3ubiquos}
\sum_{i=0}^3||\vec{v}_i-\vec{v}_{\rm FT}||\leq 4,
\end{equation}
where $\vec{v}_{\rm FT}$ is the Fermat-Torricelli (FT) point of the following four points in $\mathbb{R}^3$: $\vec{v}_0=-\sum_{i=1}^3\eta_i\vec{n}_i$ and $\vec{v}_j=-2\eta_j\vec{n}_j-\vec{v}_0$ for $j\in\{1,2,3\}$.\footnote{Given a set of points in $\mathbb{R}^3$, the Fermat-Torricelli point of this set is a point that minimizes the sums of distances from itself to points in the set. See Refs.~\cite{PG11, KM94} for more on the FT point.}
\end{theorem}

For the case of unbiased qubit POVMs where $\vec{n}_1,\vec{n}_2,\vec{n}_3$ are mutually orthogonal directions, this condition (Eq.~\eqref{jmfor3ubiquos}) reduces to the sufficient condition for joint measurability proved earlier in Ref.~\cite{Busch86}, namely,
\begin{equation}
\eta_1^2+\eta_2^2+\eta_3^2\leq 1.
\end{equation}
Note that $\vec{v}_{\rm FT}=(0,0,0)$ in this case.

 In Ref.~\cite{YO13}, the sufficiency of Eq.~\eqref{jmfor3ubiquos} for the case of any three unbiased qubit POVMs (that need not be mutually orthogonal) was shown. Hence, Eq.~\eqref{jmfor3ubiquos} is necessary and sufficient for joint measurability of three unbiased binary qubit POVMs. 
 Furthermore, if a set of three unbiased binary qubit POVMs is also coplanar, then the four vectors $\{\vec{v}_i\}_{i=0}^3$ are also coplanar and their FT vector can be explicitly computed and is given by the following recipe: if the four coplanar vectors $\{\vec{v}_i\}_{i=0}^3$ make a convex quadrilateral then their FT point is located at the intersection of the diagonals of that quadrilateral; on the other hand, if one of the points is inside the triangle determined by the convex hull of the other three (i.e., it is a convex combination of the other three points), then the FT point is the point inside the triangle (see Ref~\cite{Plastria06}). Using this property in Ref.~\citep{YO13} (their Example 1), the necessary and sufficient conditions were derived for this coplanar case and we mention them below.
 	\begin{corollary}\label{3copljmc}
 		Three unbiased binary and coplanar qubit POVMs $\{E_k\}_{k=1}^3$, with Bloch vectors $\{\vec{a}_k\}_{k=1}^3$ lying on the Bloch lines as in the Fig~\ref{s1} and pointing towards the upper half of the plane, are jointly measurable if and only if
 	\begin{equation}
 			\left|\left|\vec{a}_1+\vec{a}_3\right|\right|+\left|\left|\vec{a}_2-\vec{a}_1\right|\right|+\left|\left|\vec{a}_3-\vec{a}_2\right|\right|\leq2,
 	\end{equation} 
 			when $\vec{a}_2$ is not a convex combination of $\vec{a}_1$ and $\vec{a}_3$, and
 	 \begin{equation}
 			\left|\left|\vec{a}_1+\vec{a}_3\right|\right|+\left|\left|\vec{a}_3-\vec{a}_1\right|\right|\leq2,
 	\end{equation}
 			when $\vec{a}_2$ is a convex combination of $\vec{a}_1$ and $\vec{a}_3$. 
 		
\begin{proof}
 In Example 1 of Ref~\cite{YO13}, it was shown that in the case that $\vec{a}_2$ is not a convex combination of $\vec{a}_1$ and $\vec{a}_3$ then $\vec{v}_2=\vec{a}_1-\vec{a}_2+\vec{a}_3$ from Eq~\eqref{jmfor3ubiquos} lies inside the triangle determined by the convex hull of $\vec{v}_0$, $\vec{v}_1$ and $\vec{v}_3$ and thus $\vec{v}_\text{FT}=\vec{v}_2$. In the case that $\vec{a}_2$ is a convex combination of $\vec{a}_1$ and $\vec{a}_3$, the vectors $\{\vec{v}_j\}_{j=0}^3$ make a convex quadrilateral and we have that $\vec{v}_\text{FT}$ points from the origin to the intersection of the diagonals of that quadrilateral. Substituting these values for $\vec{v}_\text{FT}$ into Eq~\eqref{jmfor3ubiquos}, we obtain the joint measurability conditions of Corollary \ref{3copljmc}.
\end{proof}
\end{corollary} 
Notice that in the case where $\vec{a}_2$ is a convex combination of $\vec{a}_1$ and $\vec{a}_3$ (and, thus, $E_2$ is a convex combination of $E_1$ and $E_3$ due to unbiasedness), then the 3-way joint measurability condition reduces to the joint measurability condition for two POVMs $E_1$ and $E_3$, cf.~Theorem ~\ref{2unbjmc}. In other words, as intuitively expected, if two unbiased binary qubit POVMs are compatible, then any convex combination of theirs is also compatible with them.
\begin{corollary}\label{3unbcopjmc}
Three unbiased binary coplanar qubit POVMs, $\{E_k\}_{k=1}^3$, with the same purity $\eta$ are jointly measurable if and only if 
\begin{equation}\label{3sameeta}
\eta\leq\frac{1}{\cos\frac{\phi_1+\phi_2}{2}+\sin\frac{\phi_1}{2}+\sin\frac{\phi_2}{2}},
\end{equation}
where $\phi_1,\phi_2\leq\pi/2$ are angles between the lines determined by the Bloch vectors of $E_1,$ $E_2$ and $E_3$ with the arrangement shown in Fig.~\ref{s1}.
\label{3unb}
\begin{proof}
		This is a special case of the Corollary~\ref{3copljmc} where the observables have the same purity. Then the norm of Bloch vectors of all three POVMs is equal to $\eta$ and, hence, none of them can be a convex combination of other two. It is then easy to verify that inequality given in Eq.~\eqref{jmfor3ubiquos} reduces to Eq.~\eqref{3sameeta}.
\end{proof}
\end{corollary}
\begin{figure}[H]
\centering
\includegraphics[scale=0.54]{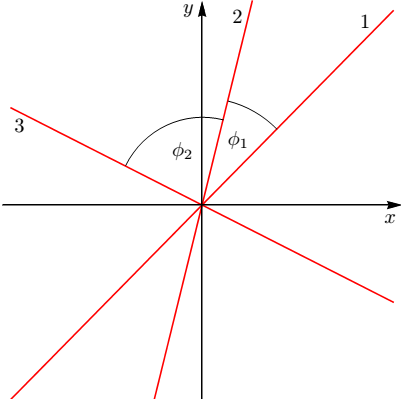}
\caption{The direction of the Bloch vector of $E_3$ is between those of $E_1$ and $E_2$. Here $\phi_1$ is the (acute) angle between $1$ and $3$ and $\phi_2$ is the (acute) angle between $2$ and $3$.}
\label{s1}
\end{figure}

So far, we have a necessary and sufficient condition for compatibility of two binary qubit POVMs (Theorem \ref{2povms}) \cite{YLL10} and a condition that is necessary and sufficient for compatibility of three unbiased binary qubit POVMs (Theorem \ref{necc3povms}) \cite{PG11, YO13}. We will now state one necessary and one sufficient condition for joint measurability of a set of $N$ unbiased binary qubit POVMs with same purity $\eta$.

\begin{theorem} Let $E_k(x_k)=\frac{1}{2}(I+\eta x_k\vec{n_k}\cdot\vec{\sigma}),$ $x_k\in\{\pm1\},$ $k=\overline{1,N}$ be $N$ unbiased binary qubit POVMs. Then:
\begin{enumerate}
\item a necessary condition for their joint measurability is 
\begin{align}\label{NNecess}
\eta&\leq\frac{1}{N}\underset{\vec{x}}{\max}\left|\left|\sum_{k=1}^Nx_k\vec{n}_k\right|\right|,\quad\text{and}
\end{align}
\item a sufficient condition for their joint measurability is 
\begin{align}\label{NSuff}
\eta&\leq\frac{2^N}{\sum_{\vec{x}}\left|\left|\sum_{k=1}^Nx_k\vec{n}_k\right|\right|},
\end{align}
where $\vec{x}=(x_1,x_2,\ldots,x_N)\in\{\pm1\}^N$.
\end{enumerate}
\label{NNecSuf}
\end{theorem}
The necessary condition (part 1) above is proven in Appendix B, Theorem 3, of Ref.~\cite{KG14} while the sufficient condition (part 2) is proven in Appendix F, Theorem 13, of Ref.~\cite{LSW11}.\footnote{Note that a different necessary condition was incorrectly proven in Ref.~\cite{LSW11}. However, the proof of the sufficient condition in Ref.~\cite{LSW11}, to which we refer, is correct.}
\begin{corollary}
A set of $N\in\{2,3\}$ orthogonal and unbiased binary qubit POVMs --- i.e., $E_k=\frac{1}{2}(I+\eta x_k\vec{e}_k\cdot\vec{\sigma})$,  $k=\overline{1,3}$, $x_k\in\{\pm1\}$, where $\vec{e}_k$ are unit vectors such that $\vec{e}_1\cdot\vec{e}_2=\vec{e}_2\cdot\vec{e}_3=\vec{e}_3\cdot\vec{e}_1=0$ --- are jointly measurable if and only if 
\begin{equation}
\eta\leq\frac{1}{\sqrt{N}}.
\end{equation}
\begin{proof}
Because the Bloch vectors are orthogonal we have that 
$$\left|\left|\sum_{\vec{x}}x_k \vec{e}_k\right|\right|=\sqrt{N},\quad \forall \vec{x}\in\{\pm1\}^3.$$
Putting that into Eqs.~\eqref{NNecess} and \eqref{NSuff} from Theorem \ref{NNecSuf} we get
\begin{align*}
\textrm{Neccesary condition}: \eta&\leq\frac{1}{N}\max\{\sqrt{N}\}=\frac{1}{\sqrt{N}},\textrm{ and}\\
\textrm{Sufficient condition}: \eta&\leq \frac{2^N}{\sqrt{N}\cdot2^N}=\frac{1}{\sqrt{N}}.
\end{align*}
\end{proof}
\label{ortnecsuf}
\end{corollary}

\subsubsection{Specker's scenario with unbiased binary qubit POVMs}
Among the most important immediate consequences of Corollaries \ref{Leta2}, \ref{3unb} and \ref{ortnecsuf} is the existence of Specker's scenario on a qubit, i.e., the setting with three incompatible qubit POVMs that are pairwise compatible, which was introduced in Example~\ref{NSpeckdefn}. Note that such a joint measurability structure is impossible with projective measurements \cite{LSW11,KHF14}. There are at least two standard ways to construct Specker's scenario using unbiased binary qubit POVMs \cite{LSW11}:
\begin{example}(Specker's scenario with orthogonal spin axes)
\upshape Let $E_k(x_k)=\frac{1}{2}(I+x_k\eta\vec{n}_k\cdot\vec{\sigma}),\quad k=\overline{1,3}$ be three orthogonal unbiased binary qubit POVMs (see Fig.~\ref{orttrine}, left), i.e., $\vec{n}_1\cdot\vec{n}_2=\vec{n}_2\cdot\vec{n}_3=\vec{n}_3\cdot\vec{n}_1=0$. By Corollary \ref{ortnecsuf} we have that each pair of POVMs is jointly measurable if $\eta\leq1/\sqrt{2}$ while all three of them are jointly measurable only if $\eta\leq1/\sqrt{3}$. Thus, there is a gap $\eta\in\left(1/\sqrt{3},1/\sqrt{2}\right]$ that yields pairwise joint measurability but no triplewise joint measurability.
\end{example}
\begin{example}
(Specker's scenario with trine spin axes)\upshape 
\hspace{2pt} Consider three POVMs in an equatorial plane of the Bloch ball equiangularly separated (see Fig.~\ref{orttrine}, right): $E_k(x_k)=\frac{1}{2}(I+x_k\eta\vec{n}_k\cdot\vec{\sigma})$, $k=\overline{1,3}$, where $\vec{n}_k=\left(\cos\frac{k-1}{3}\pi,\sin\frac{k-1}{3}\pi,0\right)$. These are the so-called trine spin POVMs.

Corollary \ref{Leta2} combined with Corollary \ref{3unb} says that each pair of these POVMs is jointly measurable if and only if 
$$\eta\leq\frac{1}{\sin\frac{\pi}{6}+\cos\frac{\pi}{6}}=\sqrt{3}-1,$$
while all three of them are compatible if and only if
$$\eta\leq\frac{1}{\cos\frac{\pi}{3}+2\sin\frac{\pi}{6}}=\frac{2}{3},$$
which means that Specker's scenario is realized for $\eta\in\left(\frac{2}{3},\sqrt{3}-1\right]$. 
\end{example}
\begin{figure*}
	\centering
	\includegraphics[scale=0.428]{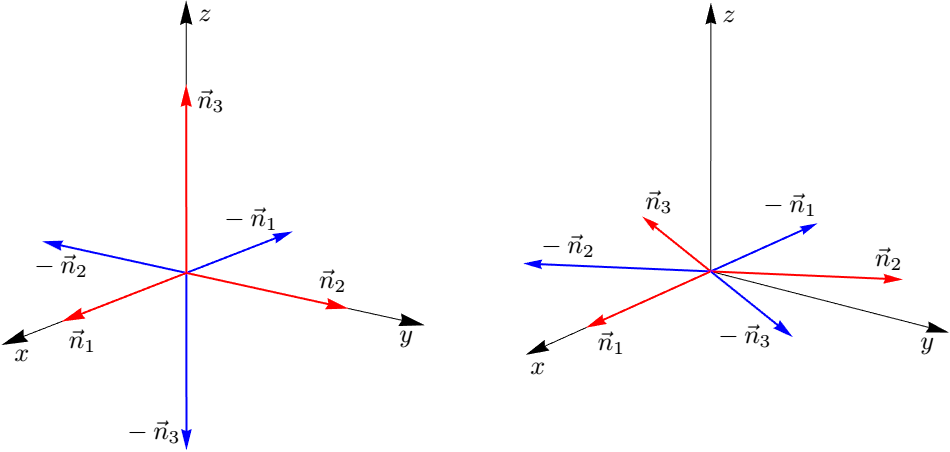}
	\caption{Orthogonal (left) and Trine Spin POVMs (right)}
	\label{orttrine}
\end{figure*}
\subsubsection{Adaptive strategy for constructing joint POVMs}
Here we recall the general adaptive strategy outlined for joint measurements in Ref.~\cite{ULMH16}. This is a technique for obtaining the joint POVM $G^s$ for the set of unbiased binary qubit POVMs $s=\left\{E_k(\pm1)=\frac{1}{2}(I\pm\eta\vec{n}_k\cdot\vec{\sigma})\right|k=\overline{1,N}\}$ that consists of the following steps. Take a set of projectors $\left\{\Pi'_k(\pm1)=\frac{1}{2}(I\pm\vec{n}'_k\cdot\vec{\sigma})|k=\overline{1,N'}\right\}$. We seek a joint measurement of $s$ in the form 
\begin{equation}
G^s(\vec{x})=\sum_{k=1}^{N'}\mu_k\sum_{y_k\in\{\pm1\}}\Pi'_k(y_k)p(\vec{x}|\Pi'_k=y_k),\text{ } \vec{x}\in\{\pm1\}^N,
\end{equation}
where $\{\mu_k\}_{k=1}^{N'}$ is a probability distribution according to which the PVMs $\{\Pi'_k\}_{k=1}^{N'}$ are implemented. By ``$\Pi'_k=y_k$", we denote the fact that the effect $\Pi'_k(y_k)$ (or the outcome $y_k$) was registered in a measurement of $\Pi'_k$.

The vectors $\{\vec{n}'_k\}_{k=1}^{N'}$ are chosen such that $\vec{n}_i\cdot\vec{n}'_k\neq 0$ for all $i=\overline{1,N}$ and $k=\overline{1,N'}$. The conditional probability distribution $p(\vec{x}|
\Pi'_k=\beta_k)$ is determined by signs of the projections of vectors $\vec{n}_i$ onto the vector $\vec{n}'_k$ in the following way:
\begin{align}
&p(\vec{x}|\Pi'_k=y_k)=\prod_{i=1}^N\delta_{x_i,\mathrm{sgn}(\vec{n}_i\cdot\vec{n}'_k)y_k},
\end{align}
where ${\rm sgn}(\vec{n}_i\cdot\vec{n}'_k)=+1$ if $\vec{n}_i\cdot\vec{n}'_k>0$, and $-1$ if $\vec{n}_i.\vec{n}'_k<0$.

The goal is to choose an appropriate set of parameters for $G^s$ such that it is indeed a valid joint POVM for $s$ i.e. that it is positive semidefinite and that the marginalization is correct:
\begin{equation}
\sum_{\vec{y}\in\{\pm1\}^N}^{y_k=x_k}G^s(\vec{y})=E_k(x_k).
\end{equation}

\subsubsection{Planar symmetric POVMs}
We now define a family of qubit POVMs that will be of particular interest in this paper:
\begin{definition}[Planar symmetric POVMs]
A set of planar symmetric POVMs is any set of $N$ coplanar and unbiased binary qubit POVMs with the same sharpness whose Bloch vectors define lines that equiangularly dissect the plane, i.e., the angle between successive lines is $\phi_0=\frac{\pi}{N}$, where $N$ is a positive integer.
\label{plansymdef}
\end{definition}
Note that, by Proposition~\ref{mpovms}, the joint measurability structure of  such a set of POVMs is independent of how we label the outcomes, i.e., the orientation along the line of a particular Bloch vector does not matter. However, unless specified otherwise, we will always consider the case where the POVMs are in an equatorial plane of the Bloch ball and the ``$+1$" outcomes are assigned to vectors in the upper half of the plane and ``$-1$" to the opposite vectors in the lower half, cf.~Fig.~\ref{plansymm}. With this convention, these POVMs are:
\begin{align}
&E_k(x_k)=\frac{1}{2}\Big(I+\eta x_k\vec{n}_k\cdot\vec{\sigma}\Big),\text{ }k=\overline{1,N},\text{ }x_k\in\{\pm1\},\nonumber\\
&\textrm{where}\quad\vec{n}_k=\left(\cos\frac{k-1}{N}\pi,\sin\frac{k-1}{N}\pi,0\right).
\label{plansymconv}
\end{align}
Note that trine spin POVMs are just the special case when $N=3$.
The following theorem, proven in Ref.~\cite{ULMH16}, is of important for many constructions in this paper:
\begin{theorem}\label{plansymjmc}
$N$ planar symmetric POVMs are jointly measurable if and only if their sharpness $\eta$ satisfies
\begin{equation}
\eta\leq\frac{1}{N\sin\frac{\pi}{2N}}.
\end{equation}
The effects of a joint POVM that reproduces the whole range $\eta\in\left(0,\frac{1}{N\sin\frac{\pi}{2N}}\right]$ are given by 
\begin{align}
G(\vec{x})&=\frac{1}{2N}\Big(I+\mu\vec{e}(\vec{x})\cdot\vec{\sigma}\Big),\nonumber\\
&\textrm{ where }\vec{x}\in\{\pm1\}^N,\text{ }\mu=\eta N\sin\frac{\pi}{2N}\in(0,1],\textrm{ and} \nonumber\\
\vec{e}(\vec{x})&\in\left\{\pm\left(\left.\cos\frac{(k-1)\pi}{N},\sin\frac{(k-1)\pi}{N},0\right)\right|k=\overline{1,N}\right\}\nonumber\\
&\textrm{ for odd } N,\textrm{ while}\nonumber\\
\vec{e}(\vec{x})&\in\left\{\pm\left(\left.\cos\frac{(2k-1)\pi}{2N},\sin\frac{(2k-1)\pi}{2N},0\right)\right|k=\overline{1,N}\right\}\nonumber\\
&\text{ for even } N.
\label{optN}
\end{align}
Here, outcome $\vec{x}$ is associated uniquely to the unit vector $\vec{e}(\vec{x})$ by the relation 
\begin{align}\label{plansymout}
\vec{x}=\left(\sgn\Big(\vec{n}_1\cdot\vec{e}(\vec{x})\Big),\ldots,\sgn\Big(\vec{n}_N\cdot \vec{e}(\vec{x})\Big)\right),
\end{align}
and the null (or zero) effect is assigned to any outcome $\vec{x}$ not obtainable through Eq.~\eqref{plansymout}.
\label{gluslov}
\end{theorem}
\begin{figure*}[htb!]
	\centering
	\includegraphics[scale=0.285]{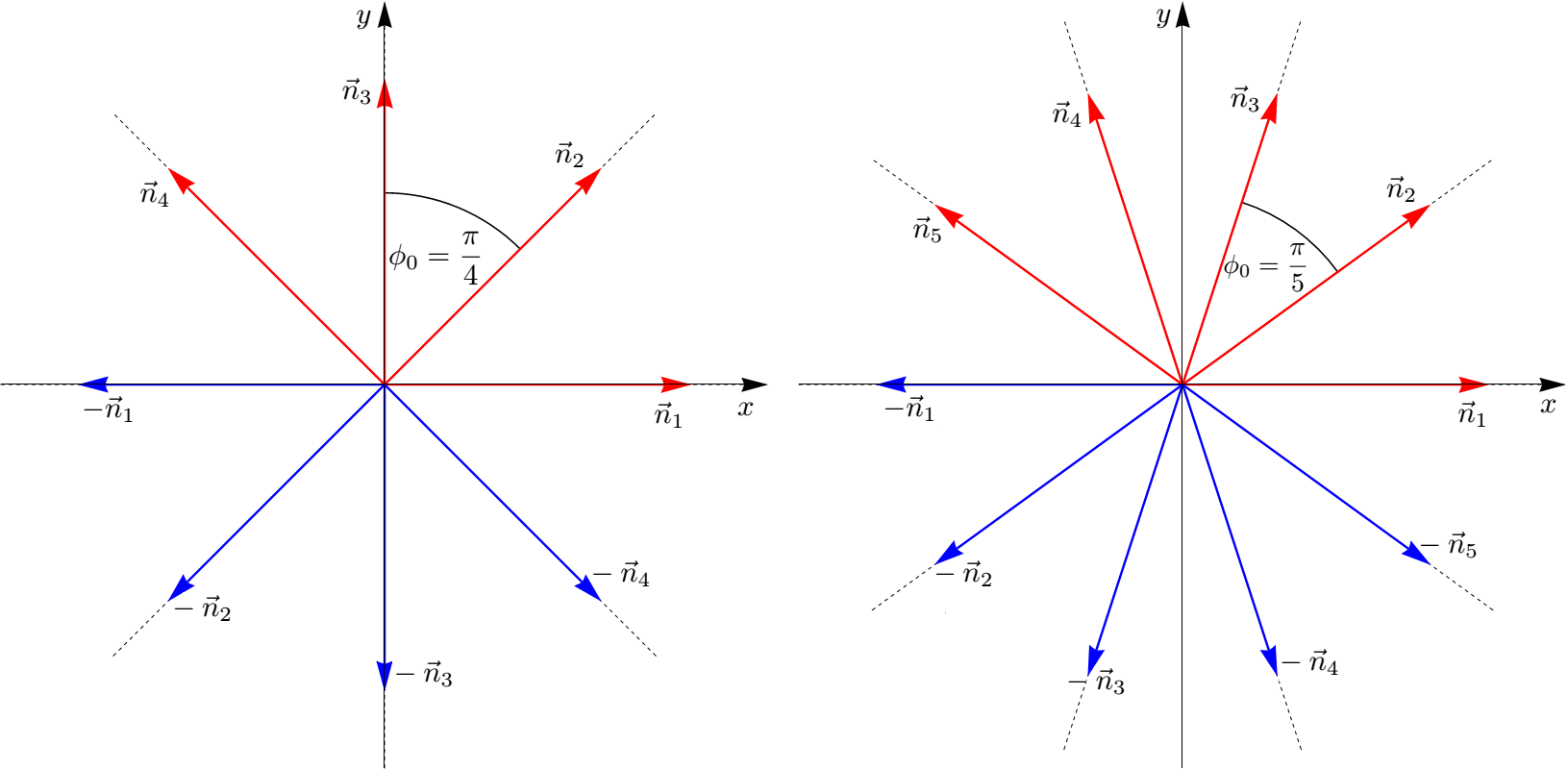}
	\caption{Planar symmetric POVMs for $N=4$ (left) and $N=5$ (right).}
	\label{plansymm}
\end{figure*}
\begin{remark}\label{rem1}\upshape Note that the set $\{\vec{e}(\vec{x})\}_{\vec{x}}$:
\begin{itemize}
\item coincides with the set $\{\pm\vec{n}_k\}$ for odd $N$;
\item contains vectors that bisect the angles between the successive $\{\pm\vec{n}_k\}$ for even $N$ (see Fig.~\ref{zeleni})
\end{itemize}
\begin{figure}[htb!]
\includegraphics[scale=0.2375]{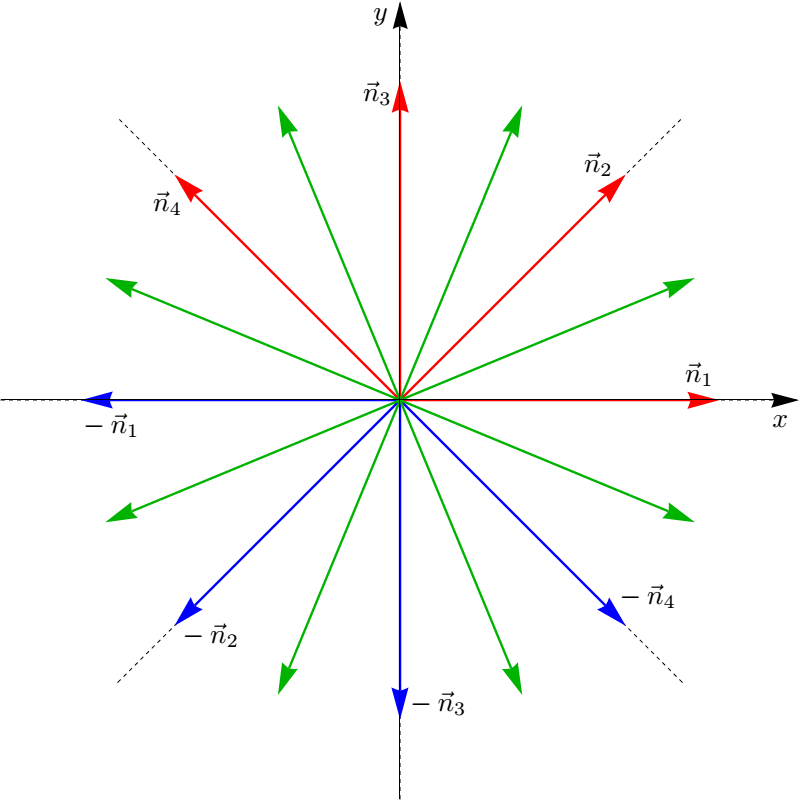}
\caption{Green arrows represent the set $\{\vec{e}(\vec{x})\}_{\vec{x}}$ for $N=4$, a representative example for even $N$. The first `green arrow' above the (positive) $x$ axis is assigned to the outcome $\vec{x}=(+1,+1,+1,-1)$. Going counter-clockwise, the next one is assigned to $\vec{x}=(+1,+1,+1,+1)$ and so on. For odd $N$, the  correspondence between $\vec{e}(\vec{x})$ and $\vec{x}$ works in a similar way.}
\label{zeleni}
\end{figure}
\label{eNset}
\end{remark}
\begin{remark}\label{remcyc}\upshape Note that all of the outcomes corresponding to non-zero effects of the joint POVM correspond to strings of length $N$ arranged in a cycle of consecutive `$+1$'s concatenated with consecutive `$-1$'s (Fig.~\ref{ciklus}). In other words, the outcome strings associated with non-zero effects are of the form \begin{align}
&\vec{x}=(\underbrace{+1,+1,\ldots,+1}_{p\text{ `$+1$'s }},\underbrace{-1,-1,\ldots,-1}_{q\text{ `$-1$'s }}) \text{ or }\nonumber\\
&\vec{x}=(\underbrace{-1,-1,\ldots,-1}_{p\text{ `$-1$'s }},\underbrace{+1,+1,\ldots,+1}_{q\text{ `$+1$'s }}),\nonumber\\
&\textrm{where } p,q=\overline{0,N},\text{ }p+q=N.
\end{align}
\end{remark}
\begin{figure}[htb!]
\centering
\includegraphics[scale=0.2]{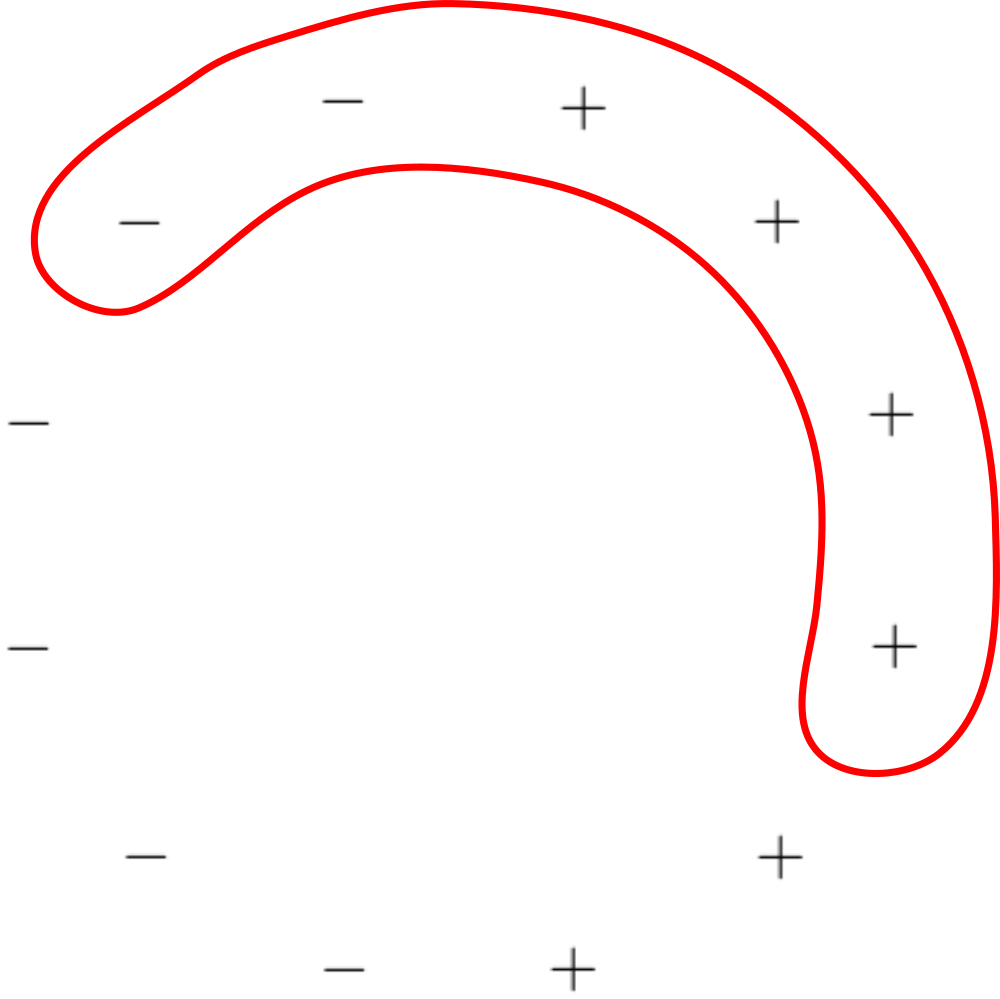}
\caption{For example, for $N=6$ we have a cycle of $6$ pluses concatenated to $6$ minuses. In this figure, the red loop identifies an outcome that has a non-zero effect associated with it. Reading counter-clockwise, this outcome is $\vec{x}=(+1,+1,+1,+1,-1,-1)$. All the other outcomes associated to non-zero effects are obtained in a similar way, i.e., by identifying, counter-clockwise, a string of six $\pm1$ in the figure.}
\label{ciklus}
\end{figure}

\subsubsection{Geometric equivalence of planar symmetric POVMs}\label{subsec2_3}
We discussed the geometric equivalence of general sets of qubit POVMs in Section \ref{subsec2_1}. One of the results presented there is Corollary \ref{GEbinary} which states that two sets of binary qubit POVMs are geometrically equivalent (and therefore exhibit the same joint measurability structure) if and only if the lines determined by their Bloch vectors are related by some orthogonal transformation $O\in \mathrm{O}(3)$. In this section we explore the consequences of this result for the joint measurability properties of planar symmetric POVMs. Considering their Bloch lines, we see that certain subsets are certainly geometrically equivalent. For example, if we act with an anti-clockwise rotation of $\frac{\pi}{N}$ on the set $\{E_1,E_2\}$, we obtain the set $\{E_2,E_3\}$ (see Fig. \ref{rot123}), and these two subsets are thus geometrically equivalent.
\begin{figure}[htb!]
\centering
\includegraphics[scale=0.41]{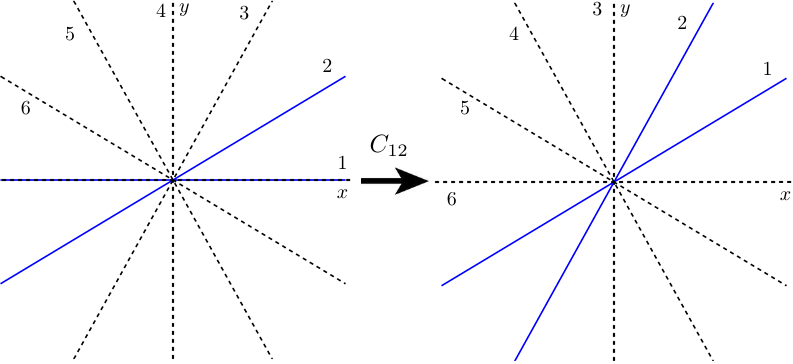}
\caption{If we act with $C_{2N}$, the lines that correspond to $E_1$ and $E_2$ take the place previously held by the lines of $E_2$ and $E_3$. (We have used $N=6$ in this figure.)}
\label{rot123}
\end{figure} 
A general question then arises: what is the symmetry group of the Bloch lines of planar symmetric POVMs? Since they are coplanar, we expect this to be a discrete group that is a subgroup of linear isometries of the plane, namely, $\mathrm{O}(2)$.\footnote{$\mathrm{O}(2)$ is itself a subgroup of $\mathrm{O}(3)$.} To see what the symmetry group is, consider a regular $2N$-gon overlaid on the Bloch lines as in Fig. \ref{sl56}. This makes it clear that the group we are searching for is the group of symmetries of the diagonals of this $2N$-gon. If the opposite ends of the $2N$-gon are labelled differently, this group is just the dihedral group $D_{2N}$\footnote{We adopt so called geometrical notation where the index denotes the number of vertices of the regular polygon and not the order of the group.} of order $4N$ (see Refs. \cite{Baumslag68} and \cite{Rotman12}). However, our Bloch lines do not have a sense of orientation since the orientation of the particular Bloch vector along them does not matter. Therefore the actual symmetry group is the symmetry group of a regular $2N$-gon with opposite vertices identified which we will denote by $S_{N}$. This group is a quotient group of $D_{2N}$ and the group $H_{2N}$ of symmetries of the regular $2N$-gon that sends every vertex either to itself or to its opposite vertex i.e.,
\begin{equation}
S_N=\frac{D_{2N}}{H_{2N}}.
\end{equation}
\begin{figure}[htb!]
\centering
\includegraphics[scale=0.431]{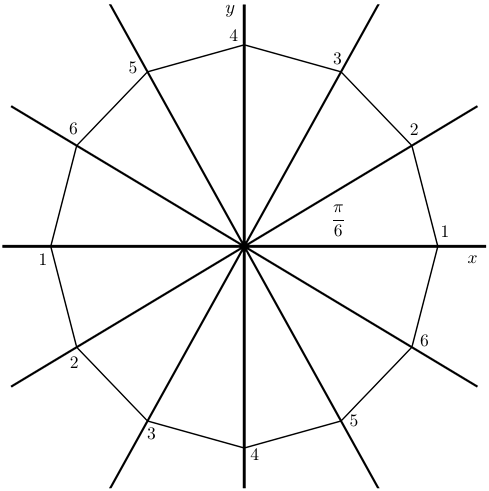}
\caption{Bloch lines of planar symmetric POVMs for $N=6$. A regular 12-gon is placed over them and its opposite vertices are identified.}
\label{sl56}
\end{figure}

We now recall of some basic definitions from group theory, namely those of a normal subgroup and semidirect product.
\begin{definition}[Normal Subgroup] Let $G$ be a group with the identity element {\rm id} and $N$ be its subgroup. We say that $N$ is a normal or invariant subgroup of $G$, which is denoted by $N\triangleleft G$, if for each $g\in G$ and for each $n\in N$ the product $gng^{-1}$ (which is called the conjugation of $n$ by $g$) belongs to $N$.
\end{definition}
\begin{theorem}[Equivalent definitions of Semidirect product] Let $G$ be a group with the identity element {\em id}, $H$ its subgroup and $N$ its normal subgroup. The following statements are equivalent:
	\begin{enumerate}
		\item for each $g\in G$, there exist $n\in N$ and $h\in H$ such that $g=nh$, where $N\cap H=\{\rm id\}$;
		\item for each $g\in G$ there exist unique $n\in N$ and $h\in H$ such that $g=nh$;
		\item for each $g\in G$ there exist unique $n\in N$ and $h\in H$ such that $g=hn$;
		\item $G$ is a semidirect product of $N$ and $H$ denoted by $G=N\ltimes H=H\rtimes N$.
	\end{enumerate}
\end{theorem}
For more properties of the semidirect product see Ref.~\cite{Rose09}, specifically Definition~7.14 and Lemma~7.15.
\begin{lemma} \label{H2N} Group $H_{2N}$ has the following properties
	\begin{enumerate}
		\item for $N=2$ it holds 
		$$H_{4}=\left<C_{2}\right>\ltimes \left<\sigma_x\right>$$ 
		where $\left<C_2\right>$ is the second order cyclic group generated by the positive rotation by $\pi$, denoted $C_2$\footnote{In general, transformation $C_n$ is the positive rotation by $2\pi/n$.}, and $\left<\sigma_x\right>$ is the second order cyclic group generated by the reflection on the $x$ axis, defined by two vertices of the polygon, here denoted by $\sigma_x$;
		\item for $N>2$ it holds
		$$H_{2N}=\left<C_2\right>.$$
	\end{enumerate}
\begin{proof}
	The symmetry transformation that preserves three vertices or two non-opposite vertices i.e. that sends each of them to itself, must preserve all vertices so it is just the identity transformation, id (cf.~Theorem 2.2 in Ref.~\cite{Johnson12})\footnote{Theorem 2.2 in Ref.~\cite{Johnson12} reads: ``Any isometry of $\mathbb{R}^2$ is determined by its effect on any three non-collinear points."}. 
	The symmetry transformation that preserves two opposite vertices (i.e. sends each of them to itself) is  
	a reflection about the diagonal connecting them.
	It is easy to see that this transformation satisfies the properties required by $H_{2N}$ (namely, that it sends every vertex either to itself or to its opposite vertex) only for $N=2$ (see Fig. \ref{sigmax}). The symmetry transformation that preserves no vertices and satisfies the requirements of $H_{2N}$ is the one that sends each vertex to its opposite one i.e. the central symmetry on the origin, and in the planar case this is the same as the positive (or counter-clockewise) rotation by $\pi$, namely the operation $C_2$. 
\begin{figure}[H]
		\centering
		\includegraphics[scale=0.41]{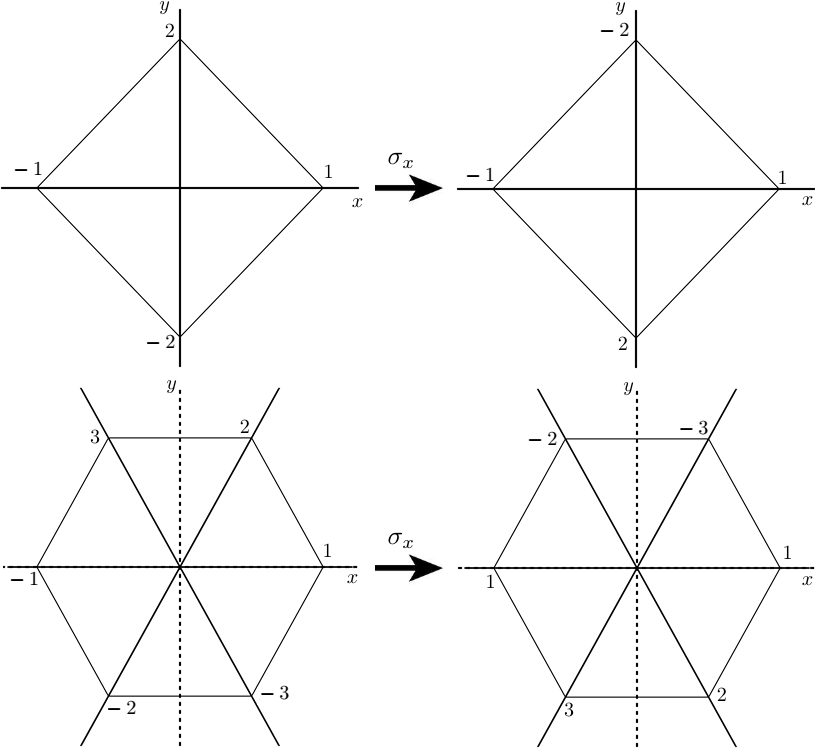}
		\caption{Reflection about the $x$-axis preserves the vertices $1$ and $-1$ and swaps the opposite vertices $2$ and $-2$ when $N=2$. When $N=3$ the same transformation sends $2$ to $-3$ etc. which is forbidden by the requirements of $H_{2N}$.}
		\label{sigmax}
	\end{figure}
	So, only in the case $N=2$ we have that 
	\begin{align}
	H_4&=\left<C_2\right>\ltimes\left<\sigma_x\right>\nonumber\\
	&=\{C_2,\sigma_x,\sigma_y=C_2\sigma_x=\sigma_x C_2,\sigma_x^2=C_2^2=\text{id}\},
	\end{align}
	and in the case $N>2$ we have 
	\begin{equation}
	H_{2N}=\left<C_2\right>=\{C_2,C_2^2=\mathrm{id}\}.
	\end{equation}
\end{proof}
\end{lemma}

\begin{proposition}
Group $S_{N}$ has the following properties
\begin{enumerate}
	\item For $N=2$ it is isomorphic to $\left<C_2\right>$ i.e. 
	$$S_{2}\cong\left<C_2\right>.$$
	\item For $N>2$ it is isomorphic to the dihedral group $D_N$ i.e.
	\begin{equation}
	S_N\cong D_N.
	\end{equation}
\end{enumerate}
\begin{proof}
	Remember that the dihedral group $D_{N}$ is the semidirect product of $\left<C_{N}\right>$ and $\left<\sigma_x\right>$ (we always take $x$ axis to be defined by two vertices of the polygon). Therefore, for $N=2$, 
	\begin{equation}
	S_2=\frac{D_4}{\left<C_2\right>\ltimes\left<\sigma_x\right>}=\frac{\left<C_4\right>\ltimes\left<\sigma_x\right>}{\left<C_2\right>\ltimes\left<\sigma_x\right>}\cong \left<C_2\right>.
	\end{equation}
	For $N>2$ we have 
	\begin{equation}
	S_N=\frac{D_{2N}}{\left<C_2\right>}=\frac{\left<C_{2N}\right>\ltimes\left<\sigma_x\right>}{\left<C_2\right>}\cong\left<C_N\right>\ltimes\left<\sigma_x\right>\cong D_N.
	\end{equation}
	Here we have used the following facts: that each cyclic subgroup is always a normal subgroup; that two cyclic groups of the same order are isomorphic to each other; that a cyclic group of the order $n$ generated by $r$ has cyclic subgroups of the order $p$ generated by $r^{n/p}$ if $p$ divides $n$ and we applied the second and the third isomorphism Theorem (cf. Theorems~4.22 and 4.23 in Ref.~\cite{Rotman12}).
\end{proof}
\end{proposition}
Let us explain this qualitatively for the case of $N>2$. Group $D_{2N}$ consists of the identity transformation, positive (i.e., counter-clockwise) rotations through $p\pi/N$, where $p=\overline{1,N-1}$, reflections about the lines connecting opposite vertices and reflections about the lines connecting the mid-points of opposite sides of the polygon (all the symmetries of the regular $2N$-gon). However, if we identify each pair of opposite vertices, we have that the reflections about mutually orthogonal lines give the same result and therefore they are identified. Also, the rotation $C_{2N}^p$ (counter-clockwise, by an angle $p\pi/N$), and $C_{2N}^{N+p}$ (counter-clockwise, by an angle $\pi+p\pi/N$) produce the same result so they have to be identified as well. Hence, the group $S_N$ consists of the following rotations: $C_{2N},C_{2N}^2,\ldots, C_{2N}^N=\mathrm{id}$\footnote{In general, $C^N_{2N}=C_2$. However, when the opposite vertices are identified, $C_2=\mathrm{id}$.} and the reflections in $D_{2N}$, where reflections about any pair of orthogonal lines are identified. This group is of order $2N$ (since it is isomorphic to $D_{N}$).
\begin{corollary}
\label{geqplansym}
Let $s$ be a set of $N$ planar symmetric POVMs. For any two subsets of $s$, say $r_1$ and $r_2$, that are related by a transformation $O\in S_N$ (i.e., $r_2=O\mathbb{ Perm}r_1$), where
\begin{align}
S_N=\big\{&C_{2N},C_{2N}^2,\ldots,C_{2N}^N=C_2=\mathrm{id},\nonumber\\
&\sigma_{0}=\sigma_{x},\sigma_{\frac{\pi}{2N}},\sigma_{\frac{2\pi}{2N}},\sigma_{\frac{3\pi}{2N}},\ldots\sigma_{\frac{(N-1)\pi}{2N}}\big\}
\end{align}
and $\sigma_\phi$ is a reflection about the line making the angle of $\phi$ with the positive $x$-axis (cf. Fig \ref{Ctablesl}), we have that $r_1$ and $r_2$ are geometrically equivalent and hence exhibit the same joint measurability structure. If a joint POVM for $r_1$ is $G^{r_1}$ then the corresponding joint POVM for $r_2$ is given by $G^{r_2}=O\mathbb{Perm}G^{r_1}$, where $\mathbb{Perm}$ denotes the relabelling of outcomes on $r_2$ that may be necessary for the orientations of the Bloch vectors of $r_1$, after being acted upon by $O$, to match those of $r_2$.
\begin{proof}
This follows from the fact that the discrete group $S_N$ is the symmetry group of planar symmetric POVMs and Corollary~\ref{GEbinary}.
\end{proof} 
\end{corollary}
\begin{definition}\label{classdef}
Geometrically equivalent sets of planar symmetric POVMs comprise an equivalence class. If the set $\{E_1,\ldots\}$ is in one of the elements of such an equivalence class, the whole equivalence class is denoted by $\left[E_1,\ldots\right]$.\footnote{Notice that there is still some ambiguity in this notation. For example, take $N=6$. The sets $\{E_1,E_2,E_5\}$ and $\{E_1,E_3,E_4\}$ are geometrically equivalent, so we may denote their equivalence class by either of them.}
\end{definition}
\begin{example}\upshape
We find the Cayley multiplication table of the group $S_4$ by applying transformations successively. This group has the following elements:
\begin{equation}
S_4=\{\mathrm{id},C_8,C_8^2,C_8^3,\sigma_x,\sigma_{\frac{\pi}{8}},\sigma_{\frac{\pi}{4}},\sigma_{\frac{3\pi}{8}}\},
\end{equation} 
where $\sigma$ denotes the reflections on the lines shown in Fig. \ref{Ctablesl}
\begin{figure}[htb!]
\centering
\includegraphics[scale=0.55]{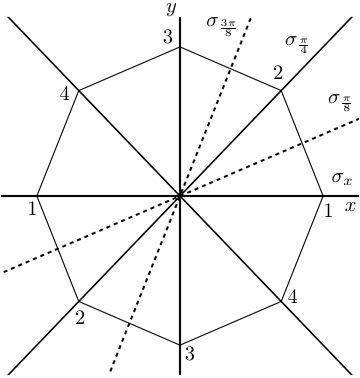}
\caption{Bloch lines of $N=4$ planar symmetric POVMs with the reflection symmetry lines denoted explicitly.}
\label{Ctablesl}
\end{figure}
\begin{table}[H]
\large
\centering
\begin{tabular}{|c|c|c|c|c|c|c|c|c|}\hline
 $S_4$& id & $C_8$ & $C_8^2$ & $C_8^3$ & $\sigma_x$ & $\sigma_{\frac{\pi}{8}}$ & $\sigma_{\frac{\pi}{4}}$ & $\sigma_{\frac{3\pi}{8}}$ \\ \hline
id & id & $C_8$ & $C_8^2$ & $C_8^3$ & $\sigma_x$ & $\sigma_{\frac{\pi}{8}}$ & $\sigma_{\frac{\pi}{4}}$ & $\sigma_{\frac{3\pi}{8}}$ \\ \hline
$C_8$ & $C_8$ & $C_8^2$ & $C_8^3$ & id & $\sigma_{\frac{\pi}{8}}$ & $\sigma_{\frac{\pi}{4}}$ & $\sigma_{\frac{3\pi}{8}}$ & $\sigma_x$ \\ \hline
$C_8^2$ & $C_8^2$ & $C_8^3$ & id & $C_8$ & $\sigma_{\frac{\pi}{4}}$ & $\sigma_{\frac{3\pi}{8}}$ & $\sigma_x$ & $\sigma_{\frac{\pi}{8}}$ \\ \hline
$C_8^3$ & $C_8^3$ & id & $C_8$ & $C_8^2$ & $\sigma_{\frac{3\pi}{8}}$ & $\sigma_x$ & $\sigma_{\frac{\pi}{8}}$ & $\sigma_{\frac{\pi}{4}}$ \\ \hline
$\sigma_x$ & $\sigma_x$ & $\sigma_{\frac{3\pi}{8}}$ & $\sigma_{\frac{\pi}{4}}$ & $\sigma_{\frac{\pi}{8}}$ & id & $C_8^3$ & $C_8^2$ & $C_8$ \\ \hline
$\sigma_{\frac{\pi}{8}}$ & $\sigma_{\frac{\pi}{8}}$ & $\sigma_x$ & $\sigma_{\frac{3\pi}{8}}$ & $\sigma_{\frac{\pi}{4}}$ & $C_8$ & id & $C_8^3$ & $C_8^2$ \\ \hline
$\sigma_{\frac{\pi}{4}}$ & $\sigma_{\frac{\pi}{4}}$ & $\sigma_{\frac{\pi}{8}}$ & $\sigma_x$ & $\sigma_{\frac{3\pi}{8}}$ & $C_8^2$ & $C_8$ & id & $C_8^3$ \\ \hline
$\sigma_{\frac{3\pi}{8}}$ & $\sigma_{\frac{3\pi}{8}}$ & $\sigma_{\frac{\pi}{4}}$ & $\sigma_{\frac{\pi}{8}}$ & $\sigma_x$ & $C_8^3$ & $C_8^2$ & $C_8$ & id \\ \hline
\end{tabular}
\label{Ctable}
\caption{Cayley's table for $S_4$. We first apply the column operation and then the row operation.}
\end{table}
\end{example}

\section{Marginal surgery: a method for tweaking joint POVM\lowercase{s} to realize new joint measurability structures}\label{sec3}
Given some joint POVM with the effects $E^{s_1}(\vec{x})$ for some set $s_1$ of $N$ POVMs, we can obtain a joint POVM for any subset $s'_1$ (with $M<N$ POVMs, say) of $s_1$ by marginalizing the effects $E^{s_1}(\vec{x})$ over the outcomes of POVMs from $s_1\backslash s'_1$.
In what follows, we will introduce a technique that tweaks the marginalized joint POVM on $s'_1$ and uses that to obtain an {\em incompatible} set $s_2$ (of same cardinality as $s_1$) and its compatible subset $s'_2$ (of same cardinality as $s'_1$). Thus, using this technique, we can go from a set of $N$ compatible POVMs to a set of $N$ POVMs that are incompatible but such that a subset of $M$ of them is still compatible.

As an example, consider a set $s_1$ of unbiased qubit POVMs with the same purity $\eta$ that are compatible if and only if $\eta\in(0,\eta_{\rm max}]$ and let  $\{E^{s_1}(\vec{x})\}_{\vec{x}\in\{\pm 1\}^N}$ be a  joint POVM for them, also valid for any $\eta\in(0,\eta_{\rm max}]$. Then we choose the set $s'_1\subset s_1$ of $M<N$ POVMs and marginalize to obtain $\{E^{s'_1}(\vec{x}')\}_{\vec{x}'\in\{\pm1\}^M}$, a joint POVM for $s'_1$. This joint POVM for $s'_1$ is thus valid for $\eta\in(0,\eta_{\rm max}]$. Now we tweak the effects $E^{s'_1}(\vec{x}')$ to get new effects $E^{s'_2}(\vec{x}')$ which still constitute a joint measurement on $s'_1$ but this new $M$-wise joint POVM allows for a broader range of purity $\eta\in(0,\eta_{\rm MAX}]$, where $\eta_{\rm MAX}>\eta_{\rm max}$, and consequently allows us to obtain a set $s'_2$ (similar to set $s'_1$ but with $\eta\notin(0,\eta_{\rm max}]$) of $M$ compatible POVMs (by marginalization) with $\eta\in (\eta_{\rm max}, \eta_{\rm MAX}]$.  It thus also allows us to obtain an {\em incompatible} set $s_2$ (of $N$ POVMs) similar to the compatible set $s_1$ but with $\eta\in (\eta_{\rm max}, \eta_{\rm MAX}]$ and containing the compatible subset $s'_2$. 

Thus, by choosing $\eta\in(\eta_{\rm max},\eta_{\rm MAX}]$ we can realize a joint measurability structure where the $N$ POVMs in $s_2$ are incompatible but the $M$ POVMs in its subset $s'_2$ are compatible. This technique can be applied to obtain many new joint measurability structures from a set of compatible POVMs, as we will show below. Since the technique proceeds by tweaking marginal joint POVMs obtained by marginalizing a bigger joint POVM, we term this technique {\em marginal surgery}. We now proceed to illustrate how this method works in practice by obtaining $N$-cycle and $N$-Specker compatibility scenarios using qubit POVMs. As a consequence of this, we will also be able to construct arbitrary joint measurability structures using these qubit POVMs, improving upon the (rather weak) dimension bounds for realizing arbitrary joint measurability structures obtained in Ref.~\cite{KHF14}.
\subsection{Marginal surgery on any pair of compatible POVMs:  constructing $N$-cycle scenarios on a qubit}\label{subsec3_1}
We will now apply marginal surgery on the joint POVM of any pair of POVMs out of $N$ planar symmetric compatible POVMs, $\{E_1,E_2,\dots,E_N\}$, and use this to construct $N$-cycle scenarios on a qubit.

Consider two of the POVMs, $E_1$ and $E_{k}$, where $1<k\leq N$. We will also denote $p\equiv k-1$, so that the angular separation between the lines defined by the considered POVMs is $\phi_{1,k}\equiv\frac{(k-1)\pi}{N}=\frac{p\pi}{N}$. Now consider the joint POVM effects $G(\vec{x}), \vec{x}\in\{\pm1\}^N,$ for all $N$ POVMs from Eq.~\eqref{optN}. To get the pairwise joint measurement $E^{(1,k)}$ we marginalize $G$ over all of the outcomes of POVMs $E_i$, $i\notin\{1,k\}$:
\begin{equation}
E^{(1,k)}(x_1,x_k)=\sum_{\vec{x}\in\{\pm1\}^N}^{x_1,x_k}G(\vec{x}),
\end{equation}
where we have chosen to indicate the outcomes that are held fixed (while all the others are summed over) on top of the summation sign.

We will now exploit the fact that a pairwise joint POVM here is completely determined by knowing just one effect, say $E^{(1,k)}(+1,+1)$, with the others given by
\begin{align}
&E^{(1,k)}(+1,-1)=E_1(+1)-E^{(1,k)}(+1,+1),\nonumber\\
&E^{(1,k)}(-1,+1)=E_k(+1)-E^{(1,k)}(+1,+1),\nonumber\\
&E^{(1,k)}(-1,-1)=I-E^{(1,k)}(+1,+1)-E^{(1,k)}(+1,-1)\nonumber\\
&\hphantom{E^{(1,k)}(-1,-1)=}-E^{(1,k)}(-1,+1),
\label{determin2}
\end{align}
so we only have to focus on finding $E^{(1,k)}(+1,+1)$:
\begin{align}
E^{(1,k)}(+1,+1)&=\sum_{\vec{x}\in\{\pm1\}^N}^{x_1=x_k=+1}G(\vec{x})\nonumber\\
&=\sum_{\vec{x}\in\{\pm1\}^N}^{x_1=x_k=+1}\frac{1}{2N}I+\sum_{\vec{x}\in\{\pm1\}^N}^{x_1=x_k=+1}\frac{\gamma}{2N}\vec{e}(\vec{x})\cdot\vec{\sigma}.
\end{align}
Paying attention to Remark \ref{remcyc}, consider the outcomes $\vec{x}=(x_1,\ldots,x_N)$ to which non-zero effects of the joint POVM $\{G(\vec{x})\}_{\vec{x}\in\{\pm1\}^N}$ are assigned. Of these, those which have $+1$ at the first index and the $k$-th index ($x_1=x_k=+1$) must be of the type 
\begin{align}
(x_1=+1,x_2=+1,&\ldots,x_k=+1,x_{k+1},\ldots,x_N)\hphantom{a}\text{ where}\nonumber\\
(x_{k+1},\ldots,x_N)\in\{&(-1,\ldots,-1),(+1,-1,\ldots,-1),\ldots,\nonumber\\
&(+1,\ldots,+1,-1),(+1,\ldots,+1)\}
\end{align}  and
there are exactly $N-k+1=N-p$ of them. This lets us write
\begin{equation}
E^{(1,k)}(+1,+1)=\frac{N-p}{2N}I+\frac{\gamma}{2N}\vec{\sigma}\cdot\sum_{\vec{x}\in\{\pm1\}^N}^{x_1,\ldots,x_k=+1}\vec{e}(\vec{x}).
\end{equation}
From the Theorem \ref{gluslov} we see that the set $\Big\{\vec{e}(\vec{x})|\vec{x}\in\{\pm1\}^N\Big\}$ is equal to 
$\Big\{\big(\cos\frac{(i-1)\pi}{N},\sin\frac{(i-1)\pi}{N},0\big)\big|i=\overline{-N+1,N}\Big\}$ in the case of odd $N$ and $\Big\{\big(\cos\frac{(2i-1)\pi}{2N},\sin\frac{(2i-1)\pi}{2N},0\big)\big|i=\overline{-N+1,N}\Big\}$ in the case of even $N$. We can convert the sum $\sum_{\vec{x}\in\{\pm1\}^N}^{x_1,\ldots,x_k=+1}\vec{e}(\vec{x})$ into the sum over some index $i$ from the conditions that
\begin{subequations}
\begin{align}
&\vec{n}_1\cdot\left(\cos\frac{(i-1)\pi}{N},\sin\frac{(i-1)\pi}{N},0\right)>0\textrm{ and}\\
&\vec{n}_k\cdot\left(\cos\frac{(i-1)\pi}{N},\sin\frac{(i-1)\pi}{N},0\right)>0
\end{align}
\end{subequations}
for odd $N$, while for even $N$
\begin{subequations}
\begin{align}
&\vec{n}_1\cdot\left(\cos\frac{(2i-1)\pi}{2N},\sin\frac{(2i-1)\pi}{2N},0\right)>0\textrm{ and}\\
&\vec{n}_k\cdot\left(\cos\frac{(2i-1)\pi}{2N},\sin\frac{(2i-1)\pi}{2N},0\right)>0,
\end{align}
\end{subequations}
which arise from the way of assigning effects to outcomes described in Theorem $\ref{gluslov}$.
In both cases, we get (see Appendix \ref{AppA}, Lemma \ref{id1})
\begin{equation}
\left\lceil p+1-\frac{N}{2}\right\rceil\leq i\leq \left\lfloor\frac{N+1}{2}\right\rfloor.
\end{equation}
This enables us to write
\begin{align}
\sum_{\vec{x}\in\{\pm1\}^N}^{x_1,\ldots,x_k=+1}\vec{e}(\vec{x})=\begin{cases}
\sum\limits_{i=p+\frac{3-N}{2}}^{\frac{N+1}{2}}\left(\cos\frac{(i-1)\pi}{N},\sin\frac{(i-1)\pi}{N},0\right)\\ \hspace{2cm
} \qquad\qquad\text{for odd } N,\\
\sum\limits_{i=p+1-\frac{N}{2}}^{\frac{N}{2}}\left(\cos\frac{(2i-1)\pi}{2N},\sin\frac{(2i-1)\pi}{2N},0\right)\\ \hspace{2cm}\qquad\qquad\text{for even } N.
\end{cases}
\end{align}
which in combination with results of Appendix \ref{AppA}, Corollary \ref{appAcorr} of Lemma \ref{id2}, Eq.~\eqref{corl2} gives
\begin{align}
\sum_{\vec{x}\in\{\pm1\}}^{x_1,\ldots,x_k=+1}\vec{e}(\vec{x})=\frac{\cos\frac{p\pi}{2N}}{\sin\frac{\pi}{2N}}\vec{s},\text{ }\vec{s}\equiv\left(\cos\frac{p\pi}{2N},\sin\frac{p\pi}{2N},0\right)
\label{jzid2}
\end{align}
where $\vec{s}$ is a unit vector. This completely specifies $E^{(1,k)}(+1,+1)$ which combined with Eq.~\eqref{determin2} yields (recalling that $p=k-1$)
\begin{align}
E^{(1,k)}(\pm1,\pm1)&=\frac{N-p}{2N}I\pm\gamma\frac{\cos\frac{p\pi}{2N}}{2N\sin\frac{\pi}{2N}}\vec{s}\cdot\vec{\sigma};\nonumber\\
E^{(1,k)}(\pm1,\mp1)&=\frac{p}{2N}I\pm\gamma\frac{\sin\frac{p\pi}{2N}}{2N\sin\frac{\pi}{2N}}\vec{t}\cdot\vec{\sigma};
\end{align}
where $\vec{t}$ is found to be 
\begin{equation}
\vec{t}=\left(\sin\frac{p\pi}{2N},-\cos\frac{p\pi}{2N},0\right).
\end{equation}
This motivates the following ansatz for the general form of the joint POVM of $E_1$ and $E_k$: 
\begin{align}
G^{(1,k)}(\pm1,\pm1)&=\beta I \pm \mu \frac{\cos\frac{p\pi}{2N}}{\sin\frac{\pi}{2N}}\vec{s}\cdot\vec{\sigma},\nonumber\\
G^{(1,k)}(\pm1,\mp1)&=\alpha I \pm \mu\frac{\sin\frac{p\pi}{2N}}{\sin\frac{\pi}{2N}}\vec{t}\cdot\vec{\sigma}.
\end{align}
Requiring correct marginalization for, say,
$E_1(+1)=G^{(1,k)}(+1,+1)+G^{(1,k)}(+1,-1)$,
gives
\begin{equation}
\alpha+\beta=\frac{1}{2};\quad\frac{\mu}{\sin\frac{\pi}{2N}}=\frac{1}{2}\eta.
\label{margarb2}
\end{equation}
Requiring positivity of $G^{(1,k)}$ we have that
\begin{equation}
\mu\frac{\sin\frac{p\pi}{2N}}{\sin\frac{\pi}{2N}}\leq\alpha,\quad \mu\frac{\cos\frac{p\pi}{2N}}{\sin\frac{\pi}{2N}}\leq\beta.
\label{posarb2}
\end{equation}
Eliminating everything but $\eta$ from the linear system consisting of Eqs.~\eqref{margarb2} and \eqref{posarb2} we get
\begin{equation}
\eta\leq\frac{1}{\sin\frac{p\pi}{2N}+\cos\frac{p\pi}{2N}}.
\label{boundarb2}
\end{equation}
This means that for every $\displaystyle\eta\in\left(0,\frac{1}{\sin\frac{p\pi}{2N}+\cos\frac{p\pi}{2N}}\right]$ there are suitable $\alpha$, $\beta$ and $\mu$ such that $G^{(1,k)}$ is a valid joint POVM for $E_1(x_1)=\frac{1}{2}(I\pm x_1\eta\vec{n}_1\cdot\vec{\sigma})$ and $E_k(x_k)=\frac{1}{2}(I\pm x_k\eta\vec{n}_k\cdot\vec{\sigma})$. We propose one particular choice that is consistent with Eqs.~\eqref{margarb2} and \eqref{posarb2}:
\begin{equation}
\alpha=\frac{1}{2}\eta\sin\frac{p\pi}{2N},\quad \beta=\frac{1}{2}\left(1-\eta\sin\frac{p\pi}{2N}\right),
\end{equation} 
which gives
\begin{subequations} 
\begin{align}
G^{(1,k)}(\pm1,\pm1)&=\frac{1}{2}\left(1-\eta\sin\frac{p\pi}{2N}\right)\pm \frac{1}{2}\eta\cos\frac{p\pi}{2N}\vec{s}\cdot\vec{\sigma},\\ G^{(1,k)}(\pm1,\mp1)&=\frac{1}{2}\eta\sin\frac{p\pi}{2N}\Big(I\pm\vec{t}\cdot\vec{\sigma}\Big),
\end{align}
\label{arb21got}
\end{subequations}
Note that this construction is such that when $\eta$ saturates the upper bound of Eq.~\eqref{boundarb2}, $G^{(1,k)}$ is a rank one POVM.
\begin{figure}[htb!]
\includegraphics[scale=0.29]{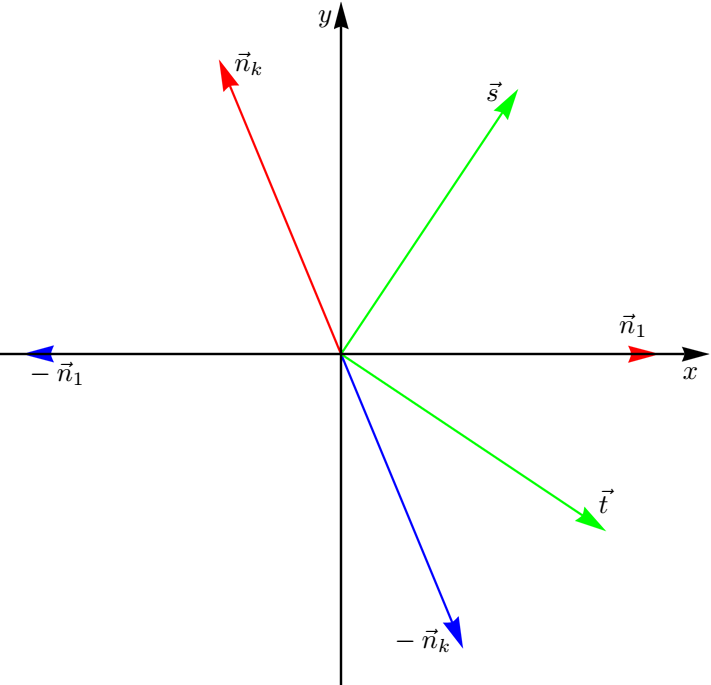}
\caption{Situation of $\vec{s}$ and $\vec{t}$ with respect to $\vec{n}_1$ and $\vec{n}_k$}
\label{figarb2}
\end{figure}
From Corollary \ref{Leta2}, we see that the sufficient condition proven by this construction for the joint measurability of $E_1$ and $E_k$ is also necessary, so the joint POVM $G^{(1,k)}$ that we have constructed is optimal. That is, it applies to the whole range of $\eta$ that is necessary and sufficient for joint measurability of $E_1$ and $E_k$.
Vectors $\vec{s}$ and $\vec{t}$ are situated so that $\vec{s}$ is the unit vector bisecting the angle between the lines determined by $\vec{n}_1$ and $\vec{n}_k$, while the vector $\vec{t}$ is the unit vector perpendicular to $\vec{s}$ that makes an acute angle with $\vec{n}_1$ (see Fig.~\ref{figarb2}).
For convenience we have chosen the set $\{E_1,E_k\}$. Acting on this set by the symmetry operation $C_{2N}^{k_1-1}$, i.e., by the positive rotation through $\frac{(k_1-1)\pi}{N}$, we can obtain every other two element set, $\{E_{k_1},E_{k_1+p\mod N}\}$, $k_1>1$, that has the angular separation of $\frac{p\pi}{N}$ between the Bloch lines of its POVMs. By Corollary \ref{geqplansym}, all of these sets are geometrically equivalent to $\{E_1,E_{k}\}$ and, thus, compatible if and only if the set $\{E_1,E_k\}$ is compatible. Therefore, the condition of Eq.~\eqref{boundarb2} is sufficient (and necessary by Corollary \ref{Leta2}) for joint measurability of every pair $E_{k_1}$ and $E_{k_2}$, where $k_2=k_1+p\mod N$.
According to Corollary~\ref{geqplansym}, the joint POVM for this pair is given by $G^{(k_1,k_2)}=C_{2N}^{k_1-1}G^{(1,k)}$. In the particular case analysed in this section this reduces to rotating the vectors $\vec{s}$ and $\vec{t}$ through $\frac{(k_1-1)\pi}{N}$ in the counter-clockwise direction to preserve the geometrical situation in Fig.~\ref{figarb2}. Using Eqs.~\eqref{arb21got} we immediately find: 
\begin{subequations}
\begin{align}
&G^{(k_1,k_2)}(\pm1,\pm1)\nonumber\\
&=\frac{1}{2}\left(1-\eta\sin\frac{p\pi}{2N}\right)\pm \frac{1}{2}\eta\cos\frac{p\pi}{2N}\vec{s}(k_1,k_2)\cdot\vec{\sigma},\\
&\vec{s}(k_1,k_2)=\left(\cos\frac{(k_2+k_1-2)\pi}{2N},\sin\frac{(k_2+k_1-2)\pi}{2N},0\right),\\
&G^{(k_1,k_2)}(\pm1,\mp1)\nonumber\\
&=\frac{1}{2}\eta\sin\frac{p\pi}{2N}\Big(I\pm\vec{t}(k_1,k_2)\cdot\vec{\sigma}\Big),\\
&\vec{t}(k_1,k_2)=\left(\sin\frac{(k_2+k_1-2)\pi}{2N},-\cos\frac{(k_2+k_1-2)\pi}{2N},0\right).
\end{align}
\end{subequations}

\begin{proposition}\label{ncycletheorem}
 Let $s=\{E_1,E_2,\ldots,E_N\}$ be planar symmetric (qubit) POVMs. Then the joint measurability structure of $s$ is an $N$-cycle scenario (introduced in Example~\ref{NCycdefn}) if and only if
\begin{equation}\label{ncyclebounds}
\eta\in\left(\frac{1}{\sin\frac{\pi}{N}+\cos\frac{\pi}{N}},\frac{1}{\sin\frac{\pi}{2N}+\cos\frac{\pi}{2N}}\right].
\end{equation}
\begin{proof}
For any set of three adjacent POVMs $\{E_j,E_{j+1},E_{j+2}\}$ ($j=\overline{1,N}$ and addition modulo $N$), we have that the pairs $\{E_j,E_{j+1}\}$ and $\{E_{j+1},E_{j+2}\}$ form compatible subsets if and only if (following Corollary \ref{Leta2})
\begin{equation}\label{uppbd}
\eta\leq \frac{1}{\sin\frac{\pi}{2N}+\cos\frac{\pi}{2N}}.
\end{equation}
This gives us the necessary and sufficient condition for all the compatibility relations required in an $N$-cycle scenario, i.e., every POVM is compatible with its immediate neighbours if and only if the purity $\eta$ satisfies Eq.~\eqref{uppbd}.
What remains is to obtain the necessary and sufficient condition for the incompatibility of all the remaining pairs of POVMs so that the set of $N$ POVMs 
realizes an $N$-cycle scenario.

Note that the function $f(x)\equiv \frac{1}{\sin \frac{x}{2}+\cos \frac{x}{2}}$ is a monotonically decreasing function of $x$ for $x\in[0,\pi/2]$, going from $f(x=0)=1$ to $f(x=\pi/2)=\frac{1}{\sqrt{2}}$. On the other hand, $f(x)$ is monotonically increasing over $x\in[\pi/2,\pi]$, going from $f(x=\pi/2)=\frac{1}{\sqrt{2}}$ to $f(x=\pi)=1$ (see Fig. \ref{graphf}).
For any $j=\overline{1,N}$, denoting the angle between $\vec{n}_{j+1}$ and $\vec{n}_{l}$ (where $l\in\{1,2,\ldots,j-1,j+3,\ldots,N\}$) by $x\equiv \frac{|l-(j+1)|\pi}{N}$,\footnote{Recall that the addition in $j+1$ is modulo $N$.} so that $\vec{n}_{j+1}\cdot\vec{n_l}=\cos x$, we have that $x\geq\frac{2\pi}{N}$ for all $l$. Overall, we have  $x\in\left\{\frac{2\pi}{N},\frac{3\pi}{N},\ldots,\frac{(N-2)\pi}{N}\right\}$. Now, $E_{j+1}$ and $E_l$ are incompatible if and only if $\eta> \frac{1}{\sin \frac{x}{2}+\cos \frac{x}{2}}$. Since we need this to hold for all pairs $\{E_{j+1},E_l\}$ such that $x\geq\frac{2\pi}{N}$, 
the necessary and sufficient condition for obtaining all the incompatibility relations required in an $N$-cycle scenario is
\begin{equation}
\eta>\max_{x\in\{\frac{q\pi}{N}\}_{q=2}^{N-2}} \frac{1}{\sin \frac{x}{2}+\cos \frac{x}{2}} = \frac{1}{\sin\frac{\pi}{N}+\cos\frac{\pi}{N}}.
\end{equation}
which follows from properties of function $f(x)$ that we discussed and which are also illustrated in Fig.~\ref{graphf}.
\begin{figure}
\centering
\includegraphics[scale=0.27]{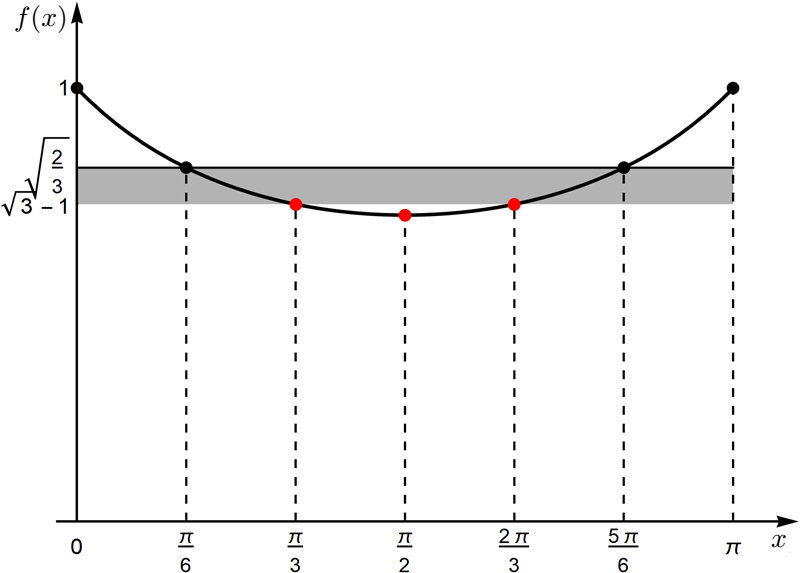}
\caption{The graph of the function $f(x)$. Red dots represent the discrete values $f(x)$ takes in the case of $N=6$ planar symmetric POVMs. Grey area with the upper boundary included and the lower boundary excluded represents the interval for $\eta$ where we have a $6$-cycle realized, cf.~Eq.~\eqref{ncyclebounds}.}
\label{graphf}
\end{figure}
\end{proof}
\end{proposition}
\subsection{Marginal surgery on arbitrary M-tuples of POVMs: N-Specker scenarios on a qubit}\label{subsec3_2}
In this section we will extend the marginal surgery described for pairs of POVMs to an arbitrary $M$-element subset of the set of $N$ planar symmetric POVMs, $s\equiv\{E_1,E_2,\ldots,E_N\}$ (cf.~Definition \ref{plansymdef}). Let us select $r\equiv\{E_{k_1},E_{k_2},\ldots,E_{k_{M}}\}$, where we assign the labels $\{k_i\}_{i=1}^M$ such that $k_1=1<k_2<k_3<\cdots<k_{M}$ and denote by $E^{r}$ the joint measurement obtained by marginalization of $G$:
\begin{equation}
E^{r}(x_{k_1},x_{k_2}\ldots,x_{k_{M}})=\sum_{\vec{x}}^{x_{k_1},x_{k_2}\ldots,x_{k_{M}}}G(\vec{x}),
\end{equation}
where $(x_{k_1},x_{k_2}\ldots,x_{k_{M}})\in\{\pm1\}^{M}$ are held fixed and the summation is carried out over the remaining entries of $\vec{x}\in\{\pm1\}^N$.

Due to Remark \ref{remcyc}, the outcomes to which non-zero effects of $E^{r}$ are assigned are of the form 
\begin{align}
(x_{k_1},x_{k_2}\ldots x_{k_{M}})=&(\underbrace{+1,\ldots,+1}_{p\text{ `$+1$'s }},\underbrace{-1,\ldots,-1}_{q\text{ `$-1$'s }})\text{ or }\nonumber\\
&(\underbrace{-1,\ldots,-1}_{p\text{ `$-1$'s }},\underbrace{+1,\ldots,+1}_{q\text{ `$+1$'s }}),\nonumber\\
\text{ where } &p,q=\overline{0,M}, p+q=M,
\end{align}
and are obtained by marginalization over all strings $\vec{x}\in\{\pm1\}^N$ such that 
\begin{align}
&x_1=x_2=\ldots=x_{k_p}=+1\nonumber\\
&\textrm{ and }x_{k_{p+1}}=x_{k_{p+1}+1}=\ldots=x_{N}=-1,\textrm{ or}\nonumber\\
&x_1=x_2=\ldots=x_{k_p}=-1\nonumber\\
&\textrm{ and } x_{k_{p+1}}=x_{k_{p+1}+1}=\ldots=x_{N}=+1.
\end{align}

Now the summation becomes 
\begin{align}
&E^{r}(\underbrace{+1,\ldots,+1}_{p\text{ `$+1$'s }},\underbrace{-1,\ldots,-1}_{q\text{ `$-1$'s }})=\sum_{\vec{x}\in\{\pm1\}^N}^{\substack{x_1,\ldots,x_{k_p}=+1\\x_{k_{p+1}},\ldots,x_{N}=-1}}G(\vec{x})\nonumber\\
&=
\sum_{\vec{x}\in\{\pm1\}^N}^{\substack{x_1,\ldots,x_{k_p}=+1\\x_{k_{p+1}},\ldots,x_{N}=-1}}\frac{1}{2N}I+\frac{\gamma}{2N}\sum_{\vec{x}\in\{\pm1\}^N}^{\substack{x_1,\ldots,x_{k_p}=+1\\x_{k_{p+1}},\ldots,x_{N}=-1}}\vec{e}(\vec{x})\cdot\vec{\sigma}\nonumber\\
&=\frac{k_{p+1}-k_p}{2N}I+\frac{\gamma}{2N}\sum_{\vec{x}\in\{\pm1\}^N}^{\substack{x_1,\ldots,x_{k_p}=+1\\x_{k_{p+1}},\ldots,x_{N}=-1}}\vec{e}(\vec{x})\cdot\sigma,\label{rels1}\\\textrm{ and }\nonumber\\
&E^{r}(\underbrace{-1,\ldots,-1}_{p\text{ `$-1$'s }},\underbrace{+1,\ldots,+1}_{q\text{ `$+1$'s }})=\sum_{\vec{x}\in\{\pm1\}^N}^{\substack{x_1,\ldots,x_{k_p}=-1\\x_{k_{p+1}},\ldots,x_{N}=+1}}G(\vec{x})\nonumber\\
&=
\sum_{\vec{x}\in\{\pm1\}^N}^{\substack{x_1,\ldots,x_{k_p}=-1\\x_{k_{p+1}},\ldots,x_{N}=+1}}\frac{1}{2N}I+\frac{\gamma}{2N}\sum_{\vec{x}\in\{\pm1\}^N}^{\substack{x_1,\ldots,x_{k_p}=-1\\x_{k_{p+1}},\ldots,x_{N}=+1}}\vec{e}(\vec{x})\cdot\vec{\sigma}\nonumber\\
&=\frac{k_{p+1}-k_p}{2N}I-\frac{\gamma}{2N}\sum_{\vec{x}\in\{\pm1\}^N}^{\substack{x_1,\ldots,x_{k_p}=-1\\x_{k_{p+1}},\ldots,x_{N}=+1}}\vec{e}(\vec{x})\cdot\vec{\sigma},
\label{rels2}
\end{align}
where we used the fact that there are $k_{p+1}-k_p$ non-zero effects assigned to outcomes that start with $+1$ in the first $k_p$ slots and end with $-1$ starting from the $k_{p+1}$-th slot (and the same is true when the first $k_p$ slots are $-1$ and from $k_{p+1}$ they are all $+1$), and also that $\vec{e}(-\vec{x})=-\vec{e}(\vec{x})$. There are exactly $M$ non-zero effects of $E^{r}$ assigned to outcomes that start with at least one $+1$, i.e., outcomes labelled by $M$-element strings in
 $$\{(+1,-1,\ldots,-1),(+1,+1,-1\ldots,-1),\ldots,(+1,\ldots,+1)\},$$ and another $M$ of those assigned to outcomes that start with at least one $-1$,  i.e., outcomes labelled by $M$-element strings in $$\{(-1,+1,\ldots,+1),(-1,-1,+1\ldots,+1),\ldots,(-1,\ldots,-1)\},$$ 
thus $2M$ non-zero effects in all. Equations \eqref{rels1} and \eqref{rels2} show that we can focus on specifying those effects that start with at least one $+1$ and then those that start with at least one $-1$ can be found just by inverting the sign of the corresponding geometric parts (i.e., the $\vec{\sigma}$-dependent parts of the POVM elements). In order to obtain the effects of $E^{r}$, we need to evaluate the sum $$\sum_{\vec{x}\in\{\pm1\}^N}^{\substack{x_1,\ldots,x_{k_p}=+1\\x_{k_{p+1}},\ldots,x_{N}=-1}}\vec{e}(\vec{x})$$
and we will do so by converting it into the sum over some index $i$ as we did in the previous section for the case of pairs of POVMs (see Appendix \ref{AppA}, Lemma \ref{id3} for how the limits on $i$ in the summation are obtained):
\begin{equation}
\resizebox{1\hsize}{!}{$\sum\limits_{\vec{x}\in\{\pm1\}^N}^{\substack{x_1,\ldots,x_{k_p}=+1\\x_{k_{p+1}},\ldots,x_{N}=-1}}\vec{e}(\vec{x})=\begin{cases}
\sum\limits_{i=k_p-\frac{N-1}{2}}^{k_{p+1}-\frac{N+1}{2}}\left(\cos\frac{(i-1)\pi}{N},\sin\frac{(i-1)\pi}{N},0\right)\\
\hspace{3.5cm}\text{for odd }N\\
\sum\limits_{i=k_p-\frac{N}{2}}^{k_{p+1}-1-\frac{N}{2}}\left(\cos\frac{(2i-1)\pi}{2N},\sin\frac{(2i-1)\pi}{2N},0\right)\\
\hspace{3.5cm}\text{for even }N.
\end{cases}$}
\end{equation}
where we have assumed that $p<M$. Using results from Appendix \ref{AppA}, Lemma \ref{id2}, we get in both cases
\begin{align}
&\sum_{\vec{x}\in\{\pm1\}^N}^{\substack{x_1,\ldots,x_{k_p}=+1\\x_{k_{p+1}},\ldots,x_{N}=-1}}\vec{e}(\vec{x})=\frac{\sin\left(\frac{k_{p+1}-k_p}{2N}\pi\right)}{\sin\frac{\pi}{2N}}\vec{t}(k_p,k_{p+1}),\textrm{ where}\nonumber\\
&\vec{t}(k_p,k_{p+1})\nonumber\\
&=\left(\sin\left(\frac{k_{p+1}+k_p-2}{2N}\pi\right),-\cos\left(\frac{k_{p+1}+k_p-2}{2N}\pi\right),0\right).
\end{align}

The case $p=M$ corresponds to the outcome $(x_1,x_{k_2},\ldots,x_{k_M})=(\pm1,\pm1,\ldots,\pm1)$ (all signs are the same), so that we have:
\begin{align}
&E^{(s')}(\pm1,\ldots,\pm1)=\sum_{\vec{x}\in\{\pm1\}^N}^{x_1,\ldots,x_{k_M}=\pm1}G(\vec{x})\nonumber\\
&=\sum_{\vec{x}\in\{\pm1\}^N}^{x_1,\ldots,x_{k_M}=\pm1}\frac{1}{2N}I\pm\frac{\gamma}{2N}\sum_{\vec{x}\in\{\pm1\}^N}^{x_1,\ldots,x_{k_M}=\pm1}\vec{e}(\vec{x})\cdot\vec{\sigma}\nonumber\\
&=\frac{N-k_M+1}{2N}I\pm\frac{\gamma}{2N}\sum_{\vec{x}\in\{\pm1\}^N}^{x_1,\ldots,x_{k_M}=\pm1}\vec{e}(\vec{x})\cdot\vec{\sigma}.
\end{align}
Note that there are $N-k_M+1$ terms in the summation above for $x_1=x_2=\ldots=x_{k_M}=+1$, given by
\begin{align}
&(x_{k_M+1},\ldots,x_N)\nonumber\\
&\in\{(-1,\ldots,-1),(+1,-1,\ldots,-1),(+1,+1,-1,\ldots,-1),\nonumber\\
&\ldots,(+1,\ldots,+1)\}
\end{align}
and similarly for $x_1=x_2=\ldots=x_{k_M}=-1$,
\begin{align}
&(x_{k_M+1},\ldots,x_N)\nonumber\\
&\in\{(+1,\ldots,+1),(-1,+1,\ldots,+1),(-1,-1,+1,\ldots,-1),\nonumber\\
&\ldots,(-1,\ldots,-1)\}.
\end{align} 

We can directly use Eq.~\eqref{jzid2} to get
\begin{align}
&\sum_{\vec{x}\in\{\pm1\}}^{x_1,\ldots,x_{k_M}=\pm1}\vec{e}(\vec{x})=\frac{\sin\left(\frac{N-k_M+1}{2N}\pi\right)}{\sin\frac{\pi}{2N}}\vec{s}(k_M),\nonumber\\
&\textrm{where }\vec{s}(k_M)=\left(\cos\frac{(k_M-1)\pi}{2N},\sin\frac{(k_M-1)\pi}{2N},0\right).
\end{align}
Now we have a complete specification of the effects of the joint POVM $E^{r}$:
\begin{subequations}
\begin{align}
&E^{r}(\underbrace{\pm1,\ldots,\pm1}_{p (<M) \text{ `$\pm1$'s }},\underbrace{\mp1,\ldots,\mp1}_{q\text{ `$\mp1$'s }})\nonumber\\
&=\frac{k_{p+1}-k_p}{2N}I\pm\frac{\gamma\sin\frac{(k_{p+1}-k_p)\pi}{2N}}{2N\sin\frac{\pi}{2N}}\vec{\sigma}\cdot\vec{t}(k_p,k_{p+1}),\textrm{ where}\nonumber\\
&\vec{t}(k_p,k_{p+1})\nonumber\\
&=\left(\sin\left(\frac{k_{p+1}+k_p-2}{2N}\pi\right),-\cos\left(\frac{k_{p+1}+k_p-2}{2N}\pi\right),0\right),\\&\textrm{ and}\nonumber\\
&E^{r}(\pm1,\ldots\ldots\ldots\ldots\ldots\pm1)\nonumber\\
&=\frac{N-k_M+1}{2N}I\pm\frac{\gamma\sin\frac{(N-k_M+1)\pi}{2N}}{2N\sin\frac{\pi}{2N}}\vec{\sigma}.\vec{s}(k_M),\nonumber\\
&\textrm{where }\vec{s}(k_M)=\left(\cos\frac{(k_M-1)\pi}{2N},\sin\frac{(k_M-1)\pi}{2N},0\right).
\label{margarbM}
\end{align}
\end{subequations}
Note that $\vec{t}(k_p,k_{p+1})$ is the unit vector perpendicular to the direction that bisects the angle between $\vec{n}_{k_p}$ and $\vec{n}_{k_{p+1}}$ and oriented such that it makes an acute angle with $\vec{n}_{k_p}$, while $\vec{s}(k_M)$ is a unit vector that bisects the angle between $\vec{n}_{k_M}$ and $\vec{n}_1$.

The POVM $E^{r}$ is valid on $r$ for $\eta\in(0,1/(N\sin(\pi/2N))]$ since it is a marginal of the joint POVM $G$. Now let us apply marginal surgery on $E^{(s')}$ to obtain a joint POVM for $s'$ that is valid for a broader range of the purity parameter $\eta$. Motivated by the form of $E^{r}$ we consider the following ansatz for the joint POVM of $r$:
\begin{align}
&G^{r}(\underbrace{\pm1,\ldots,\pm1}_{p (<M)\text{ `$\pm1$'s }},\underbrace{\mp1,\ldots,\mp1}_{q\text{ `$\mp1$'s }})\nonumber\\
&=\alpha(k_p,k_{p+1})I\pm \mu\frac{\sin\frac{(k_{p+1}-k_p)\pi}{2N}}{\sin\frac{\pi}{2N}}\vec{t}(k_p,k_{p+1})\cdot\vec{\sigma},\nonumber\\
&G^{r}(\pm1,\ldots\ldots\ldots\ldots\ldots\pm1)\nonumber\\
&=\beta(k_M)I\pm\mu\frac{\sin\frac{(N-k_M+1)\pi}{2N}}{\sin\frac{\pi}{2N}}\vec{s}(k_M)\cdot\vec{\sigma}
\label{optarbM}
\end{align}
Requiring the correct marginalization for $E_{k_l}$ where $l=\overline{1,M}$ we get
\begin{align}
&E_{k_l}(+1)=\sum_{\vec{x}\in\{\pm1\}^M}^{x_l=+1}G^{r}(\vec{x})=G^{r}(+1,\ldots,+1)\nonumber\\
&+\sum_{l\leq p<M}G^{r}(\underbrace{+1,\ldots,+1}_{p\text{ `$+1$'s }},\underbrace{-1,\ldots,-1}_{q\text{ `$-1$'s }})\nonumber\\
&+\sum_{1\leq p<l}G^{r}(\underbrace{-1,\ldots,-1}_{p\text{ `$-1$'s }},\underbrace{+1,\ldots,+1}_{q\text{ `$+1$'s }}),
\end{align}
which together with Eq.~\eqref{optarbM} 
yields 
\begin{equation}
\sum_{p=1}^{M-1}\alpha(k_p,k_{p+1})+\beta(k_M)=\frac{1}{2},\quad \frac{\mu}{\sin\frac{\pi}{2N}}=\frac{1}{2}\eta.
\label{constrmarg}
\end{equation}
Requiring the positivity of each effect of $G^{r}$, according to Eq.~\eqref{optarbM} we have the following constraints:
\begin{align}
&\mu\frac{\sin\frac{(k_{p+1}-k_p)\pi}{2N}}{\sin\frac{\pi}{2N}}\leq\alpha(k_p,k_{p+1}),\nonumber\\
&\mu\frac{\sin\frac{(N-k_M+1)\pi}{2N}}{\sin\frac{\pi}{2N}}\leq\beta(k_M).
\label{constrpos}
\end{align}
Eliminating everything except $\eta$ in Eqs.~\eqref{constrmarg} and \eqref{constrpos} we obtain an upper bound on $\eta$ for which we can find suitable $\alpha(k_p,k_{p+1})$ and $\beta(k_M)$ such that $G^{r}$ is a valid joint measurement on $r$ for values of $\eta$ constrained by
\begin{equation}
\eta\leq\frac{1}{\sum_{p=1}^{M-1}\sin\frac{(k_{p+1}-k_p)\pi}{2N}+\cos\frac{(k_M-1)\pi}{2N}}.
\label{boundarb}
\end{equation}
Geometrically, $(k_{p+1}-k_p)\pi/N$ and $(k_M-1)\pi/N$ are respectively the angles between $\vec{n}_{k_{p+1}}$ and $\vec{n}_{k_p}$, and $\vec{n}_{k_M}$ and $\vec{n}_1$, and we see that the upper bound on $\eta$ is an explicit function of half of those angles. In case of $M=2$, i.e., when we choose only two POVMs, Eq.~\eqref{boundarb} reduces to Eq.~\eqref{boundarb2} and in this case the condition of Eq.~\eqref{boundarb} is also necessary (cf.~Corollary \ref{Leta2}). It is readily seen from Corollary \ref{3unb} that this condition is necessary  for $M=3$ as well. For $M=N$, i.e., when we consider all $N$ POVMs, Eq.~\eqref{boundarb} reduces to $\eta\leq 1/(N\sin\pi/2N)$ which is also known to be necessary from Theorem \ref{gluslov}. So for $M\in\{2,3,N\}$ the given condition is both necessary and sufficient for joint measurability. However, for every other $M$, while it is certainly sufficient, we cannot say if it is necessary for joint measurability. We conjecture its necessity (see Conjecture \ref{conj1}).\footnote{Note that we don't need this conjecture to hold to establish the results presented in this paper.}

It remains to make a choice for $\alpha(k_p,k_{p+1})$ and $\beta(k_M)$ for each $\eta$. We do this in the following manner, subject to constraints from Eqs.~\eqref{constrmarg} and \eqref{constrpos}:
\begin{align}
\alpha(k_p,k_{p+1})&=\frac{1}{2}\eta\sin\frac{(k_{p+1}-k_p)\pi}{2N},\nonumber\\
\beta(k_M)&=\frac{1}{2}\left(1-\eta\sum_{p=1}^{M-1}\sin\frac{(k_{p+1}-k_p)\pi}{2N}\right)
\label{arbMparameters}
\end{align}
which gives for $G^{(s')}$:
\begin{align}
&G^{r}(\underbrace{\pm1,\ldots,\pm1}_{p\text{ `$\pm1$'s }},\underbrace{\mp1,\ldots,\mp1}_{q\text{ `$\mp1$'s }})\nonumber\\
&=\frac{1}{2}\eta\sin\frac{(k_{p+1}-k_p)\pi}{2N}\Big(I\pm\vec{t}(k_p,k_{p+1})\cdot\vec{\sigma}\Big),\nonumber\\&\textrm{and}\nonumber\\
&G^{r}(\pm1,\ldots\ldots\ldots\ldots\ldots\pm1)\nonumber\\
&=\frac{1}{2}\left(1-\eta\sum_{p=1}^{M-1}\sin\frac{(k_{p+1}-k_p)\pi}{2N}\right)I\nonumber\\
&\pm\frac{1}{2}\eta\cos\frac{(k_M-1)\pi}{2N}\vec{s}(k_M)\cdot\vec{\sigma}.
\label{arb1Mkonacno}
\end{align}
When $\eta$ takes the value of its upper bound, $G^{r}$ forms a rank-$1$ POVM.

Note that so far we have taken $E_{k_1}=E_1$ for calculational convenience. However, our procedure can be applied to any arbitrary choice of $M$ out of $N$ planar symmetric POVMs. Namely if we take $r'=\{E_{k_1},E_{k_2},\ldots,E_{k_M}\}$ such that $1<k_1<k_2<\dots<k_M$, then the set $r=\{E_1,E_{k_2-k_1+1},\ldots,E_{k_M-k_1+1}\}$ is geometrically equivalent to the set $r'$ because it is related to $r'$ by a counter-clockwise rotation of $\frac{(k_1-1)\pi}{N}$ in an equatorial plane ($XY$ plane) of the Bloch ball and thus it is jointly measurable if and only if $r'$ is jointly measurable (cf.~Corollary \ref{geqplansym}) and the joint POVM on $r'$ is given by 
\begin{equation}
G^{r'}=C_{2N}^{k_1-1}G^{r}.
\end{equation}
 This together with Eq.~\eqref{arb1Mkonacno} yields: 
\begin{align}
&G^{r'}(\underbrace{\pm1,\ldots,\pm1}_{p\text{ `$\pm1$'s }},\underbrace{\mp1,\ldots,\mp1}_{q\text{ `$\mp1$'s }})\nonumber\\
&=\frac{1}{2}\eta\sin\frac{(k_{p+1}-k_p)\pi}{2N}\Big(I\pm\vec{t}(k_p,k_{p+1})\cdot\vec{\sigma}\Big),\\\nonumber\\
&G^{r'}(\pm1,\ldots\ldots\ldots\ldots\ldots\pm1)\nonumber\\
&=\frac{1}{2}\left(1-\eta\sum_{p=1}^{M-1}\sin\frac{(k_{p+1}-k_p)\pi}{2N}\right)I\nonumber\\
&\pm\frac{1}{2}\eta\cos\frac{(k_M-k_1)\pi}{2N}\vec{s}(k_1,k_M)\cdot\vec{\sigma},\textrm{ where}\\
&\vec{s}(k_1,k_M)=\left(\cos\frac{(k_M+k_1-2)\pi}{2N},\sin\frac{(k_M+k_1-2)\pi}{2N},0\right).\nonumber
\label{arb1Mkonacno}
\end{align}
All of this serves as a proof for the following sufficient condition for joint measurability:
\begin{lemma}\label{MtuplePlanSymjmc}
Let $s=\{E_1,\ldots,E_N\}$ be the set of $N$ planar symmetric POVMs. A sufficient condition for its $M$-element subset $\{E_{k_1},E_{k_2},\ldots,E_{k_M}\}$ to be jointly measurable is given by
\begin{equation}
\eta\in\left(0,\frac{1}{\sum_{p=1}^{M-1}\sin\frac{(k_{p+1}-k_p)\pi}{2N}+\cos\frac{(k_M-k_1)\pi}{2N}}\right].
\label{l3}
\end{equation}
\end{lemma}
As we already noted, we conjecture this condition is necessary as well, i.e.,
\begin{conjecture}\label{conj1}
	The sufficient condition for the joint measurability of any $M$-element subset of the set of $N$ planar symmetric POVMs given by Eq.~\eqref{l3} is also necessary.
\end{conjecture}
\begin{corollary}
Let $s=\{E_1,\ldots,E_N\}$ be the set of $N$ planar symmetric POVMs. Any $(N-1)$-element subset of $s$ is jointly measurable if 
\begin{equation}
\eta\in\left(0,\frac{1}{(N-2)\sin\frac{\pi}{2N}+\sin\frac{\pi}{N}}\right].
\end{equation}
\begin{proof}
Let $r=\{E_{k_1},\ldots,E_{k_{N-1}}\}$ where $k_1<k_2<\dots<k_{N-1}$. We will split the proof into two cases:

\noindent {\em Case 1}: $k_1,\ldots,k_{N-1}$ are successive integers which geometrically corresponds to $\vec{n}_{k_1},\ldots,\vec{n}_{k_{N-1}}$ being as closely grouped as possible, i.e., the angular separation between the successive vectors is $\pi/N$ and the total angular span is $\frac{(N-1)\pi}{N}$. The two possibilities for $r$ are $\{E_1,E_2,\dots,E_{N-1}\}$ and $\{E_2,E_3,\dots,E_N\}$. In this case, $k_{p+1}-k_p=1$ and $k_{N-1}-k_1=N-2$ so we have by Eq.~\eqref{l3} that $r$ is compatible if
\begin{equation}
\eta\leq\frac{1}{\sum_{p=1}^{N-2}\sin\frac{\pi}{2N}+\cos\frac{(N-2)\pi}{2N}}=\frac{1}{(N-2)\sin\frac{\pi}{2N}+\sin\frac{\pi}{N}}.
\end{equation}

\noindent{\em Case 2}:$k_1,\ldots,k_{N-1}$ are not successive integers, i.e., there is exactly one ``missing integer" (say, $i$) between a pair of them. There are $N-2$ such possibilities and in each of them $k_1=1$ and $k_{N-1}=N$. Therefore, $k_{N-1}-k_1={N-1}$ and $k_{p+1}-k_p=1$ for $N-3$ successive pairs while $k_{p+1}-k_p=2$ for one pair, when we skip the ``missing integer" (i.e., $p=i-1$). Hence, from Eq.~\eqref{l3}, $r$ is jointly measurable if:
\begin{align}
\eta&\leq\frac{1}{\sum_{p=1}^{N-2}\sin\frac{(k_{p+1}-k_p)\pi}{2N}+\cos\frac{(N-1)\pi}{2N}}\nonumber\\
&=\frac{1}{(N-3)\sin\frac{\pi}{2N}+\sin\frac{\pi}{N}+\cos\left(\frac{\pi}{2}-\frac{\pi}{2N}\right)}\nonumber\\
&=\frac{1}{(N-2)\sin\frac{\pi}{2N}+\sin\frac{\pi}{N}},
\end{align} 
which completes the proof.
\end{proof}
\label{l4}
\end{corollary}

\begin{corollary}\label{N-Specker's}($N$-Specker scenario on a qubit with planar symmetric POVMs) Let $s=\{E_1,E_2,\ldots,E_N\}$ be the set of planar symmetric POVMs with purity $\eta$. For 
\begin{equation}
\eta\in\left(\frac{1}{N\sin\frac{\pi}{2N}},\frac{1}{(N-2)\sin\frac{\pi}{2N}+\sin\frac{\pi}{N}}\right],
\end{equation}
the joint measurability structure of $s$ forms an $N$-Specker scenario (cf.~Example~\ref{NSpeckdefn}).
\begin{proof}
We have that $s$ is compatible if and only if $$\eta\in\left(0,\frac{1}{N\sin\frac{\pi}{2N}}\right],$$
while any $(N-1)$-element subset of $s$ is compatible if
$$\eta\in\left(0,\frac{1}{(N-2)\sin\frac{\pi}{2N}+\sin\frac{\pi}{N}}\right].$$
Since
\begin{align*}
&(N-2)\sin\frac{\pi}{2N}+\sin\frac{\pi}{N}<N\sin\frac{\pi}{2N},
\end{align*}
we have an open gap
$$\eta\in\left(\frac{1}{N\sin\frac{\pi}{2N}},\frac{1}{(N-2)\sin\frac{\pi}{2N}+\sin\frac{\pi}{N}}\right],$$
for which $s$ is incompatible but every $(N-1)$-element subset of $s$ is compatible, thus realizing an $N$-Specker scenario.
\end{proof} 
\end{corollary}
Corollary \ref{N-Specker's} shows that we can realize any $N$-Specker scenario on a qubit with planar symmetric POVMs. Ref.~\cite{KHF14} provided a constructive proof that all conceivable joint measurability structures for any finite set of measurements are realizable with POVMs in quantum theory. Crucial to the construction in Ref.~\cite{KHF14} are the so-called {\em minimal incompatible sets} of POVMs (of which, $N$-Specker scenarios are examples)\footnote{A minimal incompatible set of vertices in a joint measurability structure is such that every proper subset of this set is compatible, i.e., shares a hyperedge. $N$-Specker ($N\geq 3$) scenarios constitute all examples of minimal incompatible sets, except one: a pair of incompatible measurements. Such a pair forms a minimal incompatible set since each POVM in the pair is trivially compatible with itself but the two POVMs are incompatible. One could even call it a ``$2$-Specker" scenario.} constructed for subhypergraphs in the hypergraph representing the joint measurability structure to be realized. The construction in Ref.~\cite{KHF14} proceeds by decomposing a joint measurability structure into minimal incompatible sets. Each minimal incompatible set is then realized on some Hilbert space $\mathcal{H}_i$ and the dimensionality of the overall Hilbert space, $\mathcal{H}=\oplus_i \mathcal{H}_i$, required to realize a joint measurability structure based on this recipe is $\dim \mathcal{H}=\sum_i\dim \mathcal{H}_i$. Our construction of all $N$-Specker scenarios on a qubit here can therefore be used to make the construction of all joint measurability structures in Ref.~\cite{KHF14} maximally efficient given the recipe adopted there, i.e., we have $\dim \mathcal{H}_i=2$ for all minimal incompatible sets indexed by $i$.
\section{Sufficient condition for joint measurability of binary qubit POVM\lowercase{s} from marginal surgery}\label{sec4}
In the previous sections, we started with the setting of planar symmetric POVMs (cf.~Definition \ref{plansymdef} and Figure \ref{plansymm}) and then applied marginal surgery to obtain joint POVMs for arbitrary $M$-tuples of them. It turns out that geometric insights from Eq.~\eqref{arb1Mkonacno}, which is the result of the aforementioned procedure, can be used to construct the joint POVM for an arbitrary set of $N$ coplanar unbiased binary qubit POVMs with the same purity $\eta$ bounded above by some $\eta_{\rm max}$ that depends on the particular setting.
\begin{figure}[htb!]
\centering
\includegraphics[scale=0.47]{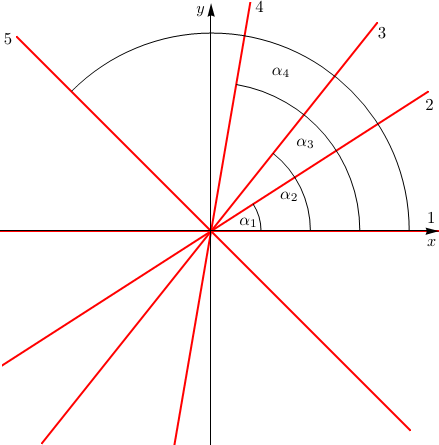}
\caption{Lines defined by $\vec{n}_k$}
\label{figarbplan}
\end{figure}
\begin{theorem}\label{counbiqu}
Let $s=\{E_1,E_2,\ldots,E_N\}$ be a set of coplanar unbiased binary qubit POVMs with the same purity defined by
\begin{equation}
E_k(x_k)\equiv \frac{1}{2}\left(1+x_k\eta\vec{n}_k\cdot\vec{\sigma}\right),\quad x_k\in\{\pm1\},
\end{equation}
and let $\alpha_k$, $k=\overline{1,N-1}$, be the angles defined by the rays in the upper half plane of the lines determined by $\vec{n}_k$ (where the $X$-axis is chosen to be along $\vec{n}_1$)(see Fig. ~\ref{figarbplan}). Then a sufficient  condition for the compatibility of $s$ is
\begin{equation}
\eta\leq\frac{1}{\sum_{k=1}^{N-1}\sin\frac{\alpha_{k}-\alpha_{k-1}}{2}+\cos\frac{\alpha_{N-1}}{2}},\textrm{ where }\alpha_0\equiv0,
\end{equation} 
and a joint POVM satisfying this constraint is
\begin{align}
&G^{s}(\underbrace{\pm1,\ldots,\pm1}_{p\text{ `$\pm1$'s }},\underbrace{\mp1,\ldots,\mp1}_{N-p\text{ `$\mp1$'s }})\nonumber\\
&=\frac{1}{2}\eta\sin\frac{\alpha_{p}-\alpha_{p-1}}{2}\Big(I\pm\vec{t}_p\cdot\vec{\sigma}\Big),\textrm{ where}\nonumber\\
&\vec{t_p}=\left(\sin\frac{\alpha_{p}+\alpha_{p-1}}{2},-\cos\frac{\alpha_{p}+\alpha_{p-1}}{2},0\right), 1\leq p<N,\nonumber\\&\textrm{and}\nonumber\\
&G^{s}(\pm1,\ldots\ldots\ldots\ldots\ldots\pm1)\nonumber\\
&=\frac{1}{2}\left(1-\eta\sum_{p=1}^{N-1}\sin\frac{\alpha_{p}-\alpha_{p-1}}{2}\right)I\pm\frac{1}{2}\eta\cos\frac{\alpha_{N-1}}{2}\vec{s}\cdot\vec{\sigma},\nonumber\\
&\textrm{where }\vec{s}=\left(\cos\frac{\alpha_{N-1}}{2},\sin\frac{\alpha_{N-1}}{2},0\right),\label{jpovmarbcop}
\end{align}
setting all other effects of $G^s$ to $0$.
\begin{proof}
We will prove this by construction of a joint POVM for $s$ using basic geometric insights gained in the previous subsection. Without loss of generality, we can choose that all of the Bloch vectors are in the upper half plane. With this convention,  $\alpha_{k-1}$ is the angle between $\vec{n}_k$ and the $x$-axis i.e.,
\begin{equation}
\vec{n}_k=(\cos\alpha_{k-1},\sin\alpha_{k-1},0).
\end{equation}
Let us return to Eq.~\eqref{optarbM}. We have previously mentioned (between Eqs.~\eqref{margarbM} and \eqref{optarbM}) that:
\begin{itemize}
\item the quantity $\displaystyle\frac{(k_{p+1}-k_p)\pi}{N}$ is the angle between two successive Bloch vectors, $\vec{n}_{k_p}$ and $\vec{n}_{k_{p+1}}$, from the $M$-element set of POVMs $\{E_{k_1},E_{k_2},\dots,E_{k_M}\}$.
\item $\vec{t}(k_p,k_{p+1})$ is the unit vector that is orthogonal to the line that bisects the angle between $\vec{n}_{k_p}$ and $\vec{n}_{k_{p+1}}$ and makes an acute angle with respect to $\vec{n}_{k_p}$;
\item $\vec{s}(k_M)$ is a unit vector that bisects the angle between $\vec{n}_1$ and $\vec{n}_{k_M}$ and is oriented towards the upper half plane.
\end{itemize}
This suggests searching for a joint POVM for $s$ of the following form given in Eq.~\eqref{jpovmarbcop} (cf.~Eq.~\eqref{arbMparameters}), setting all other effects of $G^s$ to $0$.
The positivity constraint on the effects implies that  each $G^s(\vec{x})$ is a valid effect for (cf.~Eq.~\eqref{boundarb})\\
\begin{align}
&1-\eta\sum_{p=1}^{N-1}\sin\frac{\alpha_{p}-\alpha_{p-1}}{2}\geq\eta\cos\frac{\alpha_{N-1}}{2}\nonumber\\
&\Longrightarrow \eta\leq\frac{1}{\sum_{k=1}^{N-1}\sin\frac{\alpha_{p}-\alpha_{p-1}}{2}+\cos\frac{\alpha_{N-1}}{2}}.
\end{align}
Now we have to show that the marginalization is correct:
\begin{equation}
\sum_{\vec{y}\in\{\pm\}^N}^{y_k=\pm1}G^s(\vec{y})=\frac{1}{2}(I\pm\eta\vec{n}_k\cdot\vec{\sigma})=E_k(\pm1).\label{margidth1}
\end{equation}
We show in detail that this is indeed the case in Appendix~\ref{AppA}, Lemma~\ref{appAmarg} which completes this proof.
\end{proof}
\begin{figure}[htb!]
\centering
\includegraphics[scale=0.57]{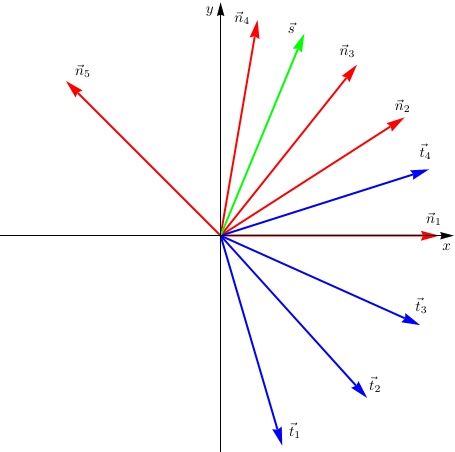}
\caption{An example with $5$ POVMs}
\end{figure}
\label{tplan}
\end{theorem}
\noindent The sufficient condition given by this theorem is known to be necessary as well for $N=2$ and $N=3$ (cf. Corollary \ref{Leta2} and Corollary \ref{3unbcopjmc}), as well as for arbitrary $N$ planar symmetric POVMs (cf. Theorem \ref{plansymjmc}). Therefore, we make the following conjecture:
\begin{conjecture}\label{conj2}
The sufficient condition for joint measurability of coplanar unbiased binary qubit POVMs with same purity given by Theorem \ref{counbiqu} is also necessary.
\end{conjecture}
The truth of this conjecture would imply the truth of Conjecture \ref{conj1}.
 
 Notice that we can rewrite:
\begin{align}
\sin\frac{\alpha_p-\alpha_{p-1}}{2}\vec{t}_p&=\frac{1}{2}(\vec{n}_p-\vec{n}_{p+1})\\
\cos\frac{\alpha_{N-1}}{2}\vec{s}&=\frac{1}{2}(\vec{n}_1+\vec{n}_N),
\end{align}
which suggest further generalizations of the result given by Theorem \ref{tplan}. We pursue this suggestion below for the case of $N$, possibly biased, binary qubit POVMs.
\begin{theorem}\label{suffNnsp}
Let $s=\{E_1,\ldots,E_N\}$ be the set of some binary qubit POVMs (possibly biased) specified with 
\begin{equation}
E_k(\pm1)=\frac{1}{2}\bigg((1\pm b_k)I\pm\vec{a}_k\cdot\vec{\sigma}\bigg).
\end{equation}
The set $s$ is compatible if 
\begin{equation}
||\vec{a_1}+\vec{a}_N||+\sum_{p=1}^{N-1}||\vec{a}_p-\vec{a}_{p+1}||\leq 2(1-\max\{|b_k|\}).
\end{equation}
Adopting the convention that the outcomes of $E_k$ are labelled such that positive bias is associated with the outcome `$+1$', i.e. all $b_k\geq 0$, and that the POVMs themselves are labelled by non-decreasing bias ($b_{k}\geq b_{p}$ for $k>p$), a joint POVM that satisfies the given joint measurability condition on $s$ is given by  
\begin{subequations}
\label{arbcoppovm}
\begin{align}
&G^{s}(\underbrace{+1,\ldots,+1}_{p\text{ `$+1$'s }},\underbrace{-1,\ldots,-1}_{N-p\text{ `$-1$'s }})=\alpha_p+ \vec{T}_p\cdot\vec{\sigma},\\
&G^{s}(\underbrace{-1,\ldots,-1}_{p\text{ `$-1$'s }},\underbrace{+1,\ldots,+1}_{N-p\text{ `$+1$'s }})=\gamma_p - \vec{T}_p\cdot\vec{\sigma},\\
&G^{s}(+1,\ldots\ldots\ldots\ldots\ldots+1)=\beta+\vec{S}\cdot\vec{\sigma},\\
&G^{s}(-1,\ldots\ldots\ldots\ldots\ldots-1)=\delta-\vec{S}\cdot\vec{\sigma},
\end{align}
\end{subequations}
where 
\begin{subequations}
\begin{align}
\vec{T}_p&=\frac{1}{4}\left(\vec{a}_p-\vec{a}_{p+1}\right),\text{ }1\leq p<N\\ 
\vec{S}&=\frac{1}{4}\left(\vec{a}_1+\vec{a}_N\right),\\
\alpha_p&=\left|\left|\vec{T}_p\right|\right|,\quad \gamma_p=\left|\left|\vec{T}_p\right|\right|+\frac{b_{p+1}-b_p}{2},\label{alphap}\\
\beta&=\frac{1}{2}(1+b_1)-\sum_{p=1}^{N-1}\left|\left|\vec{T}_p\right|\right|,\\
\delta&=\frac{1}{2}(1-b_N)-\sum_{p=1}^{N-1}\left|\left|\vec{T}_p\right|\right|,\label{deltap}
\end{align}
with all other effects of $G^s$ set to $0$.
\end{subequations}
\begin{proof}
We will prove this by construction. Consider the candidate for a joint POVM of $s$ given in the form as in Eq.~\eqref{arbcoppovm} with other effects of $G^s$ set to $0$. Positivity of $G^s$ requires
\begin{equation}
\left|\left|\vec{T}_p\right|\right|\leq\alpha_p,\quad \left|\left|\vec{T}_p\right|\right|\leq\gamma_p,\quad\left|\left|\vec{S}\right|\right|\leq\beta,\quad\left|\left|\vec{S}\right|\right|\leq\delta.
\label{pos3D}
\end{equation}
As for the marginalization, it is easy to see that the geometric part of the joint POVM marginalizes correctly: 
\begin{align}
&\pm\vec{S}\pm\sum_{p=k}^{N-1}\vec{T}_k\mp\sum_{p=1}^{k-1}\vec{T}_k\nonumber\\
&=\pm\frac{1}{4}\left(\vec{a}_1+\vec{a}_N+\sum_{p=k}^{N-1}(\vec{a}_p-\vec{a}_{p+1})+\sum_{p=1}^{k-1}(\vec{a}_{p+1}-\vec{a}_{p})\right)\nonumber\\
&=\pm\frac{1}{2}\vec{a}_k.
\end{align}
Requiring the correct marginalization of the remaining part of $G^s$, we get
\begin{align}
&\sum_{p=1}^{k-1}\gamma_p+\sum_{p=k}^{N-1}\alpha_p+\beta=\frac{1}{2}(1+b_k),\nonumber\\
&\sum_{p=1}^{k-1}\alpha_p+\sum_{p=k}^{N-1}\gamma_p+\delta=\frac{1}{2}(1-b_k).
\label{marg3D}
\end{align}
From Eqs.~\eqref{pos3D} and \eqref{marg3D} we get 
\begin{equation}
\left|\left|\vec{S}\right|\right|+\sum_{p=1}^{N-1}\left|\left|\vec{T}_p\right|\right|\leq\frac{1}{2}(1-\max\{|b_k|\}).
\end{equation}
Substituting the expressions for $\vec{S}$ and $\vec{T}_p$ we get
\begin{equation}
\left|\left|\vec{a}_1+\vec{a}_N\right|\right|+\sum_{p=1}^{N-1}\left|\left|\vec{a}_p-\vec{a}_{p+1}\right|\right|\leq 2(1-\max\{|b_k|\}).
\label{usl3D}
\end{equation}
If we label $\{E_p\}_{p=1}^N$ such that $b_p$ increases with $p$ and further label outcomes such that $b_p$ is positive, we see that the choice for $\alpha_p$, $\gamma_p$, $\beta$, and $\delta$ given in Eqs.~\eqref{alphap}-\eqref{deltap} is valid given positivity and marginalization constraints. Thus, we have constructed a joint POVM for $s$.
\end{proof}
\end{theorem}
The sufficient condition derived here is known to be necessary in the case of $N=2$ unbiased POVMs (cf. Theorem~\ref{2unbjmc}) and in the case of $N=3$ coplanar unbiased POVMs if $\vec{a}_2$ is not a convex combination of $\vec{a}_1$ and $\vec{a}_3$ (cf. Corollary~ \ref{3copljmc}). Motivated by this, we make the following conjecture:
\begin{conjecture}\label{conj3}
		The sufficient condition given in Theorem~\ref{suffNnsp} is also necessary in the case of unbiased coplanar POVMs where 1) the Bloch vectors $\vec{a}_k$, in order from $k=1$ to $k=N$, lie on the lines that make the total angular span less than $\pi$ (as in Fig. \ref{figarbplan}), 2) the  outcomes are labelled such that they all point towards the upper half of the plane, and 3) none of the Bloch vectors is a convex combination of any other pair of Bloch vectors.
	\end{conjecture}
The truth of Conjecture \ref{conj3} would imply the truth of Conjecture \ref{conj2} and, hence, the truth of Conjecture \ref{conj1}.

\section{Miscellaneous joint measurability structures on a qubit}
\label{sec5}
In this section we give explicit construction of the joint measurability structures realizable on a qubit, many of which are not realizable with PVMs on any quantum system, using previous results.
\subsection{Progression from N-Cycle to N-Specker scenarios}\label{subsec5_1}
We start with planar symmetric POVMs $E_k(x_k)=\frac{1}{2}\big(I+\eta x_k \vec{n}_k\cdot\vec{\sigma} \big)$ (cf.~Definition \ref{plansymdef}):
We will examine how the joint measurability structure changes when we change $\eta$ for sets of planar symmetric POVMs, up to $N=6$. All sets geometrically equivalent (cf.~Definition \ref{classdef}) to $\{E_{k_1},E_{k_2},\ldots,E_{k_M}\}$ belong to a class denoted by $[E_{k_1}, E_{k_2},\ldots,E_{k_M}]$. 
\subsubsection{$N=4$ planar symmetric POVMs}
\begin{example}\upshape
According to Proposition~\ref{ncycletheorem}, we have a $4$-Cycle scenario (Fig. \ref{fig4ncyc}) if and only if $\eta\in\left(\frac{\sqrt{2}}{2},\sqrt{2-\sqrt{2}}\right].$
\begin{figure}[H]
\centering
\includegraphics[scale=0.18]{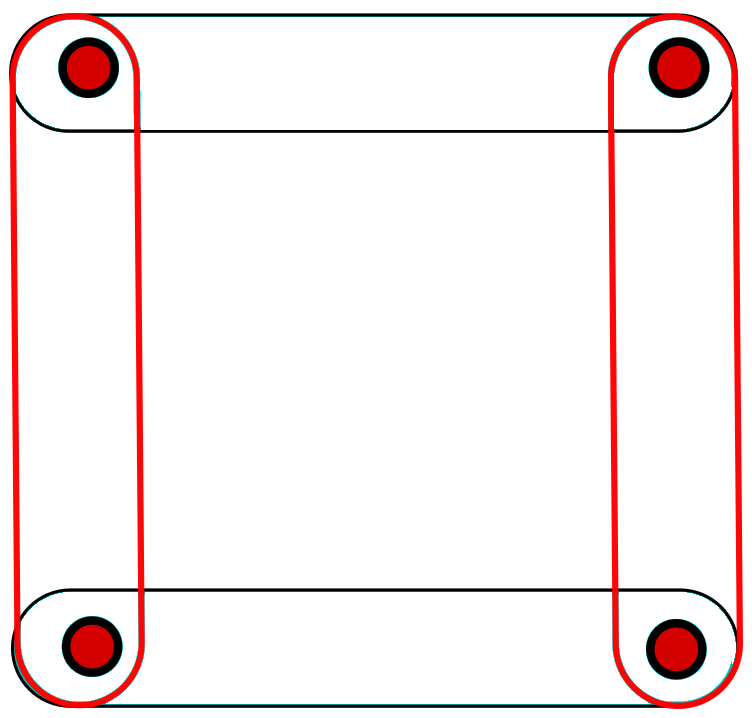}
\label{fig4ncyc}
\caption{4-cycle scenario}
\end{figure}
\end{example}
\begin{example}\upshape
Using Corollary \ref{Leta2}, $[E_1,E_3]$ is compatible when and only when
$$\eta\leq\frac{1}{\sin\frac{\pi}{4}+\cos\frac{\pi}{4}}=\frac{\sqrt{2}}{2}\approx0.70711.$$
On the other hand $[E_1,E_2,E_3]$ is incompatible, by Corollary \ref{3unb}, if and only if
$$\eta>\frac{1}{\cos\frac{\pi}{4}+\sin\frac{\pi}{8}+\sin\frac{\pi}{8}}=\frac{\sqrt{2}}{1+\sqrt{2(2-\sqrt{2})}}\approx0.67913.$$
So, if and only if $\eta\in\left(\frac{\sqrt{2}}{1+\sqrt{2(2-\sqrt{2})}},\frac{\sqrt{2}}{2}\right]$, we have the joint measurability structure $$\{\{E_1,E_2\},\{E_1,E_3\},\{E_2,E_3\}, \{E_2,E_4\},\{E_1,E_4\},\{E_2,E_4\}\},$$ where every pair of POVMs is jointly measurable (but no larger sets are), illustrated in Figure \ref{4all}. The joint measurability structure is a complete graph with four vertices.
\begin{figure}[H]
\centering
\includegraphics[scale=0.18]{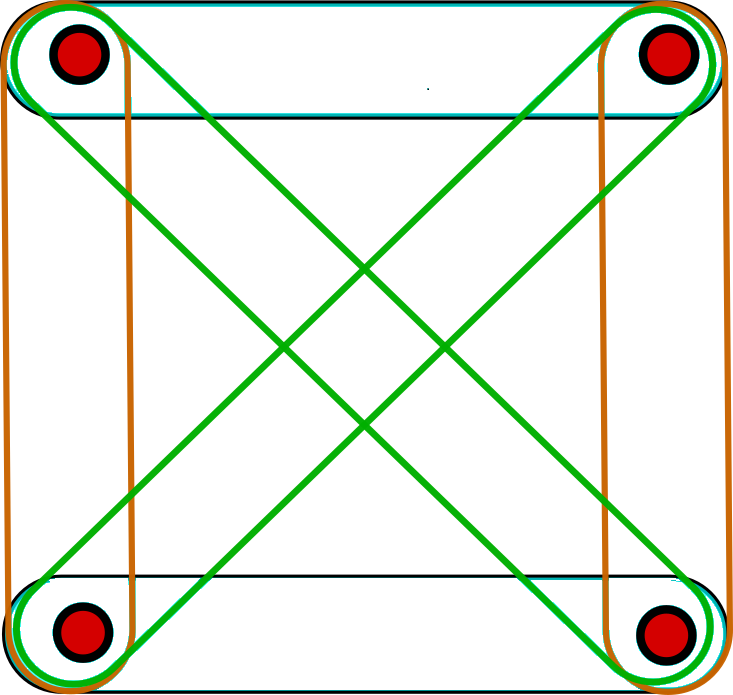}
\caption{The complete graph joint measurability structure of $N=4$ planar symmetric POVMs such that every pair (and no other subset) is jointly measurable.}
\label{4all}
\end{figure}
\end{example}
\begin{example}\upshape
Using Theorem \ref{gluslov}, we know that the set of four POVMs is incompatible if and only if
\begin{equation}
\eta>\frac{1}{4\sin\frac{\pi}{8}}=\frac{1}{2}\sqrt{1+\frac{\sqrt{2}}{2}}\approx0.65328,
\end{equation}
while $[E_1,E_2,E_3]$ is compatible if and only if (cf.~Corollary \ref{3unb})
$$\eta\leq\frac{1}{2\sin\frac{\pi}{8}+\cos\frac{\pi}{4}}=\frac{\sqrt{2}}{1+\sqrt{2(2-\sqrt{2})}}\approx0.67913,$$
meaning that we have a 4-Specker scenario (Figure \ref{fig4specker}) if and only if $\eta\in\left(\frac{1}{2}\sqrt{1+\frac{\sqrt{2}}{2}},\frac{\sqrt{2}}{1+\sqrt{2(2-\sqrt{2})}}\right]$.
\begin{figure}[H]
\centering
\includegraphics[scale=0.18]{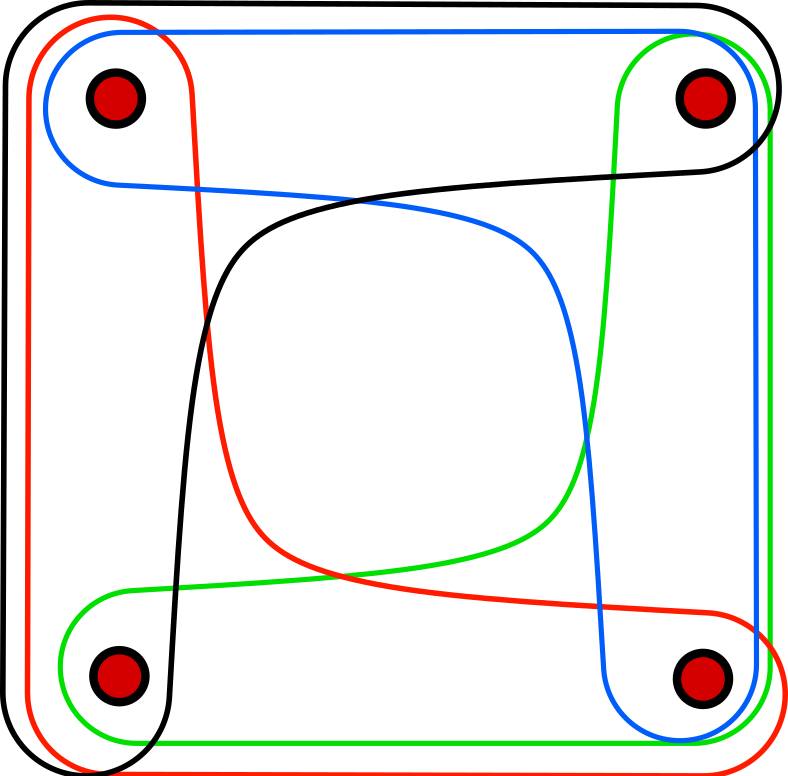}
\caption{4-Specker scenario.}
\label{fig4specker}
\end{figure}
\end{example}
\subsubsection{$N=5$ planar symmetric POVMs}
\begin{example}\upshape
According to Proposition \ref{ncycletheorem}, we have a 5-cycle scenario (Figure \ref{5cycle}) if and only if $\eta\in\left(\frac{1}{2}\left(3+\sqrt{5}-\sqrt{2(5+\sqrt{5})}\right),\frac{4}{\sqrt{5}-1+\sqrt{2(5+\sqrt{5})}}\right]$.
\begin{figure}[H]
\centering
\includegraphics[scale=0.162]{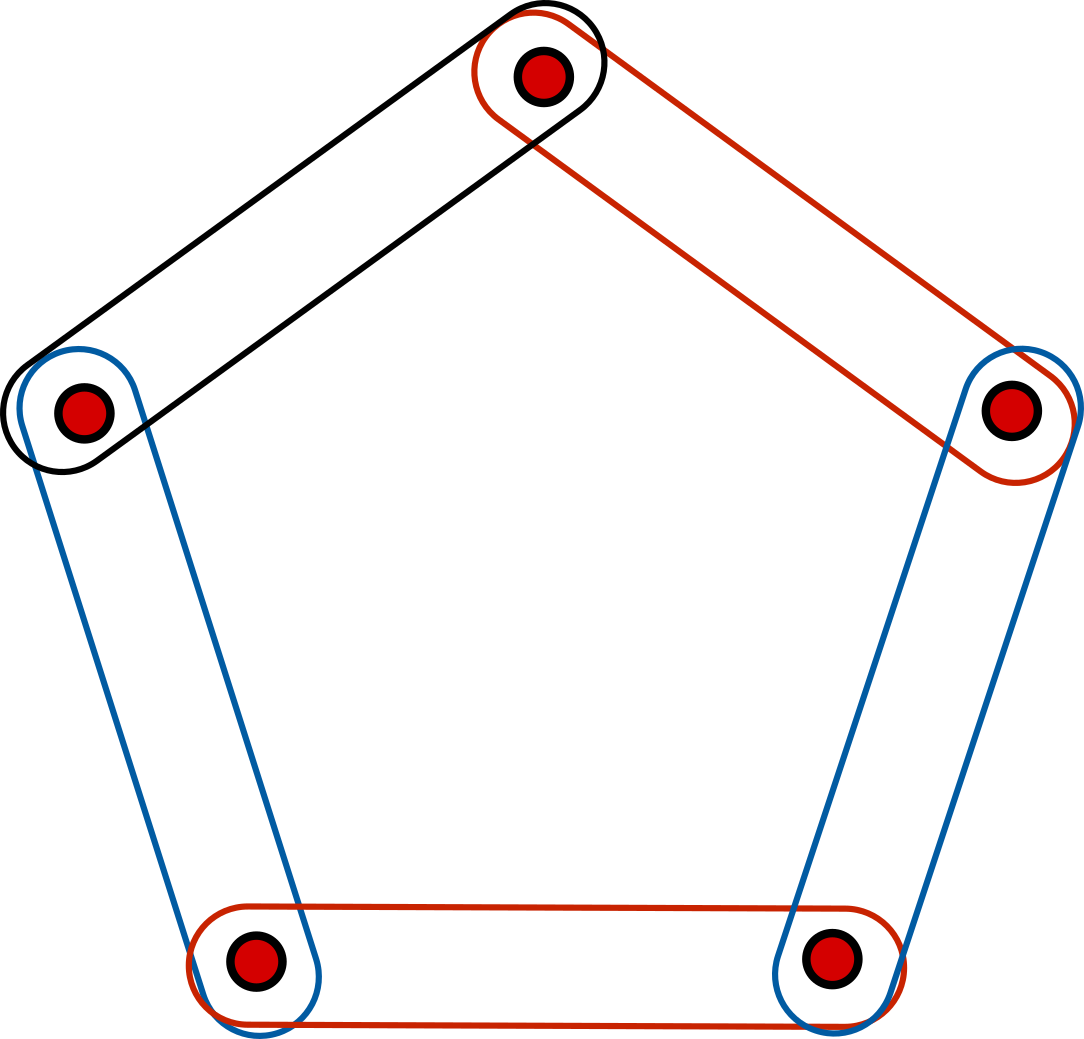}
\caption{5-cycle scenario.}
\label{5cycle}
\end{figure}
\end{example}
\begin{example}\upshape
$[E_1,E_3]$ is compatible if and only if
$$\eta\leq \frac{1}{\sin\frac{\pi}{5}+\cos\frac{\pi}{5}}=\frac{1}{2}\left(3+\sqrt{5}-\sqrt{2(5+\sqrt{5})}\right)\approx0.71592,$$
while $[E_1,E_2,E_3]$ is incompatible if and only if
$$\eta>\frac{1}{2\sin\frac{\pi}{10}+\cos\frac{\pi}{5}}=\frac{1}{11}(1+3\sqrt{5})\approx 0.70075.$$
This means that if and only if $\eta\in\left(\frac{1}{11}(1+3\sqrt{5}),\frac{1}{2}\left(3+\sqrt{5}-\sqrt{2(5+\sqrt{5})}\right)\right]$, we have the following joint measurability structure where every pair in the set of five POVMs is compatible but every other subset is incompatible, i.e.,
\begin{align}
\Big\{&\{E_1,E_2\},\{E_1,E_3\},\{E_1,E_4\},\{E_1,E_5\},\{E_2,E_3\},\nonumber\\
&\{E_2,E_4\},\{E_2,E_5\},\{E_3,E_4\},\{E_3,E_5\},\{E_4,E_5\}\Big\}.\nonumber
\end{align}
 This joint measurability structure therefore corresponds to a complete graph with five vertices (Figure \ref{5complete}). 
\begin{figure}[H]
\centering
\includegraphics[scale=0.162]{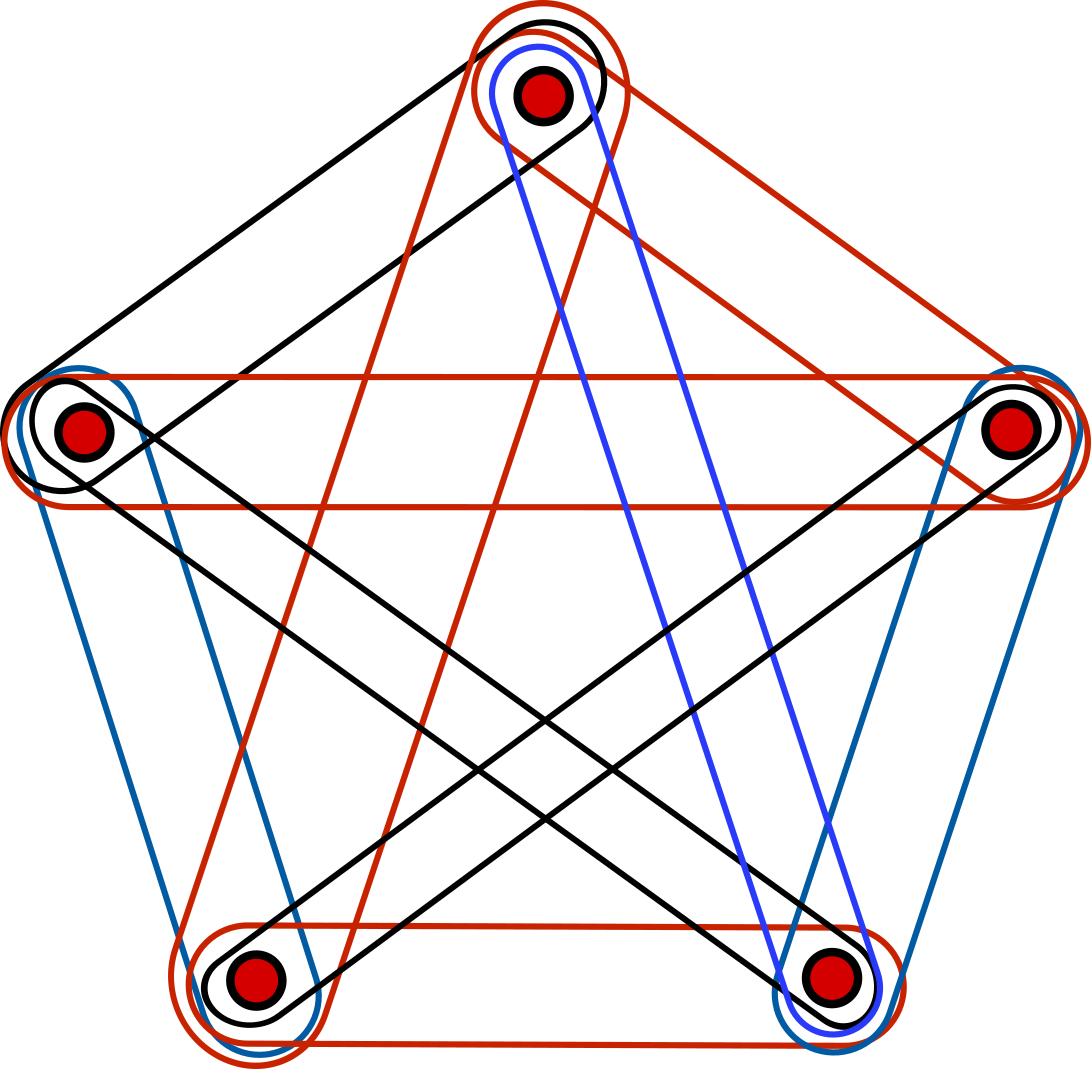}
\caption{Joint measurability structure represented by a complete graph with five vertices.}
\label{5complete}
\end{figure}
\end{example}
\begin{example}\upshape\label{ex5123}
 By Corollary \ref{3unb}, $[E_1,E_2,E_3]$ is compatible if and only if
\begin{equation}
\eta\leq\frac{1}{2\sin\frac{\pi}{10}+\cos\frac{\pi}{5}}=\frac{1}{11}(1+3\sqrt{5})\approx0.70075,
\end{equation}
while $[E_1,E_2,E_4]$ is incompatible by the same corollary if and only if
\begin{align}
\eta&>\frac{1}{\sin\frac{\pi}{10}+\sin\frac{\pi}{5}+\cos\frac{3\pi}{10}}\nonumber\\
&=\frac{4}{\sqrt{5}-1+2\sqrt{10-2\sqrt{5}}}\approx0.67359.
\end{align}
So, if and only if $$\displaystyle \eta\in\left(\frac{4}{\sqrt{5}-1+2\sqrt{10-2\sqrt{5}}},\frac{1}{11}(1+3\sqrt{5})\right],$$ we have the joint measurability structure given by  $$\displaystyle\Big\{\{1,2,3\},\{2,3,4\}, \{3,4,5\},\{4,5,1\},\{5,1,2\}\Big\},$$ represented by Figure \ref{c5123}.

\begin{figure}[H]
\centering
\includegraphics[scale=0.162]{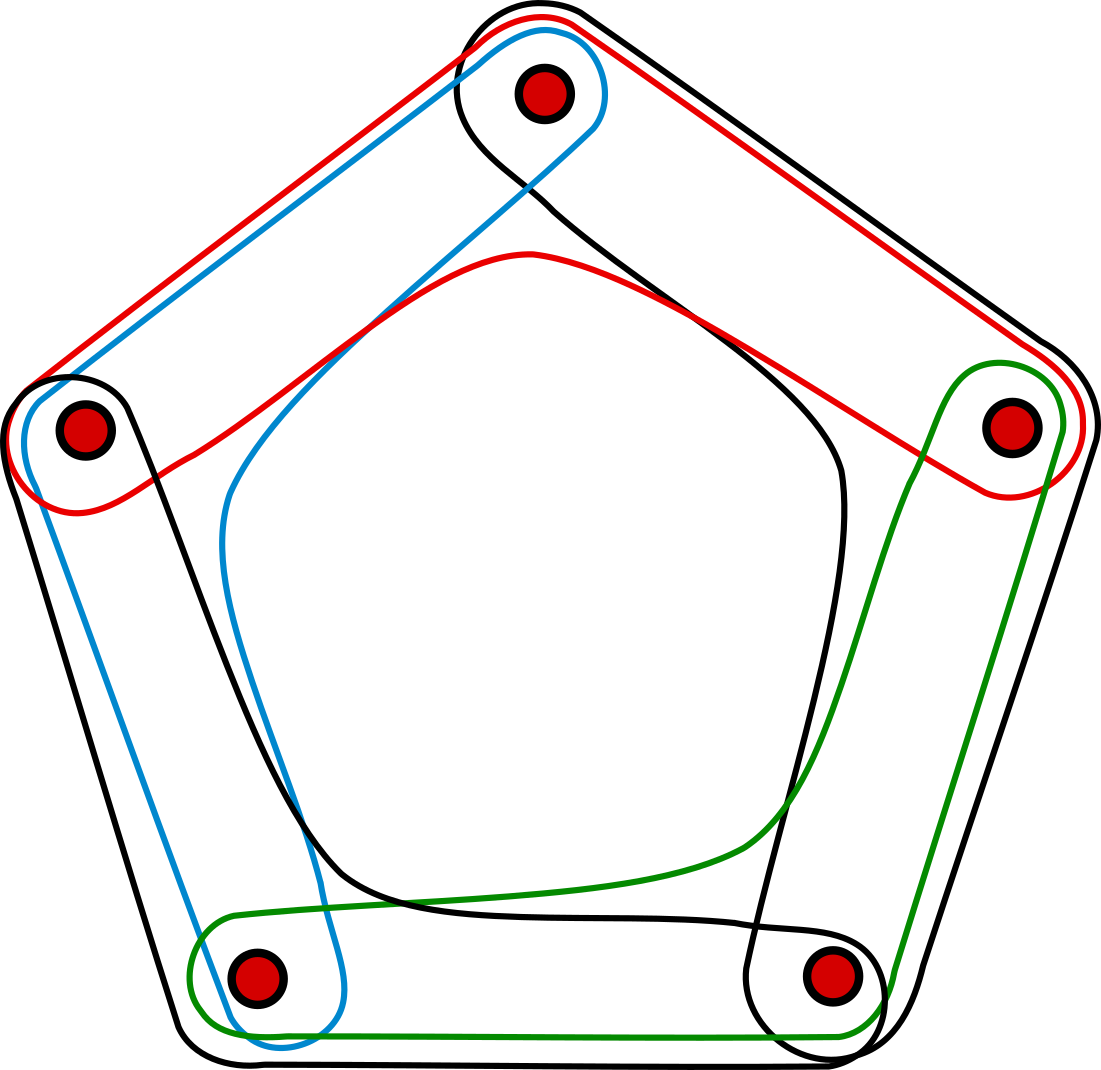}
\caption{Joint measurability structure $\{\{1,2,3\},\{2,3,4\},$ $\{3,4,5\},\{4,5,1\},\{5,1,2\}\}$ realized by five planar symmetric POVMs.}
\label{c5123}
\end{figure}

\end{example}
\begin{example}\upshape\label{ex5speck}
5-Specker scenario is realized for $\displaystyle\eta\in\left(\frac{1}{5}+\frac{1}{\sqrt{5}},\frac{4}{3(\sqrt{5}-1)+\sqrt{10-2\sqrt{5}}}\right]$, according to Corollary \ref{N-Specker's}. Remember that this is only a sufficient condition -- the upper bound is only known to be sufficient for $4$-way compatibility (cf. Lemma \ref{MtuplePlanSymjmc} and Conjecture \ref{conj1}), even though the lower bound is necessary and sufficient for $5$-way incompatibility (Theorem \ref{plansymjmc}).
\begin{figure}[H]
\centering
\includegraphics[scale=0.125]{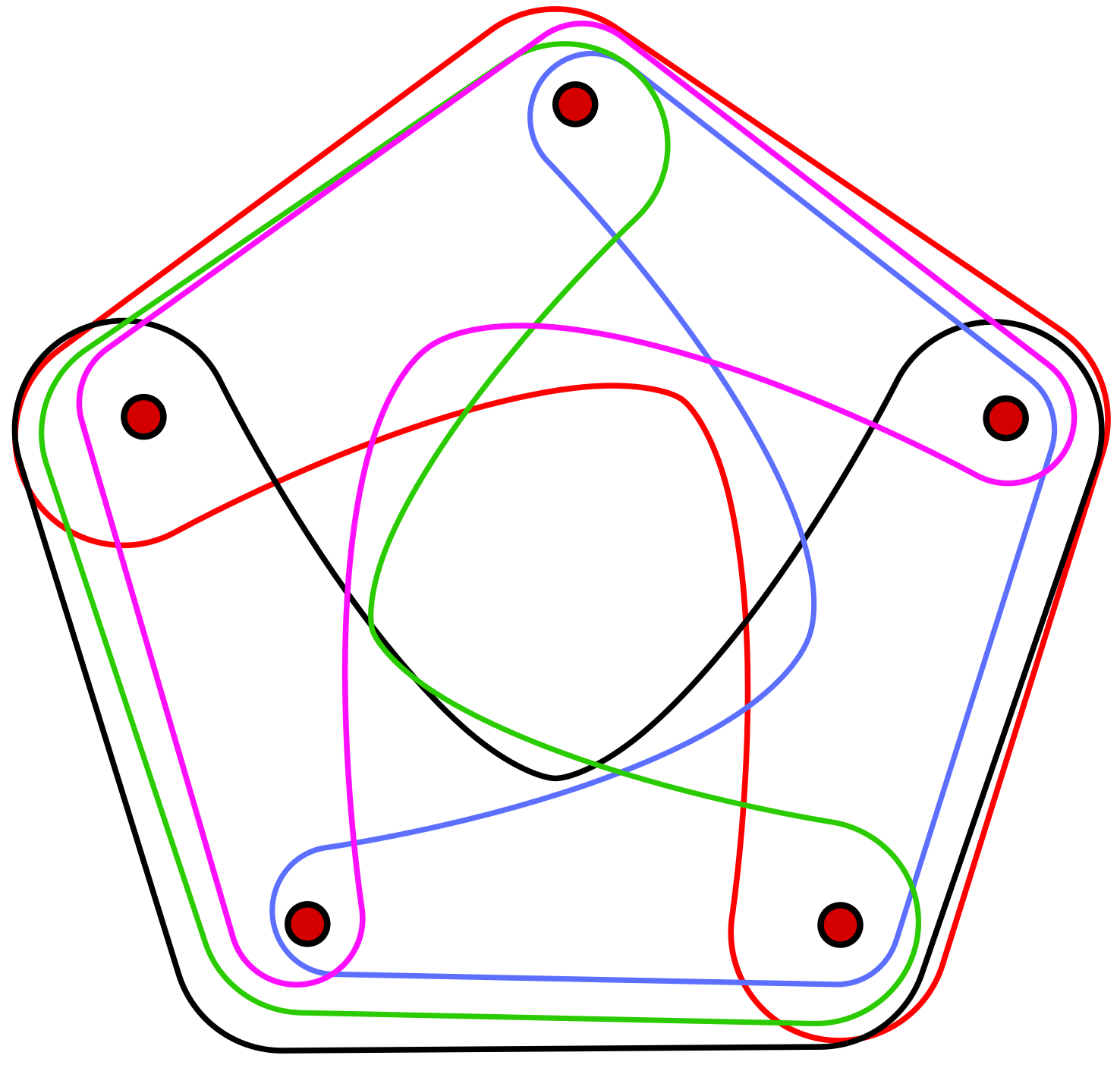}
\caption{5-Specker scenario}
\label{misc5speck}
\end{figure} 
\end{example}
\subsubsection{$N=6$ planar symmetric POVMs}
\begin{example}\upshape
According to Proposition \ref{ncycletheorem}, a $6$-cycle scenario (Fig.~\ref{fig6cycl}) is realized if and only if $\displaystyle\eta\in\left(\sqrt{3}-1,\sqrt{\frac{2}{3}}\right]$.
\begin{figure}[H]
\centering
\includegraphics[scale=0.146]{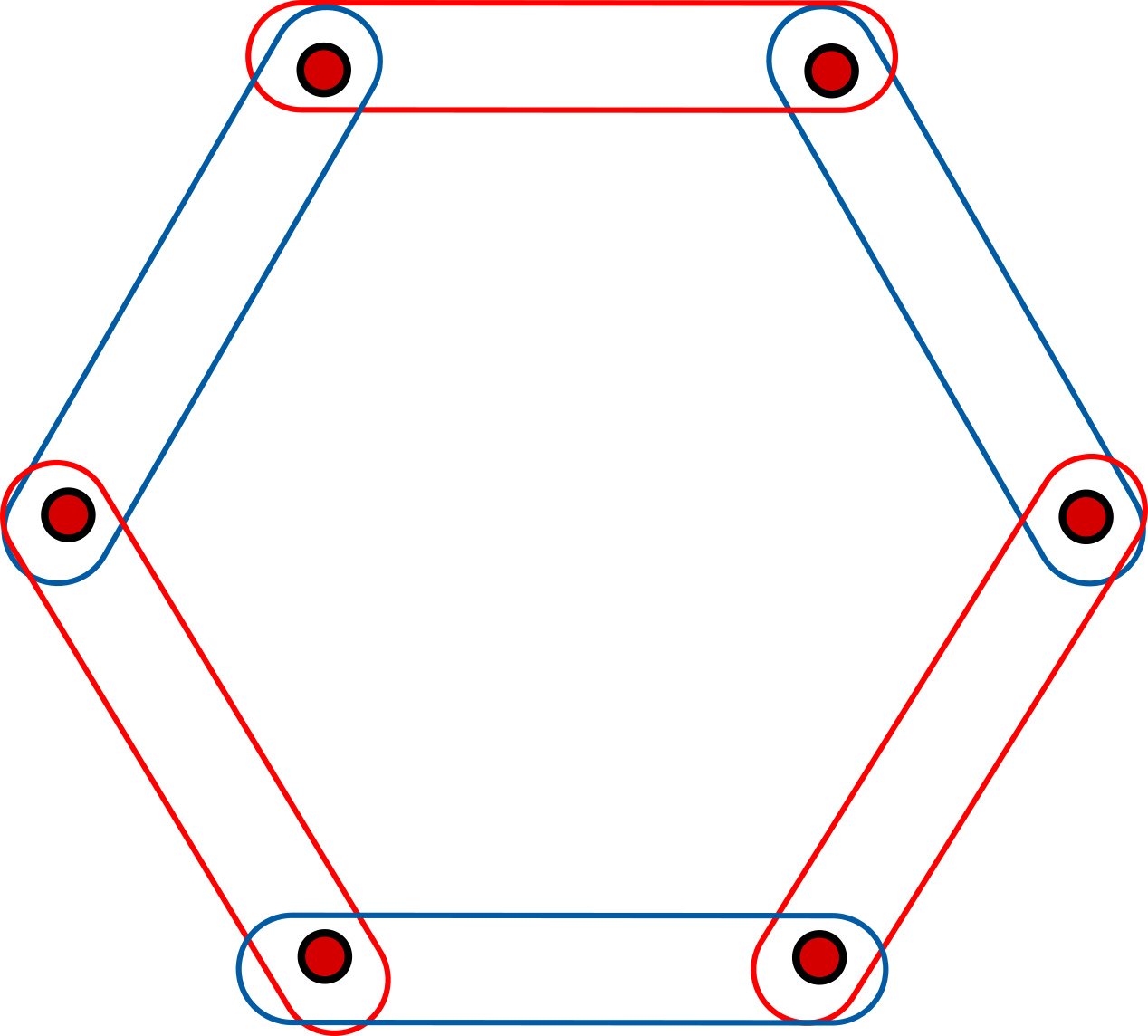}
\caption{6-cycle scenario}
\label{fig6cycl}
\end{figure}
\end{example}
\begin{example}\upshape\label{ex61213}
$[E_1,E_2]$ and $[E_1,E_3]$ are simultaneously compatible if and only if 
\begin{equation}
\eta\leq\frac{1}{\sin\frac{\pi}{6}+\cos\frac{\pi}{6}}=\sqrt{3}-1\approx0.73205,
\end{equation}
while $[E_1,E_2,E_3]$ is incompatible if and only if 
\begin{equation}
\eta>\frac{1}{2\sin\frac{\pi}{12}+\cos\frac{\pi}{6}}=\frac{2}{\sqrt{6}+\sqrt{3}-\sqrt{2}}\approx0.72272.
\end{equation}
The joint measurability structure 
\begin{align}
\Big\{&\{1,2\},\{1,3\},\{2,3\},\{2,4\},\{3,4\},\{3,5\},\{4,5\},\{4,6\},\nonumber\\
&\{5,6\},\{5,1\},\{6,1\},\{6,2\}\Big\}\nonumber
\end{align}
 is then realized if and only if
$\eta\in\left(\frac{2}{\sqrt{6}+\sqrt{3}-\sqrt{2}},\sqrt{3}-1\right]$.
\begin{figure}[H]
\centering
\includegraphics[scale=0.146]{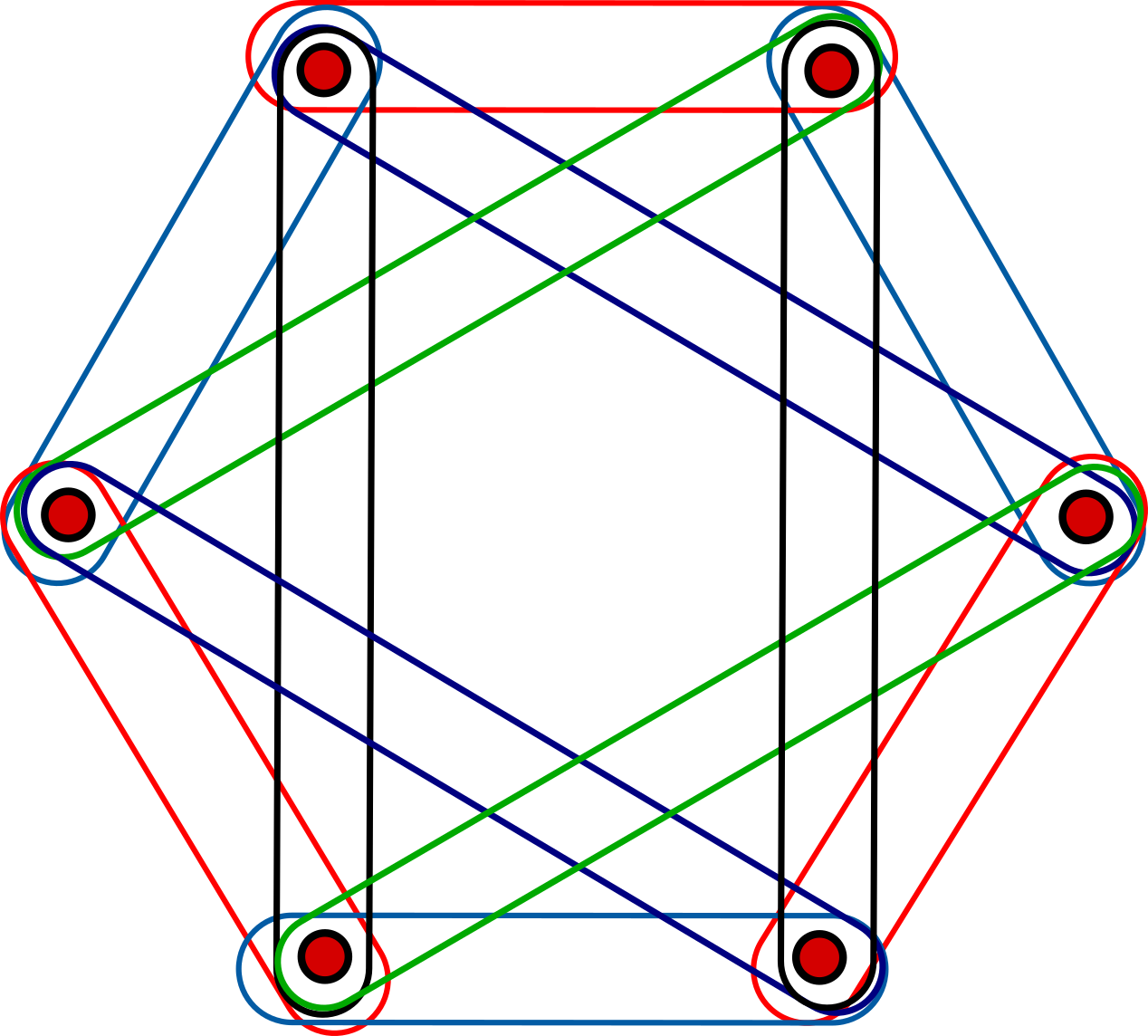}
\caption{Joint measurability structure $\Big\{\{1,2\},\{1,3\},\{2,3\},$ $\{2,4\},\{3,4\},\{3,5\},\{4,5\},\{4,6\},\{5,6\},$ $\{5,1\},\{6,1\},\{6,2\}\Big\}$.}
\end{figure}
\end{example}
\begin{example}\upshape\label{ex6123}
$[E_1,E_2,E_3]$ is compatible if and only if 
\begin{equation}
\eta\leq\frac{1}{2\sin\frac{\pi}{12}+\cos\frac{\pi}{6}}=\frac{2}{\sqrt{6}+\sqrt{3}-\sqrt{2}}\approx0.72272,
\end{equation}
while $[E_1,E_4]$ is incompatible if and only if 
\begin{equation}
\eta>\frac{1}{\sin\frac{\pi}{4}+\cos\frac{\pi}{4}}=\frac{\sqrt{2}}{2}\approx0.70711.
\end{equation}
This means that the compatibility structure $\Big\{\{1,2,3\},\{2,3,4\},\{3,4,5\},\{4,5,6\},\{5,6,1\},\{6,1,2\}\Big\}$ is realized if and only if $\eta\in\left(\frac{\sqrt{2}}{2},\frac{2}{\sqrt{6}+\sqrt{3}-\sqrt{2}}\right]$.
\begin{figure}[H]
\centering
\includegraphics[scale=0.146]{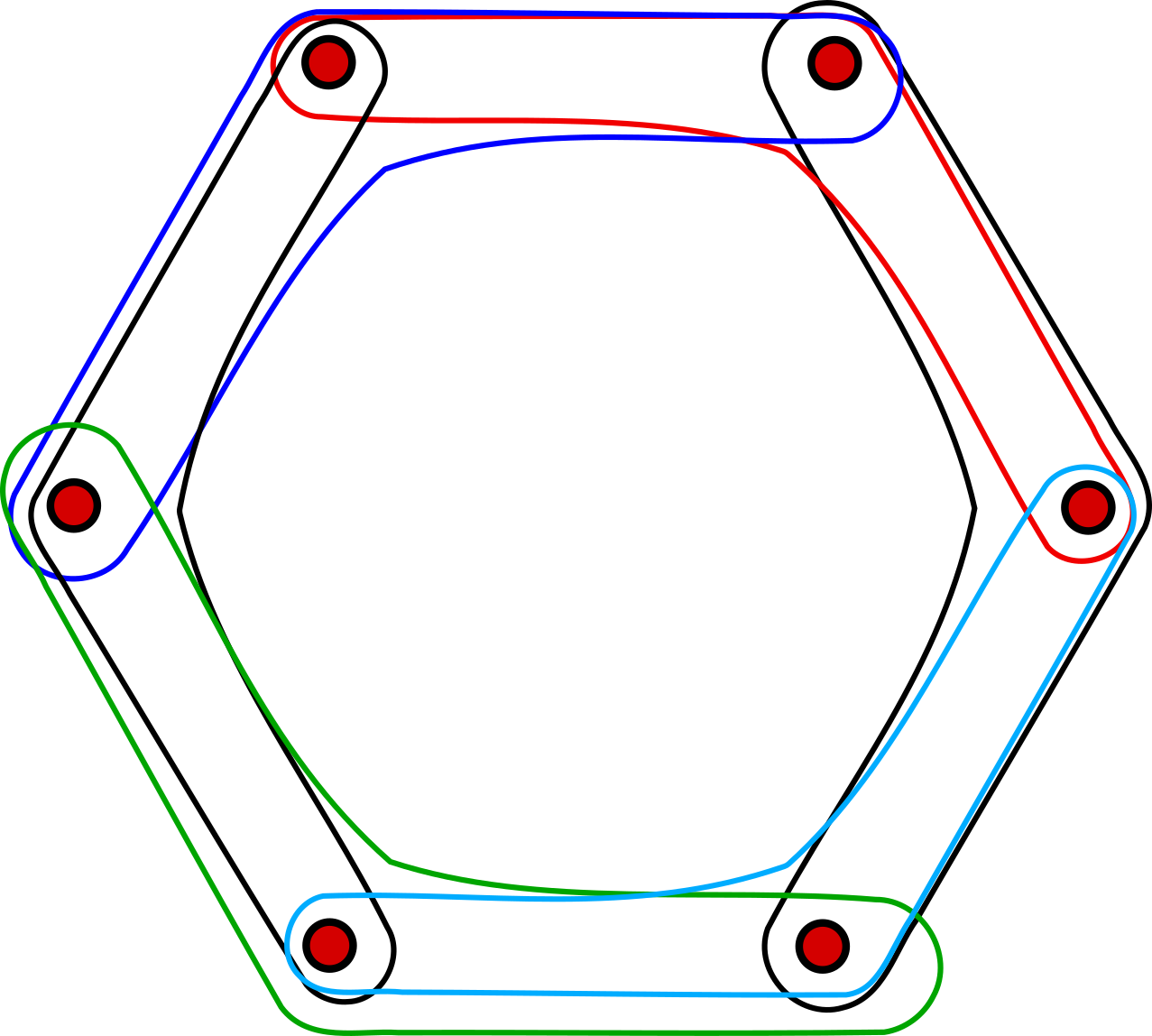}
\caption{Joint measurability structure $\Big\{\{1,2,3\},\{2,3,4\},$ $\{3,4,5\},\{4,5,6\},\{5,6,1\},\{6,1,2\}\Big\}$.}
\end{figure}
\end{example}
\begin{example}\upshape\label{ex612314}
$[E_1,E_2,E_3]$ and $[E_1,E_4]$ are compatible if and only if
\begin{equation}
\eta\leq\frac{1}{\sin\frac{\pi}{4}+\cos\frac{\pi}{4}}=\frac{\sqrt{2}}{2}\approx0.70711,
\end{equation}
while $[E_1,E_2,E_4]$ is incompatible if and only if 
\begin{equation}
\eta>\frac{1}{\sin\frac{\pi}{12}+\sin\frac{\pi}{6}+\cos\frac{\pi}{4}}=\frac{4}{2+\sqrt{2}+\sqrt{6}}\approx0.68236,
\end{equation}
so the joint measurability structure 
\begin{align}
\Big\{&\{1,2,3\},\{2,3,4\},\{3,4,5\},\{4,5,6\},\{5,6,1\},\{6,1,2\},\nonumber\\
&\{1,4\},\{2,5\},\{3,6\}\Big\}
\end{align}
is realized for and only for $\eta\in\left(\frac{4}{2+\sqrt{2}+\sqrt{6}},\frac{\sqrt{2}}{2}\right]$.

\begin{figure}[H]
\centering
\includegraphics[scale=0.146]{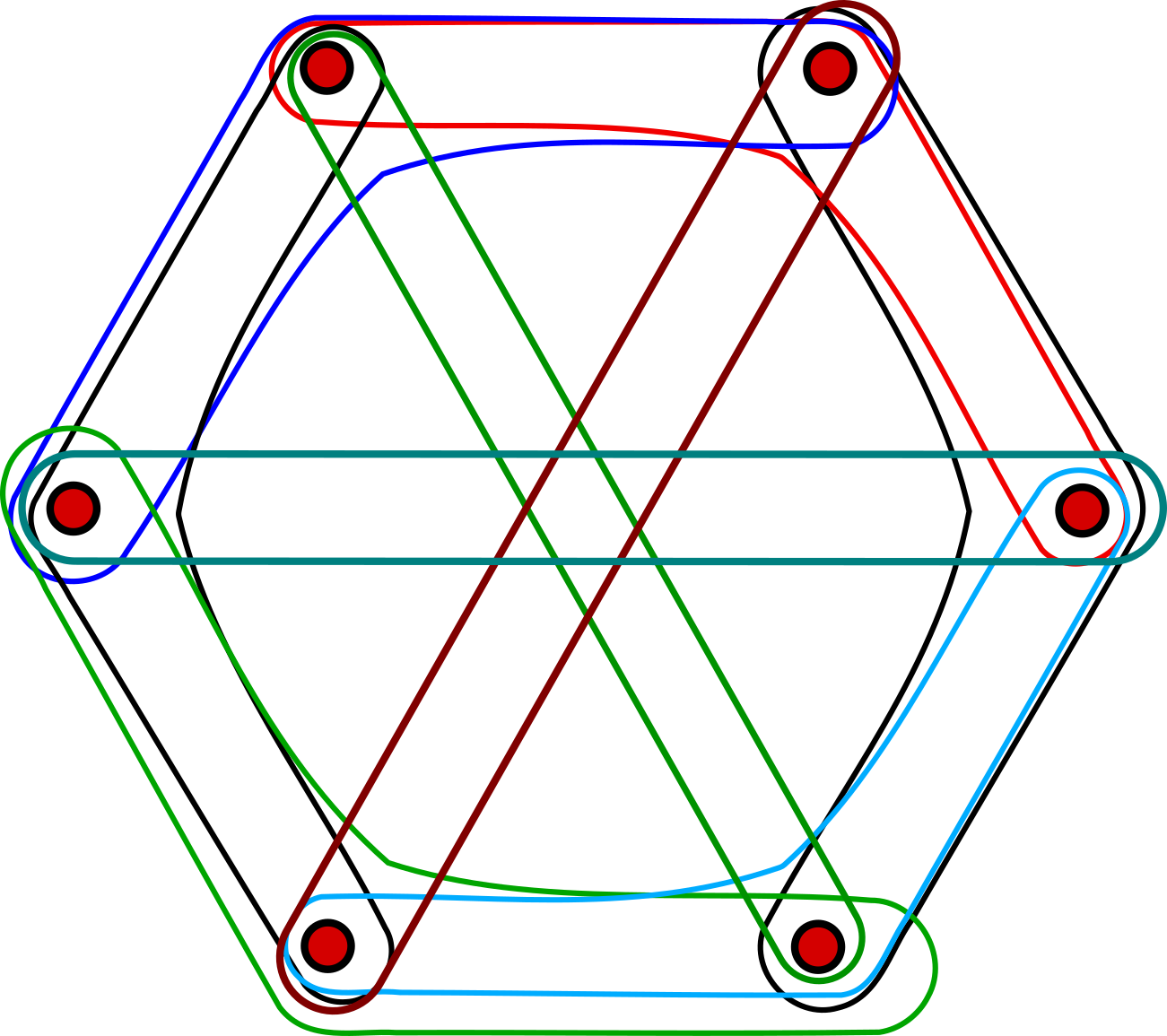}
\caption{Joint measurability structure  $\{\{1,2,3\},$ $\{2,3,4\},\{3,4,5\},\{4,5,6\},\{5,6,1\},\{6,1,2\},\{1,4\},\{2,5\},\{3,6\}\}$.}
\end{figure}
\end{example}
\begin{example}\upshape
According to Corollary \ref{N-Specker's}, 6-Specker scenario is realized for $\eta\in\left(\frac{\sqrt{2+\sqrt{3}}}{3},\frac{2}{1-2\sqrt{2}+2\sqrt{6}}\right]\approx(0.64396, 0.65135]$. Again, the upper bound has to be respected while it is only enough to respect the lower bound so this is only a sufficient condition.
\begin{figure}[H]
\centering
\includegraphics[scale=0.146]{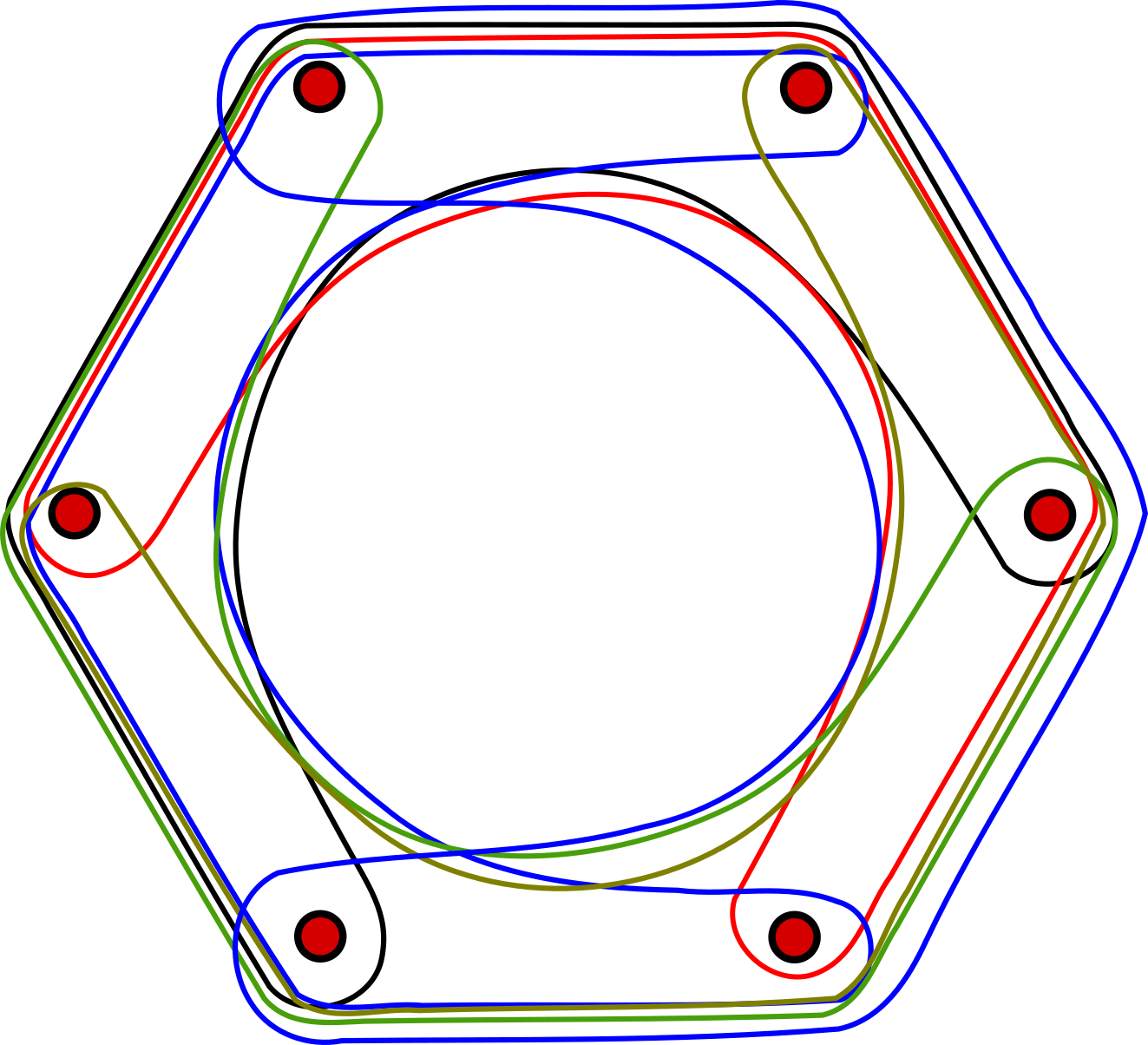}
\caption{6-Specker scenario hypergraph.}
\end{figure}
\end{example}
We could continue enumerating examples for $N=7,8,\ldots$. By continuously decreasing $\eta$ (i.e., by making the measurements more noisy) from $1$ to $0$, we go from $N$ incompatible POVMs to $N$ compatible ones, passing through various joint measurability structures in between. In a sense, we get more and more compatibility as we add noise. For $N=4$, we know the exact compatibility structure for each $\eta\in[0,1]$, which we can verify passing through examples for $4$ POVMs in the next subsection. This follows from the fact that Conjecture \ref{conj1} holds for $N=4$. 
However, passing through the examples for $N=5$, we see that we cannot infer the exact compatibility relation for $\displaystyle\left(\frac{4}{3(\sqrt{5}-1)+\sqrt{10-2\sqrt{5}}},\frac{4}{\sqrt{5}-1+2\sqrt{10-2\sqrt{5}}}\right]\approx\\(0.66014,0.673588]$ (cf. Examples \ref{ex5123} and \ref{ex5speck}). Its hypergraph is somewhere ``in between'' Figs.~\ref{c5123} and \ref{misc5speck}, and is just a $(5,3)$-compatibility hypergraph at least for $\eta$ in some subset of the range $(0.66014,0.673588]$. Namely, the presented upper bound is both necessary and sufficient for $3$-way joint measurability since then and only then both $[E_1,E_2,E_3]$ and $[E_1,E_2,E_4]$ are compatible at the same time. However, the lower bound is not sufficient (but only necessary) to ensure $4$-way incompatibility. In other words, we don't know if and  when (i.e., at which value of $\eta$), as we add noise (i.e. decrease $\eta$ from the upper towards the lower bound), the joint measurability structure changes from a $(5,3)$-compatible scenario to a $5$-Specker scenario. A similar situation arises for $N=6,7,8,\ldots$ and so on. If the conjectured necessity of Eq.~\eqref{l3} (i.e., Conjecture \ref{conj1}) holds, then we can fill in all of these gaps in understanding the joint measurability structure for arbitrary $N$.
\subsection{All joint measurability structures with $4$ vertices are realizable with binary qubit POVMs}\label{subsec5_2}
We first enumerate by hand all unlabeled hypergraphs (representing all conceivable joint measurability structures) that can be realized with $4$ vertices. There are exactly $20$ of them and they are given in Figure \ref{4hgs}. In every case, with the exception of number $6$, it will be possible to realize the desired joint measurability structure as a sub-structure of some larger joint measurability structure realized by $N\geq4$ planar symmetric POVMs. We discuss each joint measurability structure of Figure \ref{4hgs} in turn:

\begin{figure*}
\centering
\includegraphics[scale=0.24]{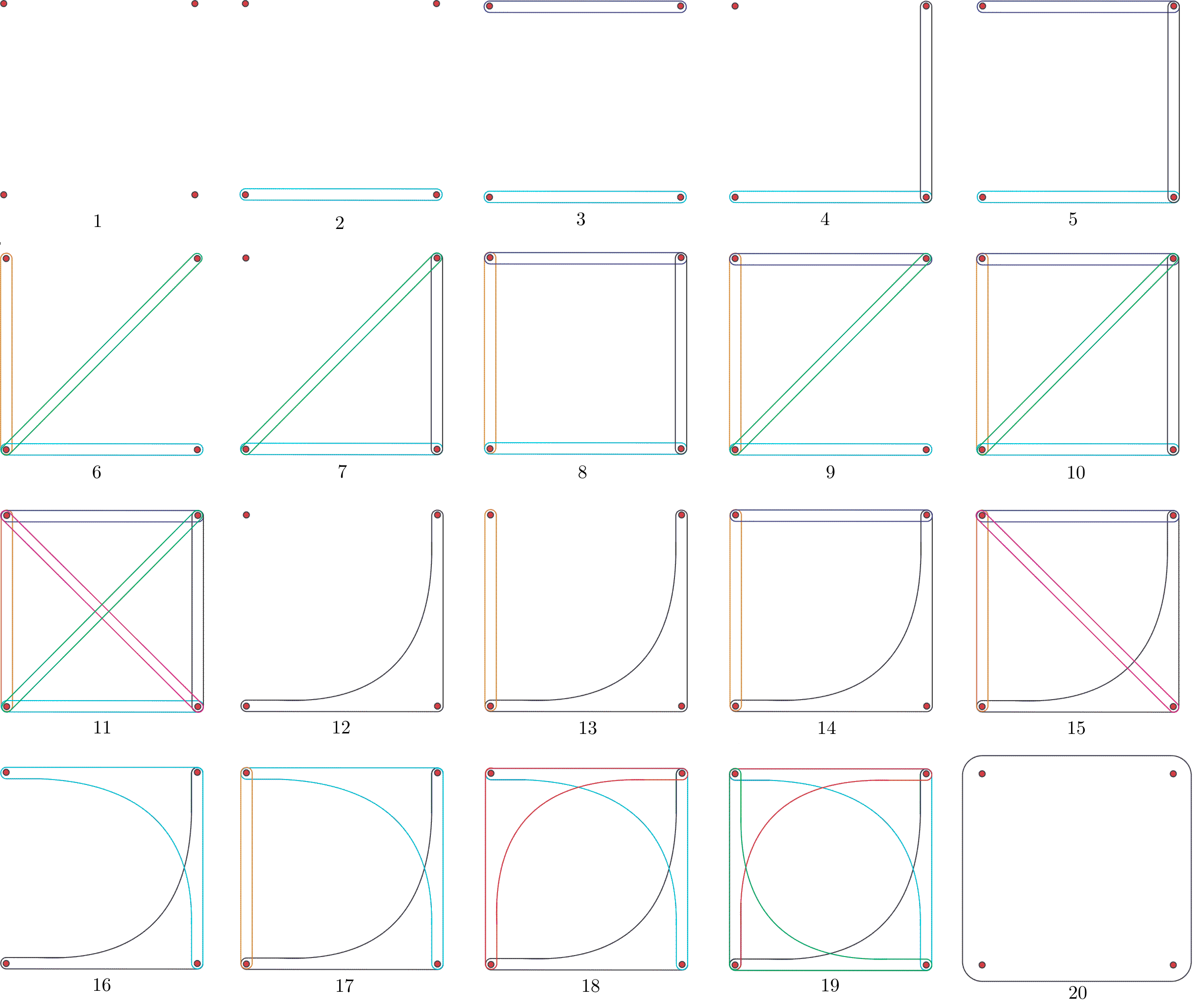}
\caption{All possible unlabelled hypergraphs on 4 vertices.}
\label{4hgs}
\end{figure*}

\begin{enumerate}

\item This is a set of $4$ incompatible POVMs. We can realize it with $4$ planar symmetric POVMs with sharpness
\begin{equation}
\eta>\frac{1}{4\sin\frac{\pi}{8}}\approx 0.65328.
\end{equation}

\item Consider $N=7$ planar symmetric POVMs. They comprise a $7$-cycle for 
\begin{align}
\eta&\in\left(\frac{1}{\sin\frac{\pi}{7}+\cos\frac{\pi}{7}},\frac{1}{\sin\frac{\pi}{14}+\cos\frac{\pi}{14}}\right]\nonumber\\
&\approx\left(0.74915,0.83510\right].
\end{align}
Therefore, $[E_1,E_2]$ is compatible, while $[E_1,E_3]$ and $[E_1,E_4]$ are incompatible. Therefore, the desired hypergraph is realized for $\eta$ in the given interval on the subset $\{E_1,E_2,E_4,E_6\}$.

\item Consider $N=6$ planar symmetric POVMs. They constitute a $6$-cycle for 
\begin{align}
\eta&\in\left(\frac{1}{\sin\frac{\pi}{12}+\cos\frac{\pi}{12}},\frac{1}{\sin\frac{\pi}{6}+\cos\frac{\pi}{6}}\right]\nonumber\\
&\approx(0.73206,0.81649].
\end{align}
Since $[E_1,E_2]$ is compatible and $[E_1,E_4]$ is incompatible the desired hybergraph is realized on $\{E_1,E_2,E_4,E_5\}$ for $\eta$ in the given range.

\item Consider $N=6$ planar symmetric POVMs that constitute a $6$-cycle like in the previous example but now take the following POVMs to realize the desired hypergraph $\{E_1,E_2,E_3,E_5\}$.

\item Consider $N=5$ planar symmetric POVMs. They make a $5$-cycle for 
\begin{align}
\eta&\in\left(\frac{1}{\sin\frac{\pi}{5}+\cos\frac{\pi}{5}},\frac{1}{\sin\frac{\pi}{10}+\cos\frac{\pi}{10}}\right]\nonumber\\
&\approx(0.71593,0.79360].
\end{align}
It is easy to see that the desired hypergraph is realized on $\{E_1,E_2,E_3,E_4\}$ for $\eta$ in the given range.

\item It turns out that this hypergraph cannot be realized with coplanar qubit POVMs with the same purity. We study it more thoroughly in the next section, where we explicitly prove this claim (see Lemma \ref{limit}) but we also show that it can still be realized with unbiased binary qubit POVMs (see Examples \ref{exnotsp} and \ref{exnoncopl}).

\item Consider $N=8$ planar symmetric POVMs. $[E_1,E_2]$ and $[E_1,E_3]$ are both compatible for 
\begin{equation}
\eta\leq\frac{1}{\sin\frac{\pi}{8}+\cos\frac{\pi}{8}}=0.76536,
\end{equation}
while $[E_1,E_2,E_3]$ and $[E_1,E_6]$, $[E_2,E_6]$  are incompatible for 
\begin{equation}
\eta>\frac{1}{2\sin\frac{\pi}{16}+\cos\frac{\pi}{8}}=0.76100.
\end{equation}
Therefore, for 
\begin{align}
\eta&\in\left(\frac{1}{2\sin\frac{\pi}{16}+\cos\frac{\pi}{8},\frac{1}{\sin\frac{\pi}{8}+\cos\frac{\pi}{8}}}\right]\nonumber\\
&\approx(0.76100,0.76536],
\end{align}
we have realized the desired hypergraph on $\{E_1,E_2,E_3,E_6\}$

\item This is a $4$-cycle scenario and it is already described in the previous section how it can be realized with $N=4$ planar symmetric POVMs.

\item Consider $N=7$ planar symmetric POVMs. $[E_1,E_2]$, and $[E_1,E_3]$ are compatible for
\begin{equation}
\eta\leq\frac{1}{\sin\frac{\pi}{7}+\cos\frac{\pi}{7}}=0.74914,
\end{equation}
while $[E_1,E_2,E_3]$ and $[E_1,E_4]$ are incompatible for 
\begin{equation}
\eta>\frac{1}{2\sin\frac{\pi}{14}+\cos\frac{\pi}{7}}=0.74294.
\end{equation}
So for 
\begin{align}
\eta&\in\left(\frac{1}{2\sin\frac{\pi}{14},\cos\frac{\pi}{7}},\frac{1}{\sin\frac{\pi}{7}+\cos\frac{\pi}{7}}\right]\nonumber\\
&\approx(0.74294,0.74914]
\end{align}
we have realized the desired hypergraph on $\{E_1,E_2,E_3,E_6\}$.

\item Take $N=6$ planar symmetric POVMs. $[E_1,E_2]$ and $[E_1,E_3]$ are both compatible for 
\begin{equation}
\eta\leq\frac{1}{\sin\frac{\pi}{6}+\cos\frac{\pi}{6}}=0.73205.
\end{equation}
On the other hand $[E_1,E_2,E_3]$ and $[E_1,E_4]$ are incompatible for
\begin{equation}
\eta>0.72272.
\end{equation}
This means that for 
\begin{align}
\eta&\in\left(\frac{1}{2\sin\frac{\pi}{12}+\cos\frac{\pi}{6}},\frac{1}{\sin\frac{\pi}{6}+\cos\frac{\pi}{6}}\right]\nonumber\\
&\approx(0.72272,0.73205]
\end{align}
the desired hypergraph is realized on $\{E_1,E_2,E_3,E_5\}$.

\item This is a complete graph for $N=4$ and we have already shown its realizability with qubit POVMs in the previous section. 

\item Consider $N=8$ planar symmetric POVMs. $[E_1,E_2,E_3]$ is compatible for 
\begin{equation}
\eta\leq\frac{1}{2\sin\frac{\pi}{16}+\cos\frac{\pi}{8}}=0.76100,
\end{equation}
while $[E_1,E_4]$ is incompatible for 
\begin{equation}
\eta>\frac{1}{\sin\frac{3\pi}{16}+\cos\frac{3\pi}{16}}=0.72096.
\end{equation}
Therefore, we have realized the desired hypergraph on $\{E_1,E_2,E_3,E_6\}$ for 
\begin{align}
\eta&\in\left(\frac{1}{\sin\frac{3\pi}{16}+\cos\frac{3\pi}{16}},\frac{1}{2\sin\frac{\pi}{16}+\cos\frac{\pi}{8}}\right]\nonumber\\
&\approx(0.72096,0.76100].
\end{align}

\item Consider $N=7$ planar symmetric POVMs. $[E_1,E_2,E_3]$ is compatible for 
\begin{equation}
\eta\leq\frac{1}{2\sin\frac{\pi}{14}+\cos\frac{\pi}{7}}=0.74293,
\end{equation}
while $[E_1,E_4]$ is incompatible for 
\begin{equation}
\eta>\frac{1}{\sin\frac{3\pi}{14}+\cos\frac{3\pi}{14}}=0.71159.
\end{equation}
So for 
\begin{align}
\eta&\in\left(\frac{1}{\sin\frac{3\pi}{14}+\cos\frac{3\pi}{14}},\frac{1}{2\sin\frac{\pi}{14}+\cos\frac{\pi}{7}}\right]\nonumber\\
&\approx(0.71159,0.74293],
\end{align}
we have realized the desired hypergraph on 
$\{E_1,E_2,E_3,E_6\}$.

\item Consider $N=6$ planar symmetric POVMs. $[E_1,E_2,E_3]$ is compatible for 
\begin{equation}
\eta\leq\frac{1}{2\sin\frac{\pi}{12}+\cos\frac{\pi}{6}}=0.72271,
\end{equation}
while $[E_1,E_4]$ and $[E_1,E_3,E_5]$ are incompatible for 
\begin{equation}
\eta>\frac{1}{\sin\frac{3\pi}{12}+\cos\frac{3\pi}{12}}=0.70711.
\end{equation}
Therefore, we have realized the desired hypergraph for
\begin{align}
\eta&\in\left(\frac{1}{\sin\frac{3\pi}{12}+\cos\frac{3\pi}{12}},\frac{1}{2\sin\frac{\pi}{12}+\cos\frac{\pi}{6}}\right]\nonumber\\
&\approx(0.70711,0.72271],
\end{align}
on $\{E_1,E_2,E_3,E_5\}$.
\item Consider $N=6$ planar symmetric POVMs. $[E_1,E_2,E_3]$, and $[E_1,E_4]$ are both compatible for  
\begin{equation}
\eta\leq\frac{1}{\sin\frac{3\pi}{12}+\cos\frac{3\pi}{12}}=0.70710,
\end{equation}
while $[E_1,E_2,E_4]$ and $[E_1,E_3,E_5]$ are both incompatible for
\begin{equation}
\eta>\frac{1}{\sin\frac{\pi}{12}+\sin\frac{\pi}{6}+\cos\frac{3\pi}{12}}=0.68217.
\end{equation}
Therefore, we have realized the desired hypergraph on $\{E_1,E_2,E_3,E_5\}$ for
\begin{align}
\eta&\in\left(\frac{1}{\sin\frac{\pi}{12}+\sin\frac{\pi}{6}+\cos\frac{3\pi}{12}},\frac{1}{\sin\frac{3\pi}{12}+\cos\frac{3\pi}{12}}\right]\nonumber\\
&\approx(0.68217,0.70710].
\end{align}
\item Consider $N=6$ planar symmetric POVMs with $\eta$ same as for the hypergraph 14, but now take $\{E_1,E_2,E_3,E_4\}$ to realize the desired hypergraph.
\item Consider $N=5$ planar symmetric POVMs. $[E_1,E_2,E_3]$ is compatible for 
\begin{equation}
\eta\leq\frac{1}{2\sin\frac{\pi}{10}+\cos\frac{\pi}{5}}=0.70074,
\end{equation}
while $[E_1,E_3,E_4]$ in incompatible for
\begin{equation}
\eta>\frac{1}{\sin\frac{\pi}{5}+\sin\frac{\pi}{10}+\cos\frac{3\pi}{5}}=0.67359.
\end{equation}
Therefore, taking
\begin{align}
\eta&\in\left(\frac{1}{\sin\frac{\pi}{5}+\sin\frac{\pi}{10}+\cos\frac{3\pi}{5}},\frac{1}{2\sin\frac{\pi}{10}+\cos\frac{\pi}{5}}\right]\nonumber\\
&\approx(0.67359,0.70074],
\end{align}
we have realized the desired hypergraph on $\{E_1,E_2,E_3,E_4\}$.

\item Consider $N=7$ planar symmetric POVMs. $[E_1,E_2,E_3]$ and $[E_1,E_2,E_4]$ are both compatible for
\begin{equation}
\eta\leq\frac{1}{\sin\frac{\pi}{12}+\sin\frac{\pi}{6}+\cos\frac{3\pi}{12}}=0.68216,
\end{equation}
while $[E_1,E_3,E_5]$ is incompatible for
\begin{equation}
\eta>\frac{1}{2\sin\frac{\pi}{5}+\cos\frac{2\pi}{5}}=0.67359.
\end{equation}
Therefore, for
\begin{align}
\eta&\in\left(\frac{1}{2\sin\frac{\pi}{5}+\cos\frac{2\pi}{5}},\frac{1}{\sin\frac{\pi}{12}+\sin\frac{\pi}{6}+\cos\frac{3\pi}{12}}\right]\nonumber\\
&\approx(0.67359,0.68216],
\end{align}
we have realized the desired hypergraph on $\{E_1,E_2,E_3,E_5\}$.

\item This hypergraph represents a $4$-Specker scenario that was already realized in the previous section. 
\item This is the hypergraph of $4$ compatible POVMs. We can just take $4$ compatible planar symmetric POVMs.
\end{enumerate}

\subsection{Limitations of hypergraph relizability with coplanar unbiased POVMs with the same purity}\label{subsec5_3}
\begin{lemma}\label{limit}
Not all joint measurability structures can be realized with unbiased binary coplanar qubit POVMs with the same purity.
\end{lemma}
\begin{proof}
We provide an example of a joint measurability structure that does not admit such a realization. The hypergraph for this example is shown in Fig.~\ref{count3}). This hypergraph is the same as the one in Fig.~\ref{4hgs} labelled by number $6$.
\begin{figure}[H]
\centering
\includegraphics[scale=0.335]{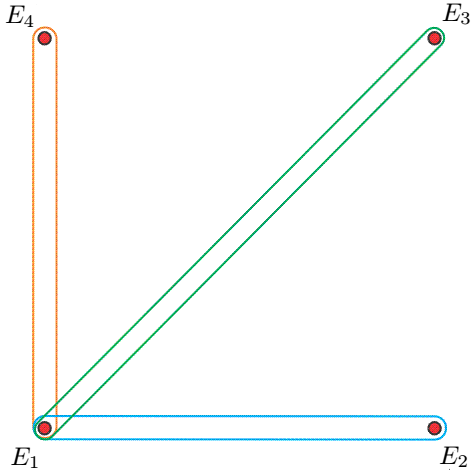}
\caption{Hypergraph representation of
(in)compatibility structure $\{\{E_1,E_2\},\{E_1,E_3\},\{E_1,E_4\}\}$.}
\label{count3}
\end{figure}
Suppose that this hypergraph can be realized with binary, unbiased, planar qubit POVMs with the same purity denoted by $\{E_1,E_2,E_3,E_4\}$. Without loss of generality, we can take the directions of their Bloch vectors as in Fig.~\ref{blochcount3}, i.e., that they are of the form
\begin{equation}
E_{k}(x_k)=\frac{1}{2}\left(I+x_k\eta\vec{n}_k\cdot\vec{\sigma}\right).
\end{equation}
\begin{figure}[H]
\centering
\includegraphics[scale=0.47]{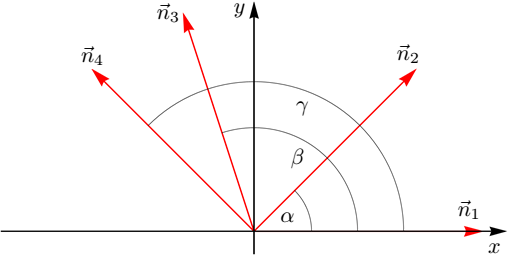}
\caption{Situation of Bloch vectors.}
\label{blochcount3}
\end{figure}
In this case we have that 
\begin{equation}
0<\alpha<\beta<\gamma<\pi.
\label{nejednakost1}
\end{equation}
According to Corollary \ref{Leta2}, the prescribed pairwise joint measurability relations require that following inequalities hold:
\begin{align}
&\eta\leq\frac{1}{\cos\frac{\alpha}{2}+\sin\frac{\alpha}{2}},\nonumber\\
&\eta\leq\frac{1}{\cos\frac{\beta}{2}+\sin\frac{\beta}{2}},\nonumber\\ &\eta\leq\frac{1}{\cos\frac{\gamma}{2}+\sin\frac{\gamma}{2}},\label{nejednakost2}
\end{align}
as well as
\begin{align}
&\eta>\frac{1}{\cos\frac{\beta-\alpha}{2}+\sin\frac{\beta-\alpha}{2}},\nonumber\\ &\eta>\frac{1}{\cos\frac{\gamma-\alpha}{2}+\sin\frac{\gamma-\alpha}{2}},\nonumber\\ &\eta>\frac{1}{\cos\frac{\gamma-\beta}{2}+\sin\frac{\gamma-\beta}{2}}
\label{nejednakost3}
\end{align}
Using first inequalities from both Eqs.~\eqref{nejednakost2} and \eqref{nejednakost3}, we get 
\begin{equation}
\cos\frac{\beta-\alpha}{2}+\sin\frac{\beta-\alpha}{2} >\cos\frac{\alpha}{2}+\sin\frac{\alpha}{2}
\end{equation}
Rearranging terms in this new inequality we find 
\begin{equation}
\cos\frac{\beta-\alpha}{2}-\cos\frac{\alpha}{2}>\sin\frac{\alpha}{2}-\sin\frac{\beta-\alpha}{2}.
\end{equation}
Applying sum to product rules this becomes 
\begin{equation}
\sin\frac{\beta}{4}\sin\frac{2\alpha-\beta}{4}>\sin\frac{2\alpha-\beta}{4}\cos\frac{\beta}{4}.
\label{nj11}
\end{equation}
Inequality \eqref{nj11} branches into two cases:
\begin{align}
\tan\frac{\beta}{4}&>1\quad\text{for}\quad 2\alpha>\beta,\quad\text{or}\\
\tan\frac{\beta}{4}&<1\quad\text{for}\quad 2\alpha<\beta,
\end{align}
which reduces to 
\begin{align}
\beta>&\pi\quad\text{for}\quad 2\alpha>\beta,\quad\text{or}\\
\beta<&\pi\quad\text{for}\quad 2\alpha<\beta.
\end{align}
We take only the possibility that is consistent with the constraint from Eq.~\eqref{nejednakost1} (in particular, $\beta<\pi$). Hence, we consider the case $2\alpha<\beta$. In the similar succession of steps combining the second relation in Eq.~\eqref{nejednakost3} with the first one in Eq.~\eqref{nejednakost2} as well as by combining the third relation from Eq.~\eqref{nejednakost3} with the second one from Eq.~\eqref{nejednakost2} we obtain $2\alpha<\gamma$ and $2\beta<\gamma$. In total, we obtain that the following constraints must hold:
\begin{equation}
2\alpha<\gamma,\quad 2\beta<\gamma, \quad 2\alpha<\beta.
\label{nejednakost4}
\end{equation}
Now if we combine the second relation from Eq.~\eqref{nejednakost3} with the third one from Eq.~\eqref{nejednakost2} we will get 
\begin{equation}
\cos\frac{\gamma-\alpha}{2}+\sin\frac{\gamma-\alpha}{2}>\cos\frac{\gamma}{2}+\sin\frac{\gamma}{2}.
\end{equation}
Rearranging this becomes:
\begin{equation}
\cos\frac{\gamma-\alpha}{2}-\cos\frac{\gamma}{2}>\sin\frac{\gamma}{2}-\sin\frac{\gamma-\alpha}{2}.
\end{equation}
Applying sum to product rules we get
\begin{equation}
\sin\frac{2\gamma-\alpha}{4}\sin\frac{\alpha}{4}>\sin\frac{\alpha}{4}\cos\frac{2\gamma-\alpha}{4}.
\end{equation}
Using Eq.~\ref{nejednakost1} i.e. that $\alpha\in(0,\pi)$ we can cancel the factor of $\sin\frac{\alpha}{4}$ knowing that it is positive so we further obtain 
\begin{equation}
\tan\frac{2\gamma-\alpha}{4}>1\Rightarrow 2\gamma-\alpha>\pi.
\end{equation}
Combining the first relation from Eq. \ref{nejednakost3} with the second one from Eq. \ref{nejednakost2} and as well as the third relations from both of them, we obtain in a similar succession of steps $2\beta-\alpha>\pi$ and $2\gamma-\beta>\pi$. In total this yields 
\begin{equation}
2\gamma-\alpha>\pi,\quad 2\beta-\alpha>\pi,\quad 2\gamma-\beta>\pi.
\label{nejednakost5}
\end{equation}
Now using the second relations in both Eqs. \eqref{nejednakost4} and \eqref{nejednakost5} we will obtain
\begin{equation}
\gamma-\alpha>\pi\Longrightarrow \gamma>\pi
\end{equation}
which is in contradiction with the constraint given by Eq.~\eqref{nejednakost1}. This means that the system of constraints listed in Eqs.~\eqref{nejednakost1}, \eqref{nejednakost2} and \eqref{nejednakost3} is inconsistent. Hence, the proposed joint measurability structure (Fig.~\ref{count3}) cannot be realized with unbiased binary coplanar qubit POVMs with the same purity.
\end{proof}
However, the hypergraph in Fig.~\ref{count3} can be realized with binary qubit POVMs if we relax some of the constraints. We provide two examples with two distinct ways of achieving this: in one of them we give up on the same purity for the POVMs and in the other one we allow one of them to not be coplanar with the others.
\begin{example}[Giving up on the same purity]\upshape\label{exnotsp} In this example we stay in an equatorial plane of the Bloch ball but we let the POVMs differ in their purity. Here it is enough to choose three POVMs, say $E_1, E_2$ and $E_3$ to have same purity $\eta$ such that the sets $\{E_1,E_2\}$ and $\{E_1,E_3\}$ are compatible but that the set $\{E_2,E_3\}$ is incompatible, and to assign some other purity $\mu$ to $E_4$. The strategy is then to arrange the Bloch vectors of $E_1$, $E_2$ and $E_3$ (using Corollary \ref{Leta2}) such that (in)compatibility relations prescribed in Fig.~\ref{count3} hold. We can then choose $E_4$ to be a measurement in the same direction as $E_1$ such that it's a projective measurement (setting $\mu=1$). This choice ensures that $E_4$ is incompatible with $E_2$ and $E_3$ but compatible with $E_1$, which is what we wanted. One of many possible choices is $E_k(\pm1)=\frac{1}{2}(I\pm\vec{a}_k\cdot\vec{\sigma})$, where $||\vec{a}_{1,2,3}||=\eta=\sqrt{\frac{2}{3}}$ and $||\vec{a}_4||=\mu=1$ with the Bloch vectors arranged as in Fig.~\ref{fexnotsp}.
\begin{figure}
\centering
\includegraphics[scale=0.43]{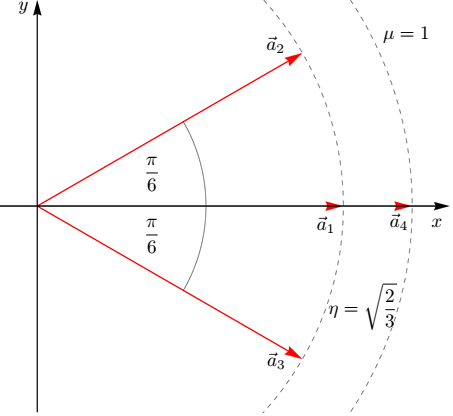}
\caption{Bloch vectors $\vec{a}_k$ corresponding to POVMs $E_k$.}
\label{fexnotsp}
\end{figure}
\end{example}
\begin{example}[Allowing non-coplanar POVMs]\upshape\label{exnoncopl}
Here we keep the same purity $\eta$ of four POVMs $\{E_1,E_2,E_3,E_4\}$ but allow for non-coplanar directions. The strategy is to put the Bloch vectors of $E_2$ and $E_3$ symmetrically in the equatorial $XY$ plane with respect to the Bloch vector of $E_1$ such that they make an angle $\alpha<\frac{\pi}{4}$ with it. Then we put the Bloch vector of $E_4$ in $XZ$ plane such that it makes the same angle $\alpha$ with the direction of the Bloch vector of $E_1$ (see Fig.~\ref{fexnoncopl}). The Bloch vector of $E_4$ then makes an angle $\phi$ with the Bloch vectors of $E_2$ and $E_3$ which satisfies $\alpha<\phi<2\alpha$. Therefore, the range for $\phi$ is a subset of $\left(0,\frac{\pi}{2}\right)$.
\begin{figure}
\centering
\includegraphics[scale=0.3]{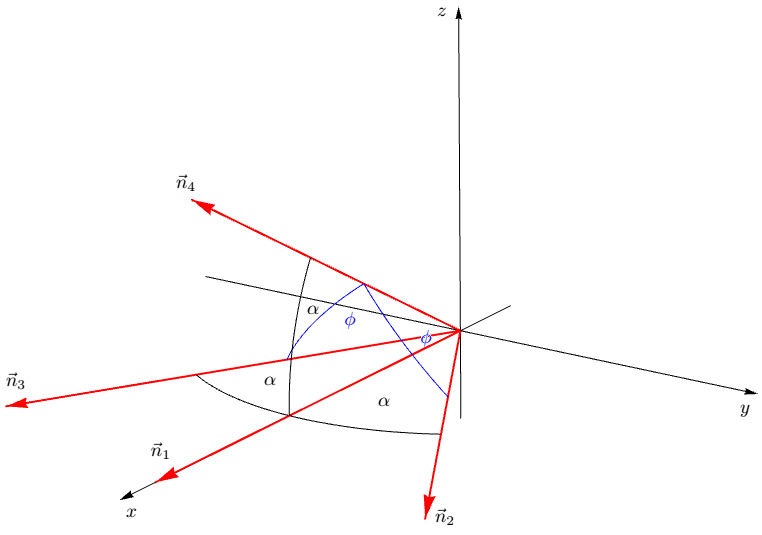}
\caption{Directions of Bloch vectors $\vec{n}_k$ corresponding to POVMs $E_k$ that are used in this example.}
\label{fexnoncopl}
\end{figure}
Now we search for $\eta$ such that the prescribed joint measurability relations in Fig. \ref{count3} are satisfied. According to Corollary \ref{Leta2}, we must have 
\begin{align}
\eta\leq\frac{1}{\sin\frac{\alpha}{2}+\cos\frac{\alpha}{2}},\label{noncopl0}
\end{align} 
to have pairwise compatibility of $E_1$ with each of $E_2$, $E_3$ and $E_4$, while at the same time we must demand
\begin{align}
\eta&>\frac{1}{\sin\frac{\phi}{2}+\cos\frac{\phi}{2}},\label{noncopl1}\\
\eta&>\frac{1}{\sin\frac{2\alpha}{2}+\frac{2\alpha}{2}},\label{noncopl2}
\end{align} 
to exclude other possible pairwise compatibility relations in $\{E_1,E_2,E_3,E_4\}$. Because $\phi<2\alpha<\frac{\pi}{2}$, and from the property that the function  
\begin{equation}
f(x)=\frac{1}{\cos\frac{x}{2}+\sin\frac{x}{2}}.
\end{equation}
is decreasing for $x\in(0,\pi/2]$ (see the graph of this function in Fig~\ref{graphf}), Eq.~\eqref{noncopl2} is automatically satisfied when Eq.~\eqref{noncopl1} is satisfied. Also from this decreasing property and the fact that $\phi>\alpha$ we find that Eqs.~\eqref{noncopl0} and \eqref{noncopl1} leave an open gap for the purity, namely,
\begin{equation}
\eta\in\left(\frac{1}{\cos\frac{\phi}{2}+\sin\frac{\phi}{2}},\frac{1}{\cos\frac{\alpha}{2}+\sin\frac{\alpha}{2}}\right],
\label{j90}
\end{equation}
such that the joint measurability structure from Fig.~\ref{count3} is realized. 
\end{example}

\section{Discussion and Outlook}\label{sec6}
In the preceding sections, we have done a fairly exhaustive study of joint measurability structures realizable on a qubit. One of our key motivations was to answer the following question: Is it possible to realize every conceivable joint measurability structure with qubit POVMs? While we haven't settled this question, the realizability of countably infinite sets of joint measurability structures (e.g., $N$-Specker and $N$-cycle scenarios for all $N$) makes us conjecture that this is the case:

\begin{conjecture}\label{conj4}
	All joint measurability structures are realizable with qubit POVMs.
\end{conjecture}

 As we have already noted, our construction of $N$-Specker scenarios for all $N$ on a qubit also renders the recipe for constructing arbitrary joint measurability structures in Ref.~\cite{KHF14} maximally efficient in terms of the dimensions required.

Critical to the constructions in this paper is the method of marginal surgery that lets us use a joint POVM for some given set of compatible POVMs to then construct non-trivial joint measurability structures. We did this by focussing, in particular, on the family of planar symmetric POVMs that were previously studied in Ref.~\cite{ULMH16}. We suspect that clever use of this method, coupled with some new results on joint measurability of interesting families of POVMs, could potentially be the key to demonstrating the realizability of arbitrary joint measurability structures with qubit POVMs. Marginal surgery may also be useful in problems where one needs to introduce incompatibility between specific subsets of POVMs, given a compatible set of POVMs.

Before we speculate on possible applications and directions for future work, let us recall the key results of this paper: 

\begin{enumerate}
	
	\item Introduction of the method of marginal surgery (Section \ref{sec3}). 
	\item Qubit realizability of joint measurability structures using marginal surgery: $N$-cycle and $N$-Specker scenarios for all $N$ (Section \ref{subsec3_1}, Proposition  \ref{ncycletheorem} and Section \ref{subsec3_2}, Corollary \ref{N-Specker's}), $N$-complete scenarios up to $N=5$ and miscellaneous other joint measurability structures (Section \ref{subsec5_1}).
	\item A sufficent condition for the joint measurability of coplanar, unbiased, binary qubit POVMs with the same purity (Section \ref{sec4}, Theorem \ref{counbiqu}).
	\item A sufficient condition for the joint measurability of any set of binary qubit POVMs (Section \ref{sec4}, Theorem \ref{suffNnsp}).
	\item Qubit realizability of arbitrary joint measurability structures comprising $4$ vertices (Section \ref{subsec5_2}).
\end{enumerate}

The open questions raised by this paper and directions for future work include:
\begin{enumerate}
	\item Conjectures \ref{conj1}, \ref{conj2}, and \ref{conj3}, where we have that the truth of Conjecture \ref{conj3} implies that Conjecture \ref{conj2} holds, which in turn implies that Conjecture \ref{conj1} holds. The importance of these conjectures lies in the following observations: if Conjecture \ref{conj1} holds, then we can determine the exact joint measurability structure exhibited by planar symmetric POVMs for every value of $\eta$; if Conjecture \ref{conj2} holds, then we can determine the exact joint measurability structure exhibited by coplanar and unbiased binary qubit POVMs with the same purity; lastly, if Conjecture~\ref{conj3} holds, then we can determine the exact joint measurability structure for arbitrary sets of coplanar and unbiased binary qubit POVMs.
	
	\item Conjecture \ref{conj4}, namely, the realizability of arbitrary joint measurability structures with qubit POVMs. Potential counter-examples to this conjecture could include, for example, a joint measurability structure with $N$ vertices for some $N>4$ such that one vertex is compatible with each of the remaining $N-1$ vertices but all the other pairs of vertices are incompatible. This would be a generalization of the joint measurability structure in Fig.~\ref{count3}. Another candidate counter-example is an $N$-complete joint measurability structure for some $N>5$.	
	
	\item The study of noise-robust contextuality for joint measurability structures not covered by previous work and implications of this for the project of characterizing the set of quantum correlations from principles \cite{KG14,Kunjwal14,Kunjwal16, GKS18}. In particular, given a joint measurability structure that is only realizable with non-projective POVMs in quantum theory, is there a meaningful distinction to be made between quantum and post-quantum theories? How should this distinction be made precise? An early hint that this is indeed the case was provided by the work of Ref.~\cite{KG14} and followed up in Ref.~\cite{Kunjwal16}.
\end{enumerate}

More broadly, the study of incompatibility as a resource for quantum information protocols has been a very active area of research in recent years \cite{Pusey15,SSC19, BCZ19}. Investigating the role that realizability of particular joint measurability structures can play in quantum information could lead to insights on previously under-appreciated aspects of incompatibility as well as their relevance for the foundations and applications of quantum theory. The satisfaction of a particular joint measurability structure by a set of POVMs is usually the first pre-requisite that determines their potential relevance, for example, in generating correlations that can demonstrate nonclassicality, whether in the case of Bell inequality violations \cite{WPF09, BV18} or of contextuality \cite{LSW11, KG14, ZCL17}. We hope that the results of this paper and the questions they motivate will go a long way towards elucidating the role of joint measurability structures in quantum theory and beyond.

\ack
We thank Tom\'a\v s Gonda who suggested that we find the symmetry point group of planar symmetric POVMs. This research was supported in part by the Perimeter Institute for Theoretical Physics, where N.A. was an undergraduate summer research student when most of this work was carried out. Research at Perimeter Institute is supported by the Government of Canada through the Department of Innovation, Science and Economic Development Canada and by the Province of Ontario through the Ministry of Research, Innovation and Science. R.K. is supported by the Program of Concerted Research Actions (ARC) of the Universit\'e libre de Bruxelles.

\bibliographystyle{apsrev4-1}
\nocite{apsrev41Control}
\bibliography{masterbibfile}

\begin{thebibliography}{51}%
\makeatletter
\providecommand \@ifxundefined [1]{%
 \@ifx{#1\undefined}
}%
\providecommand \@ifnum [1]{%
 \ifnum #1\expandafter \@firstoftwo
 \else \expandafter \@secondoftwo
 \fi
}%
\providecommand \@ifx [1]{%
 \ifx #1\expandafter \@firstoftwo
 \else \expandafter \@secondoftwo
 \fi
}%
\providecommand \natexlab [1]{#1}%
\providecommand \enquote  [1]{``#1''}%
\providecommand \bibnamefont  [1]{#1}%
\providecommand \bibfnamefont [1]{#1}%
\providecommand \citenamefont [1]{#1}%
\providecommand \href@noop [0]{\@secondoftwo}%
\providecommand \href [0]{\begingroup \@sanitize@url \@href}%
\providecommand \@href[1]{\@@startlink{#1}\@@href}%
\providecommand \@@href[1]{\endgroup#1\@@endlink}%
\providecommand \@sanitize@url [0]{\catcode `\\12\catcode `\$12\catcode
  `\&12\catcode `\#12\catcode `\^12\catcode `\_12\catcode `\%12\relax}%
\providecommand \@@startlink[1]{}%
\providecommand \@@endlink[0]{}%
\providecommand \url  [0]{\begingroup\@sanitize@url \@url }%
\providecommand \@url [1]{\endgroup\@href {#1}{\urlprefix }}%
\providecommand \urlprefix  [0]{URL }%
\providecommand \Eprint [0]{\href }%
\providecommand \doibase [0]{http://dx.doi.org/}%
\providecommand \selectlanguage [0]{\@gobble}%
\providecommand \bibinfo  [0]{\@secondoftwo}%
\providecommand \bibfield  [0]{\@secondoftwo}%
\providecommand \translation [1]{[#1]}%
\providecommand \BibitemOpen [0]{}%
\providecommand \bibitemStop [0]{}%
\providecommand \bibitemNoStop [0]{.\EOS\space}%
\providecommand \EOS [0]{\spacefactor3000\relax}%
\providecommand \BibitemShut  [1]{\csname bibitem#1\endcsname}%
\let\auto@bib@innerbib\@empty
\bibitem [{\citenamefont {Heinosaari}\ \emph {et~al.}(2008)\citenamefont
  {Heinosaari}, \citenamefont {Reitzner},\ and\ \citenamefont {Stano}}]{HRS08}%
  \BibitemOpen
  \bibfield  {author} {\bibinfo {author} {\bibfnamefont {T.}~\bibnamefont
  {Heinosaari}}, \bibinfo {author} {\bibfnamefont {D.}~\bibnamefont
  {Reitzner}}, \ and\ \bibinfo {author} {\bibfnamefont {P.}~\bibnamefont
  {Stano}},\ }\bibfield  {title} {\enquote {\bibinfo {title} {{Notes on Joint
  Measurability of Quantum Observables}},}\ }\href
  {https://doi.org/10.1007/s10701-008-9256-7} {\bibfield  {journal} {\bibinfo
  {journal} {Foundations of Physics}\ }\textbf {\bibinfo {volume} {38}},\
  \bibinfo {pages} {1133} (\bibinfo {year} {2008})}\BibitemShut {NoStop}%
\bibitem [{\citenamefont {Bell}(1964)}]{Bell64}%
  \BibitemOpen
  \bibfield  {author} {\bibinfo {author} {\bibfnamefont {J.~S.}\ \bibnamefont
  {Bell}},\ }\bibfield  {title} {\enquote {\bibinfo {title} {{On the
  Einstein-Podolsky-Rosen paradox}},}\ }\href
  {http://cds.cern.ch/record/111654} {\bibfield  {journal} {\bibinfo  {journal}
  {Physics}\ }\textbf {\bibinfo {volume} {1}},\ \bibinfo {pages} {195}
  (\bibinfo {year} {1964})}\BibitemShut {NoStop}%
\bibitem [{\citenamefont {Bell}(1966)}]{Bell66}%
  \BibitemOpen
  \bibfield  {author} {\bibinfo {author} {\bibfnamefont {J.~S.}\ \bibnamefont
  {Bell}},\ }\bibfield  {title} {\enquote {\bibinfo {title} {{On the Problem of
  Hidden Variables in Quantum Mechanics}},}\ }\href
  {https://link.aps.org/doi/10.1103/RevModPhys.38.447} {\bibfield  {journal}
  {\bibinfo  {journal} {Rev. Mod. Phys.}\ }\textbf {\bibinfo {volume} {38}},\
  \bibinfo {pages} {447} (\bibinfo {year} {1966})}\BibitemShut {NoStop}%
\bibitem [{\citenamefont {Brunner}\ \emph {et~al.}(2014)\citenamefont
  {Brunner}, \citenamefont {Cavalcanti}, \citenamefont {Pironio}, \citenamefont
  {Scarani},\ and\ \citenamefont {Wehner}}]{BCP14}%
  \BibitemOpen
  \bibfield  {author} {\bibinfo {author} {\bibfnamefont {N.}~\bibnamefont
  {Brunner}}, \bibinfo {author} {\bibfnamefont {D.}~\bibnamefont {Cavalcanti}},
  \bibinfo {author} {\bibfnamefont {S.}~\bibnamefont {Pironio}}, \bibinfo
  {author} {\bibfnamefont {V.}~\bibnamefont {Scarani}}, \ and\ \bibinfo
  {author} {\bibfnamefont {S.}~\bibnamefont {Wehner}},\ }\bibfield  {title}
  {\enquote {\bibinfo {title} {Bell nonlocality},}\ }\href
  {https://link.aps.org/doi/10.1103/RevModPhys.86.419} {\bibfield  {journal}
  {\bibinfo  {journal} {Rev. Mod. Phys.}\ }\textbf {\bibinfo {volume} {86}},\
  \bibinfo {pages} {419} (\bibinfo {year} {2014})}\BibitemShut {NoStop}%
\bibitem [{\citenamefont {Kochen}\ and\ \citenamefont {Specker}(1967)}]{KS67}%
  \BibitemOpen
  \bibfield  {author} {\bibinfo {author} {\bibfnamefont {S.}~\bibnamefont
  {Kochen}}\ and\ \bibinfo {author} {\bibfnamefont {E.~P.}\ \bibnamefont
  {Specker}},\ }\bibfield  {title} {\enquote {\bibinfo {title} {{The Problem of
  Hidden Variables in Quantum Mechanics}},}\ }\href
  {http://www.jstor.org/stable/24902153} {\bibfield  {journal} {\bibinfo
  {journal} {Journal of Mathematics and Mechanics}\ }\textbf {\bibinfo {volume}
  {17}},\ \bibinfo {pages} {59} (\bibinfo {year} {1967})}\BibitemShut {NoStop}%
\bibitem [{\citenamefont {Wolf}\ \emph {et~al.}(2009)\citenamefont {Wolf},
  \citenamefont {Perez-Garcia},\ and\ \citenamefont {Fernandez}}]{WPF09}%
  \BibitemOpen
  \bibfield  {author} {\bibinfo {author} {\bibfnamefont {M.~M.}\ \bibnamefont
  {Wolf}}, \bibinfo {author} {\bibfnamefont {D.}~\bibnamefont {Perez-Garcia}},
  \ and\ \bibinfo {author} {\bibfnamefont {C.}~\bibnamefont {Fernandez}},\
  }\bibfield  {title} {\enquote {\bibinfo {title} {{Measurements Incompatible
  in Quantum Theory Cannot Be Measured Jointly in Any Other No-Signaling
  Theory}},}\ }\href {\doibase 10.1103/PhysRevLett.103.230402} {\bibfield
  {journal} {\bibinfo  {journal} {Phys. Rev. Lett.}\ }\textbf {\bibinfo
  {volume} {103}},\ \bibinfo {pages} {230402} (\bibinfo {year}
  {2009})}\BibitemShut {NoStop}%
\bibitem [{\citenamefont {Quintino}\ \emph {et~al.}(2014)\citenamefont
  {Quintino}, \citenamefont {V\'ertesi},\ and\ \citenamefont
  {Brunner}}]{QVB14}%
  \BibitemOpen
  \bibfield  {author} {\bibinfo {author} {\bibfnamefont {M.~T.}\ \bibnamefont
  {Quintino}}, \bibinfo {author} {\bibfnamefont {T.}~\bibnamefont {V\'ertesi}},
  \ and\ \bibinfo {author} {\bibfnamefont {N.}~\bibnamefont {Brunner}},\
  }\bibfield  {title} {\enquote {\bibinfo {title} {{Joint Measurability,
  Einstein-Podolsky-Rosen Steering, and Bell Nonlocality}},}\ }\href {\doibase
  10.1103/PhysRevLett.113.160402} {\bibfield  {journal} {\bibinfo  {journal}
  {Phys. Rev. Lett.}\ }\textbf {\bibinfo {volume} {113}},\ \bibinfo {pages}
  {160402} (\bibinfo {year} {2014})}\BibitemShut {NoStop}%
\bibitem [{\citenamefont {Uola}\ \emph {et~al.}(2014)\citenamefont {Uola},
  \citenamefont {Moroder},\ and\ \citenamefont {G\"uhne}}]{UMG14}%
  \BibitemOpen
  \bibfield  {author} {\bibinfo {author} {\bibfnamefont {R.}~\bibnamefont
  {Uola}}, \bibinfo {author} {\bibfnamefont {T.}~\bibnamefont {Moroder}}, \
  and\ \bibinfo {author} {\bibfnamefont {O.}~\bibnamefont {G\"uhne}},\
  }\bibfield  {title} {\enquote {\bibinfo {title} {{Joint Measurability of
  Generalized Measurements Implies Classicality}},}\ }\href {\doibase
  10.1103/PhysRevLett.113.160403} {\bibfield  {journal} {\bibinfo  {journal}
  {Phys. Rev. Lett.}\ }\textbf {\bibinfo {volume} {113}},\ \bibinfo {pages}
  {160403} (\bibinfo {year} {2014})}\BibitemShut {NoStop}%
\bibitem [{\citenamefont {Kunjwal}\ \emph {et~al.}(2014)\citenamefont
  {Kunjwal}, \citenamefont {Heunen},\ and\ \citenamefont {Fritz}}]{KHF14}%
  \BibitemOpen
  \bibfield  {author} {\bibinfo {author} {\bibfnamefont {R.}~\bibnamefont
  {Kunjwal}}, \bibinfo {author} {\bibfnamefont {C.}~\bibnamefont {Heunen}}, \
  and\ \bibinfo {author} {\bibfnamefont {T.}~\bibnamefont {Fritz}},\ }\bibfield
   {title} {\enquote {\bibinfo {title} {Quantum realization of arbitrary joint
  measurability structures},}\ }\href {\doibase 10.1103/PhysRevA.89.052126}
  {\bibfield  {journal} {\bibinfo  {journal} {Phys. Rev. A}\ }\textbf {\bibinfo
  {volume} {89}},\ \bibinfo {pages} {052126} (\bibinfo {year}
  {2014})}\BibitemShut {NoStop}%
\bibitem [{\citenamefont {Heinosaari}\ \emph {et~al.}(2016)\citenamefont
  {Heinosaari}, \citenamefont {Miyadera},\ and\ \citenamefont {Ziman}}]{HMZ16}%
  \BibitemOpen
  \bibfield  {author} {\bibinfo {author} {\bibfnamefont {T.}~\bibnamefont
  {Heinosaari}}, \bibinfo {author} {\bibfnamefont {T.}~\bibnamefont
  {Miyadera}}, \ and\ \bibinfo {author} {\bibfnamefont {M.}~\bibnamefont
  {Ziman}},\ }\bibfield  {title} {\enquote {\bibinfo {title} {An invitation to
  quantum incompatibility},}\ }\href {\doibase 10.1088/1751-8113/49/12/123001}
  {\bibfield  {journal} {\bibinfo  {journal} {Journal of Physics A:
  Mathematical and Theoretical}\ }\textbf {\bibinfo {volume} {49}},\ \bibinfo
  {pages} {123001} (\bibinfo {year} {2016})}\BibitemShut {NoStop}%
\bibitem [{\citenamefont {Uola}\ \emph {et~al.}(2016)\citenamefont {Uola},
  \citenamefont {Luoma}, \citenamefont {Moroder},\ and\ \citenamefont
  {Heinosaari}}]{ULMH16}%
  \BibitemOpen
  \bibfield  {author} {\bibinfo {author} {\bibfnamefont {R.}~\bibnamefont
  {Uola}}, \bibinfo {author} {\bibfnamefont {K.}~\bibnamefont {Luoma}},
  \bibinfo {author} {\bibfnamefont {T.}~\bibnamefont {Moroder}}, \ and\
  \bibinfo {author} {\bibfnamefont {T.}~\bibnamefont {Heinosaari}},\ }\bibfield
   {title} {\enquote {\bibinfo {title} {Adaptive strategy for joint
  measurements},}\ }\href {\doibase 10.1103/PhysRevA.94.022109} {\bibfield
  {journal} {\bibinfo  {journal} {Phys. Rev. A}\ }\textbf {\bibinfo {volume}
  {94}},\ \bibinfo {pages} {022109} (\bibinfo {year} {2016})}\BibitemShut
  {NoStop}%
\bibitem [{\citenamefont {Cavalcanti}\ and\ \citenamefont
  {Skrzypczyk}(2016)}]{CS16}%
  \BibitemOpen
  \bibfield  {author} {\bibinfo {author} {\bibfnamefont {D.}~\bibnamefont
  {Cavalcanti}}\ and\ \bibinfo {author} {\bibfnamefont {P.}~\bibnamefont
  {Skrzypczyk}},\ }\bibfield  {title} {\enquote {\bibinfo {title} {Quantitative
  relations between measurement incompatibility, quantum steering, and
  nonlocality},}\ }\href {\doibase 10.1103/PhysRevA.93.052112} {\bibfield
  {journal} {\bibinfo  {journal} {Phys. Rev. A}\ }\textbf {\bibinfo {volume}
  {93}},\ \bibinfo {pages} {052112} (\bibinfo {year} {2016})}\BibitemShut
  {NoStop}%
\bibitem [{\citenamefont {Carmeli}\ \emph {et~al.}(2019)\citenamefont
  {Carmeli}, \citenamefont {Heinosaari},\ and\ \citenamefont {Toigo}}]{CHT19}%
  \BibitemOpen
  \bibfield  {author} {\bibinfo {author} {\bibfnamefont {C.}~\bibnamefont
  {Carmeli}}, \bibinfo {author} {\bibfnamefont {T.}~\bibnamefont {Heinosaari}},
  \ and\ \bibinfo {author} {\bibfnamefont {A.}~\bibnamefont {Toigo}},\
  }\bibfield  {title} {\enquote {\bibinfo {title} {{Quantum Incompatibility
  Witnesses}},}\ }\href {\doibase 10.1103/PhysRevLett.122.130402} {\bibfield
  {journal} {\bibinfo  {journal} {Phys. Rev. Lett.}\ }\textbf {\bibinfo
  {volume} {122}},\ \bibinfo {pages} {130402} (\bibinfo {year}
  {2019})}\BibitemShut {NoStop}%
\bibitem [{\citenamefont {Skrzypczyk}\ \emph {et~al.}(2019)\citenamefont
  {Skrzypczyk}, \citenamefont {\ifmmode \check{S}\else
  \v{S}\fi{}upi\ifmmode~\acute{c}\else \'{c}\fi{}},\ and\ \citenamefont
  {Cavalcanti}}]{SSC19}%
  \BibitemOpen
  \bibfield  {author} {\bibinfo {author} {\bibfnamefont {P.}~\bibnamefont
  {Skrzypczyk}}, \bibinfo {author} {\bibfnamefont {I.}~\bibnamefont {\ifmmode
  \check{S}\else \v{S}\fi{}upi\ifmmode~\acute{c}\else \'{c}\fi{}}}, \ and\
  \bibinfo {author} {\bibfnamefont {D.}~\bibnamefont {Cavalcanti}},\ }\bibfield
   {title} {\enquote {\bibinfo {title} {{All Sets of Incompatible Measurements
  give an Advantage in Quantum State Discrimination}},}\ }\href {\doibase
  10.1103/PhysRevLett.122.130403} {\bibfield  {journal} {\bibinfo  {journal}
  {Phys. Rev. Lett.}\ }\textbf {\bibinfo {volume} {122}},\ \bibinfo {pages}
  {130403} (\bibinfo {year} {2019})}\BibitemShut {NoStop}%
\bibitem [{\citenamefont {Uola}\ \emph {et~al.}(2019)\citenamefont {Uola},
  \citenamefont {Kraft}, \citenamefont {Shang}, \citenamefont {Yu},\ and\
  \citenamefont {G\"uhne}}]{UKS19}%
  \BibitemOpen
  \bibfield  {author} {\bibinfo {author} {\bibfnamefont {R.}~\bibnamefont
  {Uola}}, \bibinfo {author} {\bibfnamefont {T.}~\bibnamefont {Kraft}},
  \bibinfo {author} {\bibfnamefont {J.}~\bibnamefont {Shang}}, \bibinfo
  {author} {\bibfnamefont {X.-D.}\ \bibnamefont {Yu}}, \ and\ \bibinfo {author}
  {\bibfnamefont {O.}~\bibnamefont {G\"uhne}},\ }\bibfield  {title} {\enquote
  {\bibinfo {title} {{Quantifying Quantum Resources with Conic Programming}},}\
  }\href {\doibase 10.1103/PhysRevLett.122.130404} {\bibfield  {journal}
  {\bibinfo  {journal} {Phys. Rev. Lett.}\ }\textbf {\bibinfo {volume} {122}},\
  \bibinfo {pages} {130404} (\bibinfo {year} {2019})}\BibitemShut {NoStop}%
\bibitem [{\citenamefont {Oszmaniec}\ and\ \citenamefont
  {Biswas}(2019)}]{OB19}%
  \BibitemOpen
  \bibfield  {author} {\bibinfo {author} {\bibfnamefont {M.}~\bibnamefont
  {Oszmaniec}}\ and\ \bibinfo {author} {\bibfnamefont {T.}~\bibnamefont
  {Biswas}},\ }\bibfield  {title} {\enquote {\bibinfo {title} {Operational
  relevance of resource theories of quantum measurements},}\ }\href {\doibase
  10.22331/q-2019-04-26-133} {\bibfield  {journal} {\bibinfo  {journal}
  {{Quantum}}\ }\textbf {\bibinfo {volume} {3}},\ \bibinfo {pages} {133}
  (\bibinfo {year} {2019})}\BibitemShut {NoStop}%
\bibitem [{\citenamefont {Heunen}\ \emph {et~al.}(2014)\citenamefont {Heunen},
  \citenamefont {Fritz},\ and\ \citenamefont {Reyes}}]{HFR14}%
  \BibitemOpen
  \bibfield  {author} {\bibinfo {author} {\bibfnamefont {C.}~\bibnamefont
  {Heunen}}, \bibinfo {author} {\bibfnamefont {T.}~\bibnamefont {Fritz}}, \
  and\ \bibinfo {author} {\bibfnamefont {M.~L.}\ \bibnamefont {Reyes}},\
  }\bibfield  {title} {\enquote {\bibinfo {title} {Quantum theory realizes all
  joint measurability graphs},}\ }\href {\doibase 10.1103/PhysRevA.89.032121}
  {\bibfield  {journal} {\bibinfo  {journal} {Phys. Rev. A}\ }\textbf {\bibinfo
  {volume} {89}},\ \bibinfo {pages} {032121} (\bibinfo {year}
  {2014})}\BibitemShut {NoStop}%
\bibitem [{\citenamefont {Spekkens}(2014)}]{Spekkens14}%
  \BibitemOpen
  \bibfield  {author} {\bibinfo {author} {\bibfnamefont {R.~W.}\ \bibnamefont
  {Spekkens}},\ }\bibfield  {title} {\enquote {\bibinfo {title} {{The Status of
  Determinism in Proofs of the Impossibility of a Noncontextual Model of
  Quantum Theory}},}\ }\href {https://doi.org/10.1007/s10701-014-9833-x}
  {\bibfield  {journal} {\bibinfo  {journal} {Foundations of Physics}\ }\textbf
  {\bibinfo {volume} {44}},\ \bibinfo {pages} {1125} (\bibinfo {year}
  {2014})}\BibitemShut {NoStop}%
\bibitem [{\citenamefont {Spekkens}(2005)}]{Spekkens05}%
  \BibitemOpen
  \bibfield  {author} {\bibinfo {author} {\bibfnamefont {R.~W.}\ \bibnamefont
  {Spekkens}},\ }\bibfield  {title} {\enquote {\bibinfo {title} {Contextuality
  for preparations, transformations, and unsharp measurements},}\ }\href
  {\doibase 10.1103/PhysRevA.71.052108} {\bibfield  {journal} {\bibinfo
  {journal} {Phys. Rev. A}\ }\textbf {\bibinfo {volume} {71}},\ \bibinfo
  {pages} {052108} (\bibinfo {year} {2005})}\BibitemShut {NoStop}%
\bibitem [{\citenamefont {Mazurek}\ \emph {et~al.}(2016)\citenamefont
  {Mazurek}, \citenamefont {Pusey}, \citenamefont {Kunjwal}, \citenamefont
  {Resch},\ and\ \citenamefont {Spekkens}}]{MPK16}%
  \BibitemOpen
  \bibfield  {author} {\bibinfo {author} {\bibfnamefont {M.~D.}\ \bibnamefont
  {Mazurek}}, \bibinfo {author} {\bibfnamefont {M.~F.}\ \bibnamefont {Pusey}},
  \bibinfo {author} {\bibfnamefont {R.}~\bibnamefont {Kunjwal}}, \bibinfo
  {author} {\bibfnamefont {K.~J.}\ \bibnamefont {Resch}}, \ and\ \bibinfo
  {author} {\bibfnamefont {R.~W.}\ \bibnamefont {Spekkens}},\ }\bibfield
  {title} {\enquote {\bibinfo {title} {{An experimental test of
  noncontextuality without unphysical idealizations}},}\ }\href
  {https://doi.org/10.1038/ncomms11780} {\bibfield  {journal} {\bibinfo
  {journal} {Nature Communications}\ }\textbf {\bibinfo {volume} {7}},\
  \bibinfo {pages} {1} (\bibinfo {year} {2016})}\BibitemShut {NoStop}%
\bibitem [{\citenamefont {Kunjwal}\ and\ \citenamefont
  {Spekkens}(2015)}]{KS15}%
  \BibitemOpen
  \bibfield  {author} {\bibinfo {author} {\bibfnamefont {R.}~\bibnamefont
  {Kunjwal}}\ and\ \bibinfo {author} {\bibfnamefont {R.~W.}\ \bibnamefont
  {Spekkens}},\ }\bibfield  {title} {\enquote {\bibinfo {title} {{From the
  Kochen-Specker Theorem to Noncontextuality Inequalities without Assuming
  Determinism}},}\ }\href {\doibase 10.1103/PhysRevLett.115.110403} {\bibfield
  {journal} {\bibinfo  {journal} {Phys. Rev. Lett.}\ }\textbf {\bibinfo
  {volume} {115}},\ \bibinfo {pages} {110403} (\bibinfo {year}
  {2015})}\BibitemShut {NoStop}%
\bibitem [{\citenamefont {Kunjwal}\ and\ \citenamefont
  {Spekkens}(2018)}]{KS18}%
  \BibitemOpen
  \bibfield  {author} {\bibinfo {author} {\bibfnamefont {R.}~\bibnamefont
  {Kunjwal}}\ and\ \bibinfo {author} {\bibfnamefont {R.~W.}\ \bibnamefont
  {Spekkens}},\ }\bibfield  {title} {\enquote {\bibinfo {title} {From
  statistical proofs of the kochen-specker theorem to noise-robust
  noncontextuality inequalities},}\ }\href {\doibase
  10.1103/PhysRevA.97.052110} {\bibfield  {journal} {\bibinfo  {journal} {Phys.
  Rev. A}\ }\textbf {\bibinfo {volume} {97}},\ \bibinfo {pages} {052110}
  (\bibinfo {year} {2018})}\BibitemShut {NoStop}%
\bibitem [{\citenamefont {Kunjwal}(2019)}]{Kunjwal19}%
  \BibitemOpen
  \bibfield  {author} {\bibinfo {author} {\bibfnamefont {R.}~\bibnamefont
  {Kunjwal}},\ }\bibfield  {title} {\enquote {\bibinfo {title} {Beyond the
  {C}abello-{S}everini-{W}inter framework: {M}aking sense of contextuality
  without sharpness of measurements},}\ }\href {\doibase
  10.22331/q-2019-09-09-184} {\bibfield  {journal} {\bibinfo  {journal}
  {{Quantum}}\ }\textbf {\bibinfo {volume} {3}},\ \bibinfo {pages} {184}
  (\bibinfo {year} {2019})}\BibitemShut {NoStop}%
\bibitem [{\citenamefont {Kunjwal}(2020)}]{Kunjwal20}%
  \BibitemOpen
  \bibfield  {author} {\bibinfo {author} {\bibfnamefont {R.}~\bibnamefont
  {Kunjwal}},\ }\bibfield  {title} {\enquote {\bibinfo {title} {Hypergraph
  framework for irreducible noncontextuality inequalities from logical proofs
  of the {K}ochen-{S}pecker theorem},}\ }\href {\doibase
  10.22331/q-2020-01-10-219} {\bibfield  {journal} {\bibinfo  {journal}
  {{Quantum}}\ }\textbf {\bibinfo {volume} {4}},\ \bibinfo {pages} {219}
  (\bibinfo {year} {2020})}\BibitemShut {NoStop}%
\bibitem [{\citenamefont {Liang}\ \emph {et~al.}(2011)\citenamefont {Liang},
  \citenamefont {Spekkens},\ and\ \citenamefont {Wiseman}}]{LSW11}%
  \BibitemOpen
  \bibfield  {author} {\bibinfo {author} {\bibfnamefont {Y.-C.}\ \bibnamefont
  {Liang}}, \bibinfo {author} {\bibfnamefont {R.~W.}\ \bibnamefont {Spekkens}},
  \ and\ \bibinfo {author} {\bibfnamefont {H.~M.}\ \bibnamefont {Wiseman}},\
  }\bibfield  {title} {\enquote {\bibinfo {title} {Specker's parable of the
  overprotective seer: A road to contextuality, nonlocality and
  complementarity},}\ }\href {https://doi.org/10.1016/j.physrep.2011.05.001}
  {\bibfield  {journal} {\bibinfo  {journal} {Physics Reports}\ }\textbf
  {\bibinfo {volume} {506}},\ \bibinfo {pages} {1} (\bibinfo {year}
  {2011})}\BibitemShut {NoStop}%
\bibitem [{\citenamefont {Kunjwal}\ and\ \citenamefont {Ghosh}(2014)}]{KG14}%
  \BibitemOpen
  \bibfield  {author} {\bibinfo {author} {\bibfnamefont {R.}~\bibnamefont
  {Kunjwal}}\ and\ \bibinfo {author} {\bibfnamefont {S.}~\bibnamefont
  {Ghosh}},\ }\bibfield  {title} {\enquote {\bibinfo {title} {Minimal
  state-dependent proof of measurement contextuality for a qubit},}\ }\href
  {\doibase 10.1103/PhysRevA.89.042118} {\bibfield  {journal} {\bibinfo
  {journal} {Phys. Rev. A}\ }\textbf {\bibinfo {volume} {89}},\ \bibinfo
  {pages} {042118} (\bibinfo {year} {2014})}\BibitemShut {NoStop}%
\bibitem [{\citenamefont {Zhan}\ \emph {et~al.}(2017)\citenamefont {Zhan},
  \citenamefont {Cavalcanti}, \citenamefont {Li}, \citenamefont {Bian},
  \citenamefont {Zhang}, \citenamefont {Wiseman},\ and\ \citenamefont
  {Xue}}]{ZCL17}%
  \BibitemOpen
  \bibfield  {author} {\bibinfo {author} {\bibfnamefont {X.}~\bibnamefont
  {Zhan}}, \bibinfo {author} {\bibfnamefont {E.~G.}\ \bibnamefont
  {Cavalcanti}}, \bibinfo {author} {\bibfnamefont {J.}~\bibnamefont {Li}},
  \bibinfo {author} {\bibfnamefont {Z.}~\bibnamefont {Bian}}, \bibinfo {author}
  {\bibfnamefont {Y.}~\bibnamefont {Zhang}}, \bibinfo {author} {\bibfnamefont
  {H.~M.}\ \bibnamefont {Wiseman}}, \ and\ \bibinfo {author} {\bibfnamefont
  {P.}~\bibnamefont {Xue}},\ }\bibfield  {title} {\enquote {\bibinfo {title}
  {Experimental generalized contextuality with single-photon qubits},}\ }\href
  {https://doi.org/10.1364/OPTICA.4.000966} {\bibfield  {journal} {\bibinfo
  {journal} {Optica}\ }\textbf {\bibinfo {volume} {4}},\ \bibinfo {pages} {966}
  (\bibinfo {year} {2017})}\BibitemShut {NoStop}%
\bibitem [{\citenamefont {Kunjwal}(2014)}]{Kunjwal14}%
  \BibitemOpen
  \bibfield  {author} {\bibinfo {author} {\bibfnamefont {R.}~\bibnamefont
  {Kunjwal}},\ }\bibfield  {title} {\enquote {\bibinfo {title}
  {{Noncontextuality without determinism and admissible (in)compatibility
  relations: revisiting Specker's parable}},}\ }\href
  {http://pirsa.org/14010102/} {\bibfield  {journal} {\bibinfo  {journal}
  {PIRSA 14010102}\ } (\bibinfo {year} {2014})}\BibitemShut {NoStop}%
\bibitem [{\citenamefont {Kunjwal}(2016)}]{Kunjwal16}%
  \BibitemOpen
  \bibfield  {author} {\bibinfo {author} {\bibfnamefont {R.}~\bibnamefont
  {Kunjwal}},\ }\bibfield  {title} {\enquote {\bibinfo {title} {{Contextuality
  beyond the Kochen-Specker theorem}},}\ }\href
  {https://arxiv.org/abs/1612.07250} {\bibfield  {journal} {\bibinfo  {journal}
  {arXiv preprint arXiv:1612.07250}\ } (\bibinfo {year} {2016})}\BibitemShut
  {NoStop}%
\bibitem [{\citenamefont {Kunjwal}(2017)}]{Kunjwal17}%
  \BibitemOpen
  \bibfield  {author} {\bibinfo {author} {\bibfnamefont {R.}~\bibnamefont
  {Kunjwal}},\ }\bibfield  {title} {\enquote {\bibinfo {title} {{How to go from
  the KS theorem to experimentally testable noncontextuality inequalities}},}\
  }\href {http://pirsa.org/17070059} {\bibfield  {journal} {\bibinfo  {journal}
  {PIRSA 17070059}\ } (\bibinfo {year} {2017})}\BibitemShut {NoStop}%
\bibitem [{\citenamefont {Barrett}(2007)}]{Barrett07}%
  \BibitemOpen
  \bibfield  {author} {\bibinfo {author} {\bibfnamefont {J.}~\bibnamefont
  {Barrett}},\ }\bibfield  {title} {\enquote {\bibinfo {title} {Information
  processing in generalized probabilistic theories},}\ }\href {\doibase
  10.1103/PhysRevA.75.032304} {\bibfield  {journal} {\bibinfo  {journal} {Phys.
  Rev. A}\ }\textbf {\bibinfo {volume} {75}},\ \bibinfo {pages} {032304}
  (\bibinfo {year} {2007})}\BibitemShut {NoStop}%
\bibitem [{\citenamefont {Banik}\ \emph {et~al.}(2013)\citenamefont {Banik},
  \citenamefont {Gazi}, \citenamefont {Ghosh},\ and\ \citenamefont
  {Kar}}]{BGG13}%
  \BibitemOpen
  \bibfield  {author} {\bibinfo {author} {\bibfnamefont {M.}~\bibnamefont
  {Banik}}, \bibinfo {author} {\bibfnamefont {M.~R.}\ \bibnamefont {Gazi}},
  \bibinfo {author} {\bibfnamefont {S.}~\bibnamefont {Ghosh}}, \ and\ \bibinfo
  {author} {\bibfnamefont {G.}~\bibnamefont {Kar}},\ }\bibfield  {title}
  {\enquote {\bibinfo {title} {Degree of complementarity determines the
  nonlocality in quantum mechanics},}\ }\href {\doibase
  10.1103/PhysRevA.87.052125} {\bibfield  {journal} {\bibinfo  {journal} {Phys.
  Rev. A}\ }\textbf {\bibinfo {volume} {87}},\ \bibinfo {pages} {052125}
  (\bibinfo {year} {2013})}\BibitemShut {NoStop}%
\bibitem [{\citenamefont {Stevens}\ and\ \citenamefont {Busch}(2014)}]{SB14}%
  \BibitemOpen
  \bibfield  {author} {\bibinfo {author} {\bibfnamefont {N.}~\bibnamefont
  {Stevens}}\ and\ \bibinfo {author} {\bibfnamefont {P.}~\bibnamefont
  {Busch}},\ }\bibfield  {title} {\enquote {\bibinfo {title} {{Steering,
  incompatibility, and Bell-inequality violations in a class of probabilistic
  theories}},}\ }\href {\doibase 10.1103/PhysRevA.89.022123} {\bibfield
  {journal} {\bibinfo  {journal} {Phys. Rev. A}\ }\textbf {\bibinfo {volume}
  {89}},\ \bibinfo {pages} {022123} (\bibinfo {year} {2014})}\BibitemShut
  {NoStop}%
\bibitem [{\citenamefont {Filippov}\ \emph {et~al.}(2017)\citenamefont
  {Filippov}, \citenamefont {Heinosaari},\ and\ \citenamefont
  {Lepp\"aj\"arvi}}]{FHL17}%
  \BibitemOpen
  \bibfield  {author} {\bibinfo {author} {\bibfnamefont {S.~N.}\ \bibnamefont
  {Filippov}}, \bibinfo {author} {\bibfnamefont {T.}~\bibnamefont
  {Heinosaari}}, \ and\ \bibinfo {author} {\bibfnamefont {L.}~\bibnamefont
  {Lepp\"aj\"arvi}},\ }\bibfield  {title} {\enquote {\bibinfo {title}
  {Necessary condition for incompatibility of observables in general
  probabilistic theories},}\ }\href {\doibase 10.1103/PhysRevA.95.032127}
  {\bibfield  {journal} {\bibinfo  {journal} {Phys. Rev. A}\ }\textbf {\bibinfo
  {volume} {95}},\ \bibinfo {pages} {032127} (\bibinfo {year}
  {2017})}\BibitemShut {NoStop}%
\bibitem [{\citenamefont {Gonda}\ \emph {et~al.}(2018)\citenamefont {Gonda},
  \citenamefont {Kunjwal}, \citenamefont {Schmid}, \citenamefont {Wolfe},\ and\
  \citenamefont {Sainz}}]{GKS18}%
  \BibitemOpen
  \bibfield  {author} {\bibinfo {author} {\bibfnamefont {T.}~\bibnamefont
  {Gonda}}, \bibinfo {author} {\bibfnamefont {R.}~\bibnamefont {Kunjwal}},
  \bibinfo {author} {\bibfnamefont {D.}~\bibnamefont {Schmid}}, \bibinfo
  {author} {\bibfnamefont {E.}~\bibnamefont {Wolfe}}, \ and\ \bibinfo {author}
  {\bibfnamefont {A.~B.}\ \bibnamefont {Sainz}},\ }\bibfield  {title} {\enquote
  {\bibinfo {title} {Almost {Q}uantum {C}orrelations are {I}nconsistent with
  {S}pecker's {P}rinciple},}\ }\href {\doibase 10.22331/q-2018-08-27-87}
  {\bibfield  {journal} {\bibinfo  {journal} {{Quantum}}\ }\textbf {\bibinfo
  {volume} {2}},\ \bibinfo {pages} {87} (\bibinfo {year} {2018})}\BibitemShut
  {NoStop}%
\bibitem [{\citenamefont {Navascu{\'e}s}\ \emph {et~al.}(2015)\citenamefont
  {Navascu{\'e}s}, \citenamefont {Guryanova}, \citenamefont {Hoban},\ and\
  \citenamefont {Ac{\'\i}n}}]{NGH15}%
  \BibitemOpen
  \bibfield  {author} {\bibinfo {author} {\bibfnamefont {M.}~\bibnamefont
  {Navascu{\'e}s}}, \bibinfo {author} {\bibfnamefont {Y.}~\bibnamefont
  {Guryanova}}, \bibinfo {author} {\bibfnamefont {M.~J.}\ \bibnamefont
  {Hoban}}, \ and\ \bibinfo {author} {\bibfnamefont {A.}~\bibnamefont
  {Ac{\'\i}n}},\ }\bibfield  {title} {\enquote {\bibinfo {title} {Almost
  quantum correlations},}\ }\href {https://doi.org/10.1038/ncomms7288}
  {\bibfield  {journal} {\bibinfo  {journal} {Nature communications}\ }\textbf
  {\bibinfo {volume} {6}},\ \bibinfo {pages} {1} (\bibinfo {year}
  {2015})}\BibitemShut {NoStop}%
\bibitem [{\citenamefont {Ac{\'\i}n}\ \emph {et~al.}(2015)\citenamefont
  {Ac{\'\i}n}, \citenamefont {Fritz}, \citenamefont {Leverrier},\ and\
  \citenamefont {Sainz}}]{AFL15}%
  \BibitemOpen
  \bibfield  {author} {\bibinfo {author} {\bibfnamefont {A.}~\bibnamefont
  {Ac{\'\i}n}}, \bibinfo {author} {\bibfnamefont {T.}~\bibnamefont {Fritz}},
  \bibinfo {author} {\bibfnamefont {A.}~\bibnamefont {Leverrier}}, \ and\
  \bibinfo {author} {\bibfnamefont {A.~B.}\ \bibnamefont {Sainz}},\ }\bibfield
  {title} {\enquote {\bibinfo {title} {{A Combinatorial Approach to Nonlocality
  and Contextuality}},}\ }\href {https://doi.org/10.1007/s00220-014-2260-1}
  {\bibfield  {journal} {\bibinfo  {journal} {Communications in Mathematical
  Physics}\ }\textbf {\bibinfo {volume} {334}},\ \bibinfo {pages} {533}
  (\bibinfo {year} {2015})}\BibitemShut {NoStop}%
\bibitem [{\citenamefont {Cabello}(2012)}]{Cabello12}%
  \BibitemOpen
  \bibfield  {author} {\bibinfo {author} {\bibfnamefont {A.}~\bibnamefont
  {Cabello}},\ }\bibfield  {title} {\enquote {\bibinfo {title} {Specker's
  fundamental principle of quantum mechanics},}\ }\href
  {https://arxiv.org/abs/1212.1756} {\bibfield  {journal} {\bibinfo  {journal}
  {arXiv preprint arXiv:1212.1756}\ } (\bibinfo {year} {2012})}\BibitemShut
  {NoStop}%
\bibitem [{\citenamefont {Yu}\ \emph {et~al.}(2010)\citenamefont {Yu},
  \citenamefont {Liu}, \citenamefont {Li},\ and\ \citenamefont {Oh}}]{YLL10}%
  \BibitemOpen
  \bibfield  {author} {\bibinfo {author} {\bibfnamefont {S.}~\bibnamefont
  {Yu}}, \bibinfo {author} {\bibfnamefont {N.-l.}\ \bibnamefont {Liu}},
  \bibinfo {author} {\bibfnamefont {L.}~\bibnamefont {Li}}, \ and\ \bibinfo
  {author} {\bibfnamefont {C.~H.}\ \bibnamefont {Oh}},\ }\bibfield  {title}
  {\enquote {\bibinfo {title} {Joint measurement of two unsharp observables of
  a qubit},}\ }\href {\doibase 10.1103/PhysRevA.81.062116} {\bibfield
  {journal} {\bibinfo  {journal} {Phys. Rev. A}\ }\textbf {\bibinfo {volume}
  {81}},\ \bibinfo {pages} {062116} (\bibinfo {year} {2010})}\BibitemShut
  {NoStop}%
\bibitem [{\citenamefont {Busch}(1986)}]{Busch86}%
  \BibitemOpen
  \bibfield  {author} {\bibinfo {author} {\bibfnamefont {P.}~\bibnamefont
  {Busch}},\ }\bibfield  {title} {\enquote {\bibinfo {title} {Unsharp reality
  and joint measurements for spin observables},}\ }\href {\doibase
  10.1103/PhysRevD.33.2253} {\bibfield  {journal} {\bibinfo  {journal} {Phys.
  Rev. D}\ }\textbf {\bibinfo {volume} {33}},\ \bibinfo {pages} {2253}
  (\bibinfo {year} {1986})}\BibitemShut {NoStop}%
\bibitem [{\citenamefont {Pal}\ and\ \citenamefont {Ghosh}(2011)}]{PG11}%
  \BibitemOpen
  \bibfield  {author} {\bibinfo {author} {\bibfnamefont {R.}~\bibnamefont
  {Pal}}\ and\ \bibinfo {author} {\bibfnamefont {S.}~\bibnamefont {Ghosh}},\
  }\bibfield  {title} {\enquote {\bibinfo {title} {{Approximate joint
  measurement of qubit observables through an Arthur{\textendash}Kelly
  model}},}\ }\href {\doibase 10.1088/1751-8113/44/48/485303} {\bibfield
  {journal} {\bibinfo  {journal} {Journal of Physics A: Mathematical and
  Theoretical}\ }\textbf {\bibinfo {volume} {44}},\ \bibinfo {pages} {485303}
  (\bibinfo {year} {2011})}\BibitemShut {NoStop}%
\bibitem [{\citenamefont {Kupitz}\ and\ \citenamefont {Martini}(1994)}]{KM94}%
  \BibitemOpen
  \bibfield  {author} {\bibinfo {author} {\bibfnamefont {Y.~S.}\ \bibnamefont
  {Kupitz}}\ and\ \bibinfo {author} {\bibfnamefont {H.}~\bibnamefont
  {Martini}},\ }\bibfield  {title} {\enquote {\bibinfo {title} {{The
  Fermat-Torricelli point and isosceles tetrahedra}},}\ }\href
  {https://doi.org/10.1007/BF01228057} {\bibfield  {journal} {\bibinfo
  {journal} {Journal of Geometry}\ }\textbf {\bibinfo {volume} {49}},\ \bibinfo
  {pages} {150} (\bibinfo {year} {1994})}\BibitemShut {NoStop}%
\bibitem [{\citenamefont {Yu}\ and\ \citenamefont {Oh}(2013)}]{YO13}%
  \BibitemOpen
  \bibfield  {author} {\bibinfo {author} {\bibfnamefont {S.}~\bibnamefont
  {Yu}}\ and\ \bibinfo {author} {\bibfnamefont {C.}~\bibnamefont {Oh}},\
  }\bibfield  {title} {\enquote {\bibinfo {title} {Quantum contextuality and
  joint measurement of three observables of a qubit},}\ }\href
  {https://arxiv.org/abs/1312.6470} {\bibfield  {journal} {\bibinfo  {journal}
  {arXiv preprint arXiv:1312.6470}\ } (\bibinfo {year} {2013})}\BibitemShut
  {NoStop}%
\bibitem [{\citenamefont {Plastria}(2006)}]{Plastria06}%
  \BibitemOpen
  \bibfield  {author} {\bibinfo {author} {\bibfnamefont {F.}~\bibnamefont
  {Plastria}},\ }\bibfield  {title} {\enquote {\bibinfo {title} {{Four-point
  Fermat location problems revisited. New proofs and extensions of old
  results}},}\ }\href {\doibase 10.1093/imaman/dpl007} {\bibfield  {journal}
  {\bibinfo  {journal} {IMA J. Manag. Math.}\ }\textbf {\bibinfo {volume}
  {17}},\ \bibinfo {pages} {387} (\bibinfo {year} {2006})}\BibitemShut
  {NoStop}%
\bibitem [{\citenamefont {Baumslag}(1968)}]{Baumslag68}%
  \BibitemOpen
  \bibfield  {author} {\bibinfo {author} {\bibfnamefont {B.}~\bibnamefont
  {Baumslag}},\ }\href {https://cds.cern.ch/record/109046} {\emph {\bibinfo
  {title} {{Schaum's outline of theory and problems of group theory}}}},\
  Schaum's outline\ (\bibinfo  {publisher} {McGraw-Hill},\ \bibinfo {address}
  {New York, NY},\ \bibinfo {year} {1968})\BibitemShut {NoStop}%
\bibitem [{\citenamefont {Rotman}(2012)}]{Rotman12}%
  \BibitemOpen
  \bibfield  {author} {\bibinfo {author} {\bibfnamefont {J.~J.}\ \bibnamefont
  {Rotman}},\ }\href {https://doi.org/10.1007/978-1-4612-4176-8} {\emph
  {\bibinfo {title} {{An Introduction to the Theory of Groups}}}},\ Vol.\
  \bibinfo {volume} {148}\ (\bibinfo  {publisher} {Springer Science \& Business
  Media},\ \bibinfo {year} {2012})\BibitemShut {NoStop}%
\bibitem [{\citenamefont {Rose}(2009)}]{Rose09}%
  \BibitemOpen
  \bibfield  {author} {\bibinfo {author} {\bibfnamefont {H.~E.}\ \bibnamefont
  {Rose}},\ }\href {https://doi.org/10.1007/978-1-84882-889-6} {\emph {\bibinfo
  {title} {{A Course on Finite Groups}}}}\ (\bibinfo  {publisher} {Springer
  Science \& Business Media},\ \bibinfo {year} {2009})\BibitemShut {NoStop}%
\bibitem [{\citenamefont {Johnson}(2012)}]{Johnson12}%
  \BibitemOpen
  \bibfield  {author} {\bibinfo {author} {\bibfnamefont {D.~L.}\ \bibnamefont
  {Johnson}},\ }\href {https://doi.org/10.1007/978-1-4471-0243-4} {\emph
  {\bibinfo {title} {Symmetries}}}\ (\bibinfo  {publisher} {Springer Science \&
  Business Media},\ \bibinfo {year} {2012})\BibitemShut {NoStop}%
\bibitem [{\citenamefont {Pusey}(2015)}]{Pusey15}%
  \BibitemOpen
  \bibfield  {author} {\bibinfo {author} {\bibfnamefont {M.~F.}\ \bibnamefont
  {Pusey}},\ }\bibfield  {title} {\enquote {\bibinfo {title} {Verifying the
  quantumness of a channel with an untrusted device},}\ }\href {\doibase
  10.1364/JOSAB.32.000A56} {\bibfield  {journal} {\bibinfo  {journal} {J. Opt.
  Soc. Am. B}\ }\textbf {\bibinfo {volume} {32}},\ \bibinfo {pages} {A56}
  (\bibinfo {year} {2015})}\BibitemShut {NoStop}%
\bibitem [{\citenamefont {Buscemi}\ \emph {et~al.}(2020)\citenamefont
  {Buscemi}, \citenamefont {Chitambar},\ and\ \citenamefont {Zhou}}]{BCZ19}%
  \BibitemOpen
  \bibfield  {author} {\bibinfo {author} {\bibfnamefont {F.}~\bibnamefont
  {Buscemi}}, \bibinfo {author} {\bibfnamefont {E.}~\bibnamefont {Chitambar}},
  \ and\ \bibinfo {author} {\bibfnamefont {W.}~\bibnamefont {Zhou}},\
  }\bibfield  {title} {\enquote {\bibinfo {title} {{Complete Resource Theory of
  Quantum Incompatibility as Quantum Programmability}},}\ }\href {\doibase
  10.1103/PhysRevLett.124.120401} {\bibfield  {journal} {\bibinfo  {journal}
  {Phys. Rev. Lett.}\ }\textbf {\bibinfo {volume} {124}},\ \bibinfo {pages}
  {120401} (\bibinfo {year} {2020})}\BibitemShut {NoStop}%
\bibitem [{\citenamefont {Bene}\ and\ \citenamefont
  {V{\'{e}}rtesi}(2018)}]{BV18}%
  \BibitemOpen
  \bibfield  {author} {\bibinfo {author} {\bibfnamefont {E.}~\bibnamefont
  {Bene}}\ and\ \bibinfo {author} {\bibfnamefont {T.}~\bibnamefont
  {V{\'{e}}rtesi}},\ }\bibfield  {title} {\enquote {\bibinfo {title}
  {{Measurement incompatibility does not give rise to Bell violation in
  general}},}\ }\href {\doibase 10.1088/1367-2630/aa9ca3} {\bibfield  {journal}
  {\bibinfo  {journal} {New Journal of Physics}\ }\textbf {\bibinfo {volume}
  {20}},\ \bibinfo {pages} {013021} (\bibinfo {year} {2018})}\BibitemShut
  {NoStop}%
\end{thebibliography}%

\appendix
\section{Proofs of various identities}
\label{AppA}
\begin{lemma}\upshape
\label{id1}
Let $\vec{e}(\vec{x})$ be the Bloch vectors described in Theorem \ref{gluslov} and Remark \ref{rem1}. Then
\begin{equation}\label{id1eqn}
\sum_{\vec{x}\in\{\pm1\}^N}^{x_1,\ldots,x_k=+1}\vec{e}(\vec{x})=\begin{cases}\sum\limits_{i= p+\frac{3-N}{2}}^{\frac{N+1}{2}} \left(\cos\frac{(i-1)\pi}{N},\sin\frac{(i-1)\pi}{N},0\right)\\
\hspace{3.5cm} \text{for odd N},\\
\sum\limits_{i= p+1-\frac{N}{2}}^{\frac{N}{2}} \left(\cos\frac{(2i-1)\pi}{2N},\sin\frac{(2i-1)\pi}{2N},0\right)\\ \hspace{3.5cm}\text{for even N.}
\end{cases}
\end{equation}
\end{lemma}
\begin{proof}
We have to sum over all $\vec{e}(\vec{x})$ such that $\vec{x}$ ``starts" with $+1$ in the first $k$ slots. The boundaries for the summation index $i$ are determined from the conditions
\begin{subequations}
\begin{align}
\vec{n}_1\cdot\left(\cos\frac{(i-1)\pi}{N},\sin\frac{(i-1)\pi}{N},0\right)&>0,\hphantom{a}\text{and}\nonumber\\
\vec{n}_k\cdot\left(\cos\frac{(i-1)\pi}{N},\sin\frac{(i-1)\pi}{N},0\right)&>0,\hphantom{a}\text{for odd } N;\\
\vec{n}_1\cdot\left(\cos\frac{(2i-1)\pi}{2N},\sin\frac{(2i-1)\pi}{2N},0\right)&>0,\hphantom{a}\text{and}nonumber\\ 
\vec{n}_k\cdot\left(\cos\frac{(2i-1)\pi}{2N},\sin\frac{(2i-1)\pi}{2N},0\right)&>0,\hphantom{a}\text{for even } N.
\end{align}
\end{subequations}
Then for odd $N$:
\begin{equation}
\vec{n}_1\cdot\left(\cos\frac{(i-1)\pi}{N},\sin\frac{(i-1)\pi}{N},0\right)= \cos\frac{(i-1)\pi}{N}>0.
\end{equation} 
This inequality gives 
\begin{align}
\frac{(i-1)\pi}{N}\in\left(-\frac{\pi}{2},\frac{\pi}{2}\right)&\Longrightarrow 1-\frac{N}{2}<i<1+\frac{N}{2}\nonumber\\
&\Longrightarrow \frac{3-N}{2}\leq i\leq \frac{1+N}{2}.
\label{usl21}
\end{align}
Also, 
\begin{align}
&\vec{n}_k\cdot\left(\cos\frac{(i-1)\pi}{N},\sin\frac{(i-1)\pi}{N},0\right)\nonumber\\
&=\cos\frac{(k-1)\pi}{N}\cos\frac{(i-1)\pi}{N}+\sin\frac{(k-1)\pi}{N}\sin\frac{(i-1)\pi}{N}\nonumber\\
&=\cos\frac{(k-1)\pi-(i-1)\pi}{N} =\cos\frac{(k-i)\pi}{N}>0.
\end{align}
This inequality gives 
\begin{align}
\frac{(k-i)\pi}{N}\in\left(-\frac{\pi}{2},\frac{\pi}{2}\right)&\Longrightarrow k-\frac{N}{2}<i<k+\frac{N}{2}\nonumber\\
&\Longrightarrow k+\frac{1-N}{2}\leq i\leq k+\frac{N-1}{2}\nonumber\\&\Longrightarrow p+\frac{3-N}{2}\leq i \leq p+\frac{N+1}{2}.
\label{usl2k}
\end{align}
Taking the intersection of Eqs.~\eqref{usl21} and \eqref{usl2k} we get 
\begin{equation}
p+\frac{3-N}{2}\leq i \leq \frac{N+1}{2}.
\label{n}
\end{equation}
Now for even $N$:
\begin{equation}
\vec{n}_1\cdot\left(\cos\frac{(2i-1)\pi}{2N},\sin\frac{(2i-1)\pi}{2N},0\right)= \cos\frac{(2i-1)\pi}{2N}>0,
\end{equation}
which in similar succession of steps as in \eqref{usl21} gives 
\begin{equation}
1-\frac{N}{2}\leq i\leq\frac{N}{2}.
\label{usl21p}
\end{equation}
Also, 
\begin{align}
&\vec{n}_k\cdot\left(\cos\frac{(2i-1)\pi}{2N},\sin\frac{(2i-1)\pi}{2N},0\right)\nonumber\\
&= \cos\frac{(k-1)\pi}{N}\cos\frac{(2i-1)\pi}{2N}+\sin\frac{(k-1)\pi}{N}\sin\frac{(2i-1)\pi}{2N}\nonumber\\
&=\cos\frac{2(k-1)\pi-(2i-1)\pi}{2N}=\cos\frac{2(k-i)-1}{2N}\pi>0.
\end{align}
Again this will give 
\begin{equation}
p+1-\frac{N}{2}\leq i \leq p+\frac{N}{2}.
\label{usl2kp}
\end{equation}
The intersection of Eqs.~\eqref{usl21p} and \eqref{usl2kp} yields
\begin{equation}
p+1-\frac{N}{2}\leq i \leq \frac{N}{2}.
\label{p}
\end{equation}
The conditions from Eqs.~\eqref{n} and \eqref{p} can be combined into 
\begin{equation}
\left\lceil p+1-\frac{N}{2} \right\rceil\leq i \leq \left\lfloor \frac{N+1}{2} \right\rfloor,
\end{equation}
which completes the proof of Eq.~\eqref{id1eqn}.
\end{proof}

\begin{lemma}
\label{id2}
\upshape The following three identities hold:
\begin{align}
&\sum_{k=a}^{b}e^{i k \phi}=\frac{e^{ia\phi}-e^{i(b+1)\phi}}{1-e^{i\phi}},\nonumber\\
&\sum_{k=a}^{b}\sin(k\phi)=\frac{\sin\left(\frac{1-a+b}{2}\phi\right)\sin\left(\frac{a+b}{2}\phi\right)}{\sin\frac{\phi}{2}},\nonumber\\
&\sum_{k=a}^b\cos(k\phi)=\frac{\sin\left(\frac{1-a+b}{2}\phi\right)\cos\left(\frac{a+b}{2}\phi\right)}{\sin\frac{\phi}{2}}.
\end{align}
\end{lemma}

\begin{proof}
We will explicitly sum only the first equality and then obtain the other two as its real and imaginary part. Using the geometric series sum identity $\sum_{k=0}^n q^k=\frac{1-q^{n+1}}{1-q}$ we get:
\begin{align}
\sum_{k=a}^{b}e^{ik\phi}&=\sum_{k=0}^{b}e^{ik\phi}-\sum_{k=0}^{a-1}e^{ik\phi}=\frac{e^{ia\phi}-e^{i(b+1)\phi}}{1-e^{i\phi}}.
\end{align}
Now we find 
\begin{align}
&\sum_{k=a}^b\sin(k\phi)=\Imag\left(\frac{e^{ia\phi}-e^{i(b+1)\phi}}{1-e^{i\phi}}\right)\nonumber\\
&=\Imag\left(\frac{e^{ia\phi}-e^{i(b+1)\phi}-e^{i(a-1)\phi}+e^{ib\phi}}{4\sin^2\frac{\phi}{2}}\right)\nonumber\\
&=\frac{\sin a\phi+\sin b\phi -\Big(\sin(a-1)\phi+\sin(b+1)\phi\Big)}{4\sin^2\frac{\phi}{2}}\nonumber\\
&=\frac{2\sin\left(\frac{a+b}{2}\phi\right)\cos\left(\frac{a-b}{2}\phi\right)-2\sin\left(\frac{a+b}{2}\phi\right)\cos\left(\frac{a-b-2}{2}\phi\right)}{4\sin^2\frac{\phi}{2}}\nonumber\\
&=\frac{\sin\left(\frac{a+b}{2}\phi\right)\left(\cos\left(\frac{a-b}{2}\phi\right)-\cos\left(\frac{a-b-2}{2}\phi\right)\right)}{2\sin^2\frac{\phi}{2}}\nonumber\\&=\frac{\sin\left(\frac{a+b}{2}\phi\right)\times 2\sin\left(\frac{1-a+b}{2}\phi\right)\sin\frac{\phi}{2}}{2\sin^2\frac{\phi}{2}}\nonumber\\
&=\frac{\sin\left(\frac{1-a+b}{2}\phi\right)\sin\left(\frac{a+b}{2}\phi\right)}{\sin\frac{\phi}{2}},\textrm{ and}\\
&\sum_{k=a}^b\cos(k\phi)=\Real\left(\frac{e^{ia\phi}-e^{i(b+1)\phi}}{1-e^{i\phi}}\right)=\ldots=\nonumber\\
&=\frac{\sin\left(\frac{1-a+b}{2}\phi\right)\cos\left(\frac{a+b}{2}\phi\right)}{\sin\frac{\phi}{2}},
\end{align}
which completes the proof.
\end{proof}

\begin{corollary}\label{appAcorr}
\upshape The following identities hold:
\begin{align}
&\boxed{\sum_{i=p+1-\frac{N}{2}}^{\frac{N}{2}}\sin\left(\frac{2i-1}{2N}\pi\right)=\frac{\sin\frac{p\pi}{N}}{2\sin\frac{\pi}{2N}}},\nonumber\\
&\boxed{\sum_{i=p+1-\frac{N}{2}}^{\frac{N}{2}}\cos\left(\frac{2i-1}{2N}\pi\right)=\frac{\cos^2\frac{p\pi}{2N}}{\sin\frac{\pi}{2N}}},\nonumber\\
&\boxed{\sum_{i=p+\frac{3-N}{2}}^{\frac{N+1}{2}}\sin\left(\frac{i-1}{N}\pi\right)=\frac{\sin\frac{p\pi}{N}}{2\sin\frac{\pi}{2N}}},\nonumber\\
&\boxed{\sum_{i=p+\frac{3-N}{2}}^{\frac{N+1}{2}}\cos\left(\frac{i-1}{N}\pi\right)=\frac{\cos^2\frac{p\pi}{2N}}{\sin\frac{\pi}{2N}}}.
\label{corl2}
\end{align}
\end{corollary}
\begin{proof}
We prove one of the identities using Lemma \ref{id2} and the rest can be verified similarly:
\begin{align}
&\sum_{i=p+1-\frac{N}{2}}^{\frac{N}{2}}\sin\left(\frac{2i-1}{2N}\pi\right)=\sum_{i=p+1-\frac{N}{2}}^{\frac{N}{2}}\sin\left(i\frac{\pi}{N}-\frac{\pi}{2N}\right)\nonumber\\
&=\sum_{i=p+1-\frac{N}{2}}^{\frac{N}{2}}\left(\sin\left(i\frac{\pi}{N}\right)\cos\frac{\pi}{2N}-\cos\left(i\frac{\pi}{N}\right)\sin\left(\frac{\pi}{2N}\right)\right)\nonumber\\
&=\cos\left(\frac{\pi}{2N}\right)\sum_{i=p+1-\frac{N}{2}}^{\frac{N}{2}}\sin\left(i\frac{\pi}{N}\right)\nonumber\\
&\hphantom{=}-\sin\left(\frac{\pi}{2N}\right)\sum_{i=p+1-\frac{N}{2}}^{\frac{N}{2}}\cos\left(i\frac{\pi}{N}\right)\nonumber\\
&=\cos\left(\frac{\pi}{2N}\right)\frac{\sin\frac{(N-p)\pi}{2N}\sin\frac{(p+1)\pi}{2N}}{\sin\frac{\pi}{2N}}\nonumber\\
&\hphantom{=}-\sin\left(\frac{\pi}{2N}\right)\frac{\sin\frac{(N-p)\pi}{2N}\cos\frac{(p+1)\pi}{2N}}{\sin\frac{\pi}{2N}}\nonumber\\
&=\frac{\sin\frac{(N-p)\pi}{2N}}{\sin\frac{\pi}{2N}}\left(\sin\frac{(p+1)\pi}{2N}\cos\frac{\pi}{2N}-\cos\frac{(p+1)\pi}{2N}\sin\frac{\pi}{2N}\right)\nonumber\\
&=\frac{\sin\left(\frac{N-p}{2}\frac{\pi}{N}\right)}{\sin\frac{\pi}{2N}}\sin\frac{p\pi}{2N}=\frac{\sin\frac{p\pi}{2N}\cos\frac{p\pi}{2N}}{\sin\frac{\pi}{2N}}\nonumber\\
&=\frac{\sin\frac{p\pi}{N}}{2\sin\frac{\pi}{2N}}.
\end{align}
\end{proof}
\begin{lemma} We derive the bounds for index $i$. 
\begin{proof}
\label{id3}
Recalling that (cf.~Theorem \ref{gluslov}) the set $\Big\{\vec{e}(\vec{x})|\vec{x}\in\{\pm1\}^N\Big\}$ is equal to 
$$\Big\{\big(\cos\frac{(i-1)\pi}{N},\sin\frac{(i-1)\pi}{N},0\big)\Big|i=\overline{-N+1,N}\Big\}$$ for odd $N$ and 
$$\Big\{\big(\cos\frac{(2i-1)\pi}{2N},\sin\frac{(2i-1)\pi}{2N},0\big)\Big|i=\overline{-N+1,N}\Big\}$$ for even $N$ we obtain the following conditions on index $i$. 

For odd $N$ we have the following $4$ inequalities:
\begin{align}
\vec{n}_1\cdot\left(\cos\frac{(i-1)\pi}{N},\sin\frac{(i-1)\pi}{N},0\right)&>0,\\
\vec{n}_{k_p}\cdot\left(\cos\frac{(i-1)\pi}{N},\sin\frac{(i-1)\pi}{N},0\right)&>0,\\
\vec{n}_{k_{p+1}}\cdot\left(\cos\frac{(i-1)\pi}{N},\sin\frac{(i-1)\pi}{N},0\right)&<0,\\
\vec{n}_N\cdot\left(\cos\frac{(i-1)\pi}{N},\sin\frac{(i-1)\pi}{N},0\right)&<0.
\end{align}

The first of them becomes
\begin{align}
\vec{n}_1\cdot\left(\cos\frac{(i-1)\pi}{N},\sin\frac{(i-1)\pi}{N}\right)= \cos\frac{(i-1)\pi}{N}>0
\end{align}
yielding 
\begin{equation}
\frac{(i-1)\pi}{N}\in\left(-\frac{\pi}{2},\frac{\pi}{2}\right)\Rightarrow 1-\frac{N}{2}<i<1+\frac{N}{2}. \label{a1}
\end{equation}
The second one becomes
\begin{equation}
\vec{n}_{k_p}\cdot\left(\cos\frac{(i-1)\pi}{N},\sin\frac{(i-1)\pi}{N}\right)=\cos\frac{k_p-i}{N}\pi>0,
\end{equation}
yielding 
\begin{equation}
\frac{k_p-i}{N}\pi\in\left(-\frac{\pi}{2},\frac{\pi}{2}\right)\Rightarrow k_p-\frac{N}{2}<i<k_p+\frac{N}{2}.\label{a2}
\end{equation}
The intersection of \eqref{a1} and \eqref{a2} gives us 
\begin{equation}
k_p-\frac{N}{2}<i<\frac{N}{2}+1.\label{a}
\end{equation}
The third inequality gives us 
\begin{equation}
\vec{n}_{k_{p+1}}\cdot\left(\cos\frac{(i-1)\pi}{N},\sin\frac{(i-1)\pi}{N}\right)=\cos\frac{k_{p+1}-i}{N}\pi<0,
\end{equation}
yielding 
\begin{equation}
\frac{|i-k_{p+1}|}{N}\pi\in\left(\frac{\pi}{2},\frac{3\pi}{2}\right),
\end{equation}
so that 
\begin{equation}
\frac{i-k_{p+1}}{N}\pi\in\left(-\frac{3\pi}{2},-\frac{\pi}{2}\right)\textrm{ or }\frac{i-k_{p+1}}{N}\pi\in\left(\frac{\pi}{2},\frac{3\pi}{2}\right).
\end{equation}
Thus, for $i$ we have
\begin{align}
&k_{p+1}-\frac{3N}{2}<i<k_{p+1}-\frac{N}{2}, \textrm{ or }\nonumber\\
&k_{p+1}+\frac{N}{2}<i<k_{p+1}+\frac{3N}{2}.
\end{align}
Noting that $-N+1\leq i\leq N$ and $1\leq k_{p+1}\leq N$, the condition $k_{p+1}-\frac{3N}{2}<i<k_{p+1}-\frac{N}{2}$ covers the possible values of $i$ allowed in either situation, whether $i<k_{p+1}$ or $i>k_{p+1}$. Hence, we have

\begin{equation}
\label{b1}
k_{p+1}-\frac{3N}{2}<i<k_{p+1}-\frac{N}{2}.
\end{equation}

Finally, similar to the previous case, the fourth inequality gives us
\begin{equation}
\vec{n}_N\cdot\left(\cos\frac{(i-1)\pi}{N},\sin\frac{(i-1)\pi}{N}\right)=\cos\frac{(i-N)\pi}{N}<0,
\end{equation}
yielding
\begin{equation}
\frac{i-N}{N}\pi\in\left(-\frac{3\pi}{2},-\frac{\pi}{2}\right),
\end{equation}
so $i$ has to satisfy
\begin{equation}
-\frac{N}{2}<i<\frac{N}{2}.
\label{b2}
\end{equation}
Taking the intersection between \eqref{b1} and \eqref{b2} we obtain
\begin{equation}
-\frac{N}{2}<i<k_{p+1}-\frac{N}{2}\label{b}
\end{equation}
Finally taking the intersection between \eqref{a} and \eqref{b} we have
\begin{equation}
k_p-\frac{N}{2}<i<k_{p+1}-\frac{N}{2},
\end{equation}
so that the bounds for summation over $i$ are given by
\begin{equation}
\boxed{k_p-\frac{N-1}{2}\leq i \leq k_{p+1}-\frac{N+1}{2}}
\end{equation}
for odd $N$.

For even $N$, the starting inequalities are
\begin{align}
&\vec{n}_1\cdot\left(\cos\frac{(2i-1)\pi}{2N},\sin\frac{(2i-1)\pi}{2N},0\right)>0,\nonumber\\
&\vec{n}_{k_p}\cdot\left(\cos\frac{(2i-1)\pi}{2N},\sin\frac{(2i-1)\pi}{2N},0\right)>0,\\
&\vec{n}_{k_{p+1}}\cdot\left(\cos\frac{(2i-1)\pi}{2N},\sin\frac{(2i-1)\pi}{2N},0\right)<0,\nonumber\\
&\vec{n}_N\cdot\left(\cos\frac{(2i-1)\pi}{2N},\sin\frac{(2i-1)\pi}{2N},0\right)<0,
\end{align}
which, in a similar succession of steps as in the case of odd $N$, gives us
\begin{equation}
\boxed{k_p-\frac{N}{2}\leq i \leq k_{p+1}-\frac{N+2}{2}.}
\end{equation}
\end{proof}
\end{lemma}
\newpage
\begin{lemma}We derive the marginalization identity \eqref{margidth1} in the proof to Theorem~\ref{counbiqu}.\label{appAmarg}
\begin{proof}
\begin{widetext}
\begin{align}
\sum_{\vec{y}\in\{\pm\}^N}^{y_k=\pm1}G^s(\vec{y})&=G^s(\pm1,\ldots,y_k,\ldots,\pm1)+\sum_{k\leq p<N}^{y_k=\pm1}G^{s}(\underbrace{\pm1,\ldots,\pm1}_{p\text{ `$\pm1$'s }},\underbrace{\mp1,\ldots,\mp1}_{N-p\text{ `$\mp1$'s }})\nonumber\\
&+\sum_{1\leq p<k}^{y_k=\pm1}G^{s}(\underbrace{\mp1,\ldots,\mp1}_{p\text{ `$\mp1$'s }},\underbrace{\pm1,\ldots,\pm1}_{N-p\text{ `$\pm1$'s }})\nonumber\\
&=\frac{1}{2}\left(1-\eta\sum_{p=1}^{N-1}\sin\frac{\alpha_p-\alpha_{p-1}}{2}\right)I+\frac{1}{2}\eta\sum_{p=k}^{N-1}\sin\frac{\alpha_p-\alpha_{p-1}}{2}I+\frac{1}{2}\eta\sum_{p=1}^{k-1}\sin\frac{\alpha_p-\alpha_{p-1}}{2}I\nonumber\\
&\pm\frac{1}{2}\eta\cos\frac{\alpha_{N-1}}{2}\vec{s}\cdot\vec{\sigma}\pm\frac{1}{2}\eta\sum_{p=k}^{N-1}\sin\frac{\alpha_{p}-\alpha_{p-1}}{2}\vec{t}_p\cdot\vec{\sigma}\mp\frac{1}{2}\eta\sum_{p=1}^{k-1}\sin\frac{\alpha_{p}-\alpha_{p-1}}{2}\vec{t}_p\cdot\vec{\sigma}\nonumber\\
&=\frac{1}{2}I\pm\frac{1}{2}\eta\left(\cos\frac{\alpha_{N-1}}{2}\vec{s}+\sum_{p=k}^{N-1}\sin\frac{\alpha_p-\alpha_{p-1}}{2}\vec{t}_p-\sum_{p=1}^{k-1}\sin\frac{\alpha_p-\alpha_{p-1}}{2}\vec{t}_p\right)\cdot\vec{\sigma}=\frac{1}{2}I\pm\frac{1}{2}\eta\vec{g}_k\cdot\vec{\sigma}.
\label{intermed}
\end{align}
\end{widetext}
The geometric part becomes:
\begin{widetext}
\begin{align}
\vec{g}_k&=\cos\frac{\alpha_{N-1}}{2}\frac{\vec{n}_1+\vec{n}_N}{||\vec{n}_1+\vec{n}_N||}+\sum_{p=k}^{N-1}\sin\frac{\alpha_p-\alpha_{p-1}}{2}\frac{(\vec{n}_{p+1}+\vec{n}_p)\times\vec{e}_z}{||\vec{n}_{p+1}+\vec{n}_p||}-\sum_{p=1}^{k-1}\sin\frac{\alpha_p-\alpha_{p-1}}{2}\frac{(\vec{n}_{p+1}+\vec{n}_p)\times\vec{e}_z}{||\vec{n}_{p+1}+\vec{n}_p||}\nonumber\\
&=\cos\frac{\alpha_{N-1}}{2}\frac{\vec{n}_1+\vec{n}_N}{2\cos\frac{\alpha_{N-1}}{2}}+\sum_{p=k}^{N-1}\sin\frac{\alpha_p-\alpha_{p-1}}{2}\frac{(\vec{n}_{p+1}+\vec{n}_p)\times\vec{e}_z}{2\cos\frac{\alpha_{p}-\alpha_{p-1}}{2}}-\sum_{p=1}^{k-1}\sin\frac{\alpha_p-\alpha_{p-1}}{2}\frac{(\vec{n}_{p+1}+\vec{n}_p)\times\vec{e}_z}{2\cos\frac{\alpha_{p}-\alpha_{p-1}}{2}}\nonumber\\
&=\frac{1}{2}(\vec{n}_1+\vec{n}_N)+\frac{1}{2}\sum_{p=k}^{N-1}\tan\frac{\alpha_p-\alpha_{p-1}}{2}\cot\frac{\alpha_p-\alpha_{p-1}}{2}(\vec{n}_p-\vec{n}_{p+1})-\frac{1}{2}\sum_{p=1}^{k-1}\tan\frac{\alpha_p-\alpha_{p-1}}{2}\cot\frac{\alpha_p-\alpha_{p-1}}{2}(\vec{n}_p-\vec{n}_{p+1})\nonumber\\
&=\frac{1}{2}\left(\vec{n}_1+\vec{n}_N-\sum_{p=k}^{N-1}(\vec{n}_{p+1}-\vec{n}_p)+\sum_{p=1}^{k-1}(\vec{n}_{p+1}-\vec{n}_p)\right),
\label{geom}
\end{align}
\end{widetext}
where $\vec{e}_z=(0,0,1)$ and we have used the fact that
$$
(\vec{n}_{p+1}+\vec{n}_p)\times\vec{e}_z=\cot\frac{\alpha_{p}-\alpha_{p-1}}{2}(\vec{n}_p-\vec{n}_{p+1}),
$$
which can easily be derived by linear decomposition $(\vec{n}_{p+1}+\vec{n}_p)\times\vec{e}_z=a\vec{n}_p+b\vec{n}_{p+1}$, and finding $a$ and $b$. There are three distinct cases: $k=1$, $k=N$ and $1<k<N$, but in all of them by trivial algebraic manipulations of Eq.~\eqref{geom} we obtain
\begin{equation}
\vec{g}_k=\vec{n}_k
\end{equation}
and this in combination with Eq.~\eqref{intermed} yields
\begin{equation}
\sum_{\vec{y}\in\{\pm\}^N}^{y_k=\pm1}G^s(\vec{y})=\frac{1}{2}\left(I\pm\eta\vec{n}_k\cdot\vec{\sigma}\right)=E_k(\pm1).
\end{equation}
\end{proof}
\end{lemma}
\end{document}